\documentclass[12pt,a4paper]{article}

\usepackage[a4paper, margin=1in]{geometry}

\usepackage{setspace}
\usepackage[T1]{fontenc}
\usepackage[utf8]{inputenc}
\usepackage{float}
\usepackage{authblk}
\usepackage{graphicx}
\usepackage{subcaption} 
\usepackage[thinlines]{easytable}
\usepackage{hyperref}
\usepackage{wrapfig}
\usepackage{dingbat}
\usepackage{rotating}  % in the preamble
\usepackage{color}
\usepackage{refcount}
\usepackage[table]{xcolor}
\usepackage[toc,page]{appendix}
\usepackage{bbm}
\usepackage{enumitem}
\usepackage{amsmath,amssymb,amsthm}
\usepackage{mathtools}
\usepackage[english]{babel}
\usepackage{natbib}
\usepackage{bibunits}
\usepackage{xcolor}
\usepackage{booktabs}   % for \toprule, \midrule, \bottomrule
\usepackage{multirow}   % for \multirow (grouping cells)
\usepackage{caption}    % for better caption formatting (optional)
\usepackage{wrapfig}  % add this in the preamble
\usepackage{setspace}   % for line spacing
\usepackage{etoolbox}   % for patching environments
\AtBeginEnvironment{proof}{\begin{spacing}{1.5}}
\AtEndEnvironment{proof}{\end{spacing}}
\AtBeginEnvironment{lem}{\begin{spacing}{1.5}}
\AtEndEnvironment{lem}{\end{spacing}}
\AtBeginEnvironment{rem}{\begin{spacing}{1.5}}
\AtEndEnvironment{rem}{\end{spacing}}

\theoremstyle{definition}
\newtheorem{theo}{Theorem}[section]
\newtheorem{lem}[theo]{Lemma}
\newtheorem{prop}[theo]{Proposition}

\newtheorem{rem}[theo]{Remark}

\numberwithin{equation}{section}

\newcommand{\cond}{\overset{d}{\to}}
\newcommand{\bb}{b}
\newcommand{\cc}{C}
\newcommand{\R}{\mathbb{R}}
\newcommand{\lb}{\left(}
\newcommand{\rb}{\right)}

\newcommand{\E}{\mathbb{E}}
\newcommand{\N}{\mathbb{N}}
\newcommand{\PR}{\mathbb{P}}
\newcommand{\conp}{\stackrel{P}{\to}}
\usepackage{algorithm}
\usepackage{algpseudocode}
\newcommand{\Id}{\operatorname{Id}}
\newcommand{\MP}{Mar\v cenko--Pastur }
\begin{document}
		
\defaultbibliography{Literature.bib} 
\defaultbibliographystyle{chicago}
\title{
Does PCA Work for Rough Functional Data?
%Time Series
}

%\author{Tim Kutta \and Nina Dörnemann \and Piotr Kokoszka}

 \author[1]{Tim Kutta\thanks{Corresponding author; Mailing address: Department of Mathematics, Aarhus University,
Ny Munkegade 118,
8000 Aarhus C,
Denmark; Email: \texttt{tim.kutta@math.au.dk}.}}
 \author[1]{Nina Dörnemann}
\author[2]{Piotr Kokoszka}
  \affil[1]{Department of Mathematics, Aarhus University, Denmark}
 \affil[2]{Department of Statistics, Colorado State University, USA}

\date{}
\maketitle

\markboth{PCA for Rough Functional Data}{PCA for Rough Functional Data}

\begin{abstract}
Functional data analysis is concerned with the analysis of infinite-dimensional data functions. Functional principal component analysis (FPCA) is a key method to obtain finite-dimensional summaries. Consistency of FPCA has been theoretically established for sufficiently regular data functions. However,  empirical evidence shows that FPCA can become severely inconsistent when the underlying functions are too rough. This paper provides the first theoretical explanation for this phenomenon. We propose a model that explicitly captures the roughness of functional data and allows us to quantify the resulting bias of FPCA, depending on the functional roughness. The model undergoes a phase transition marking the point at which FPCA becomes entirely uninformative. Based on these probabilistic results, we discuss diagnostic tests for  informative principal components. As an additional contribution, we derive results on spectral statistics that may serve as a foundation for goodness-of-fit tests for rough functional data. Mathematically, our approach combines recent advances in random matrix theory and generic chaining with tools from FDA. We illustrate the effects of roughness on FPCA using simulations, as well as climate and environmental datasets.
\end{abstract}

\noindent%

\section{Introduction}
\begin{bibunit}
Functional data analysis (FDA) is an area of modern statistics where  data consist of entire random curves or surfaces. Such data naturally arise in finance (a daily stock price curve), hydrology (a discharge curve), meteorology (an annual temperature curve), as well as many other fields. Examples of data curves are displayed in  Figures~\ref{fig:1} and \ref{fig:2}. For introductions to FDA, we refer the reader to the monographs of \cite{bosq:2000,ramsay:silverman:2005,HKbook,hsing:eubank:2015}.
Analyzing an entire function, i.e.,\ the full trajectory of a variable, is often more informative than focusing on any single time point. However, FDA also comes with distinct challenges since the data functions live, mathematically, in an infinite-dimensional space. This means that statistical procedures always need to be based on some finite-dimensional approximations that capture the behavior of the entire functions. 
For such approximations, and more generally to make FDA amenable to the tools of multivariate statistics, FDA relies critically on  functional principal component analysis (FPCA). It is hard to overestimate the importance of FPCA for functional data, and FPCA has early on become "the
most prevalent tool in FDA" \citep{wang:chiou:mueller:2016}. Accordingly, the reliability of statistical analyses depends critically on the stable estimation of principal components.
Consistency of FPCA has been theoretically justified for random objects in Hilbert spaces and we refer to Chapter 3 in \cite{HKbook} or Chapter 9 in \cite{hsing:eubank:2015} for reviews of key mathematical results. Current theory guarantees consistency of the empirical eigensystem, as long as the data functions are sufficiently "smooth". Here, smoothness does not directly refer to continuity or differentiability of a function (even though this is an appropriate intuition), but rather to the size and decay of the eigenvalues of the covariance operator. Existing theory is developed for functions that are "smooth" in the sense that their eigenvalues are rapidly decaying and accordingly, data functions concentrate closely around a very low-dimensional vector subspace. Unfortunately, real data functions
are often quite rough and not that nicely concentrated. This can seriously
impinge on the quality of FPCA in practice, and thereby harm the entire downstream analysis in applications.

\noindent \textbf{Example 1}  
To illustrate the problems of FPCA with rough functional data, we study the annual temperature profiles displayed in Figure \ref{fig:1}. The data has been gathered over  $235$ years at the German meteorological station in Hohenpeissenberg and is publicly available (\cite{dwd:climate:archive}).  
Temperature curves have been widely analyzed in both methodological and empirical FDA studies, and for some examples we refer to \cite{fremdt:horvath:kokoszka:steinebach:2014,aue:rice:sonmez:2018,hadjipantelis:muller:2018:book,dette:kokot:volgushev:2020,shah:dewolf:paul:madden:2024}. This makes them a natural benchmark for evaluating FPCA performance. 
Our aim is to investigate the stability of FPCA results. To this end, we randomly select, without replacement, two sets of temperature curves of size $N$. For each of these samples, we 
compute the empirical eigenfunctions of order $k=1,2,3$ and compare the results. Since we draw at random from an underlying, fairly homogenous population, it is reasonable to expect that estimates for the $k$th order eigenfunctions from the two samples should be nearly identical. Geometrically speaking, this means that the angle between the estimated $k$th eigenfunctions should be close to $0$°. Larger angles would be expected if FPCA were unstable. We repeat this procedure $1000$ times for the first three eigenfunctions and for sample sizes ranging from $N=10$ to $N=110$, i.e., small to moderately large.
Average angles are displayed in the right panel of Figure \ref{fig:1}. They show that even for $N=110$, the first eigenfunctions have a large average angle, indicating a high degree of instability. Angles for second and third eigenfunctions are close to $90$°, implying that estimates are effectively no better than random guesses. \\
\begin{figure}[H]
  \centering
  \begin{subfigure}{0.49\textwidth}
    \centering
    \includegraphics[width=\linewidth]{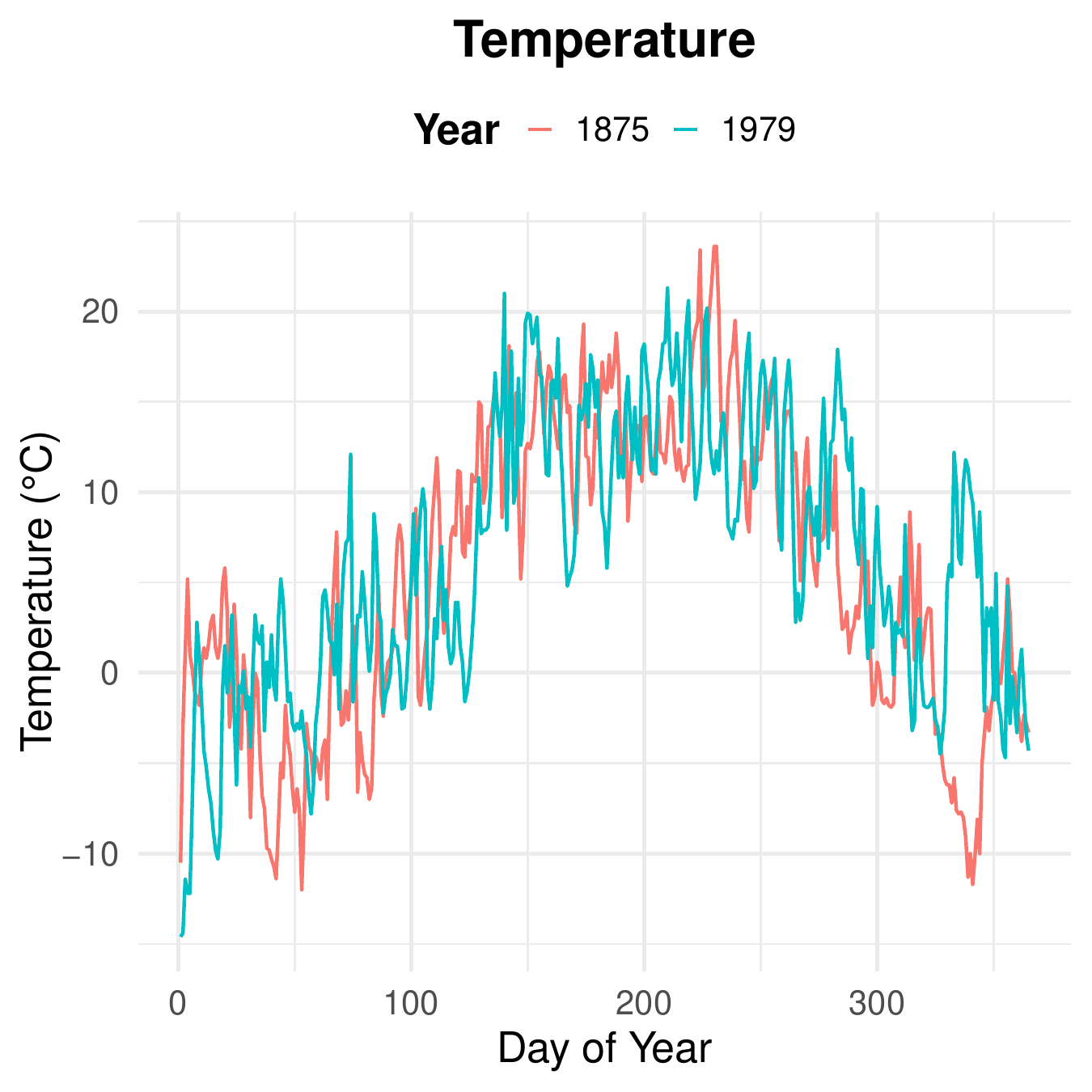} %9x9 pdf
  \end{subfigure}
  \hfill
  \begin{subfigure}{0.49\textwidth}
    \centering
    \includegraphics[width=\linewidth]{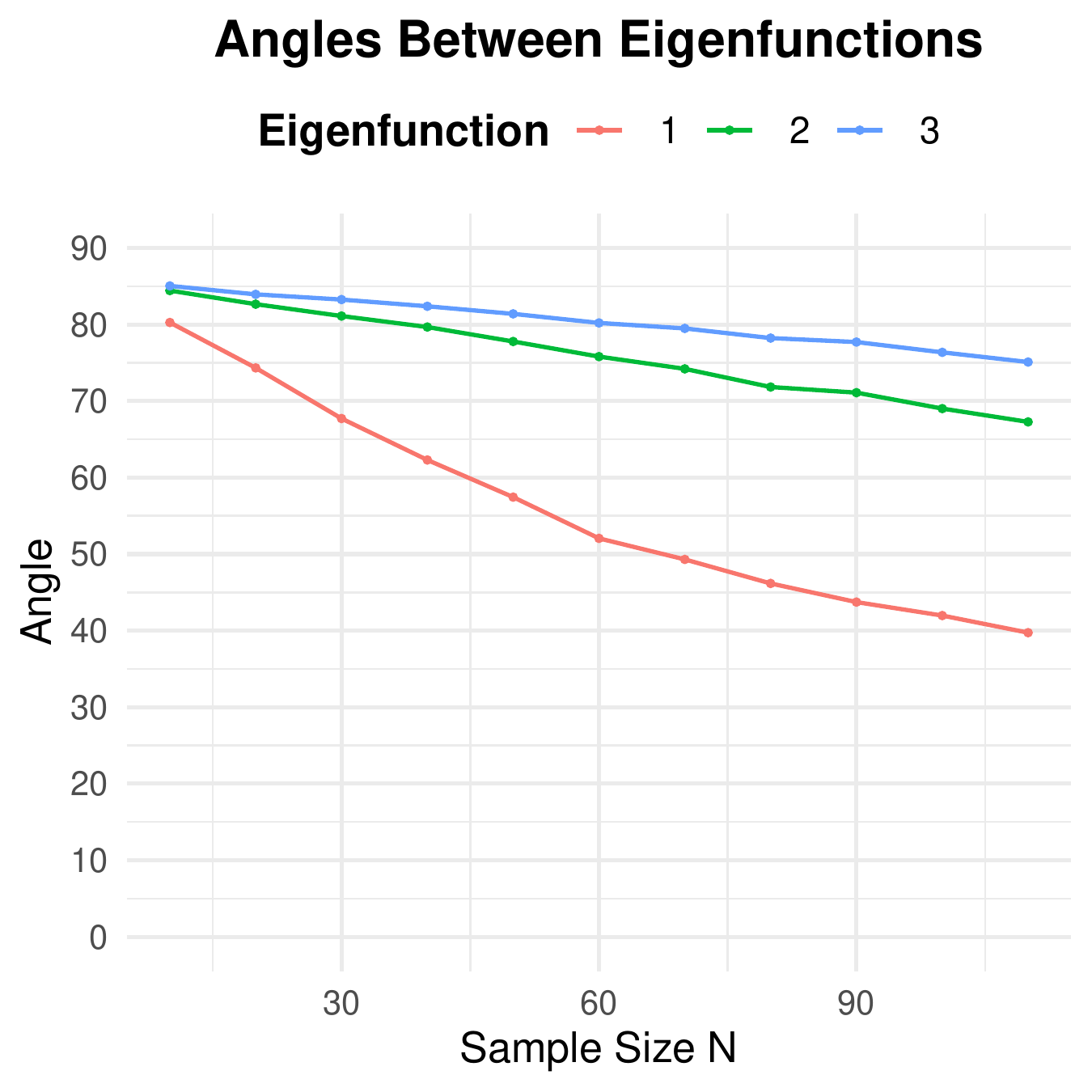}
  \end{subfigure}
  \caption{\label{fig:1} Left: Selected temperature profiles from Hohenpeissenberg. Right:  Average angles between estimated $k$th principal components, $k = 1, 2, 3$, computed from
two random samples of $N$ temperature curves. }
\end{figure}
\noindent \textbf{Example 2} To indicate that the above phenomenon is not unique to temperature curves, we consider a second example. We analyze daily streamflow (discharge) measurements of the Ninnescah River at Peck, Kansas, recorded by the United States Geological Survey (USGS) at site number 07145500. The dataset begins in January 1939 and covers 86 hydrological years. Each observation corresponds to the mean daily discharge in cubic meters per second. For our analysis, each hydrological year is treated as a single curve, starting on November 1 and ending on October 31 of the following year. Once more, we drew at random two datasets of size $N$ from the underlying population and compared the angles between the leading eigenfunctions. Since we only had $86$ curves at our disposal, we used drawing with replacement this time. The picture that emerges is similar to, but even more extreme than, for the temperature data. Average angles between the first eigenfunctions are $>50^\circ$ even for $N=100$. Second and third eigenfunctions respectively are again almost orthogonal. 

\begin{figure}[h]
  \centering
  \begin{subfigure}{0.49\textwidth}
    \centering
    \includegraphics[width=\linewidth]{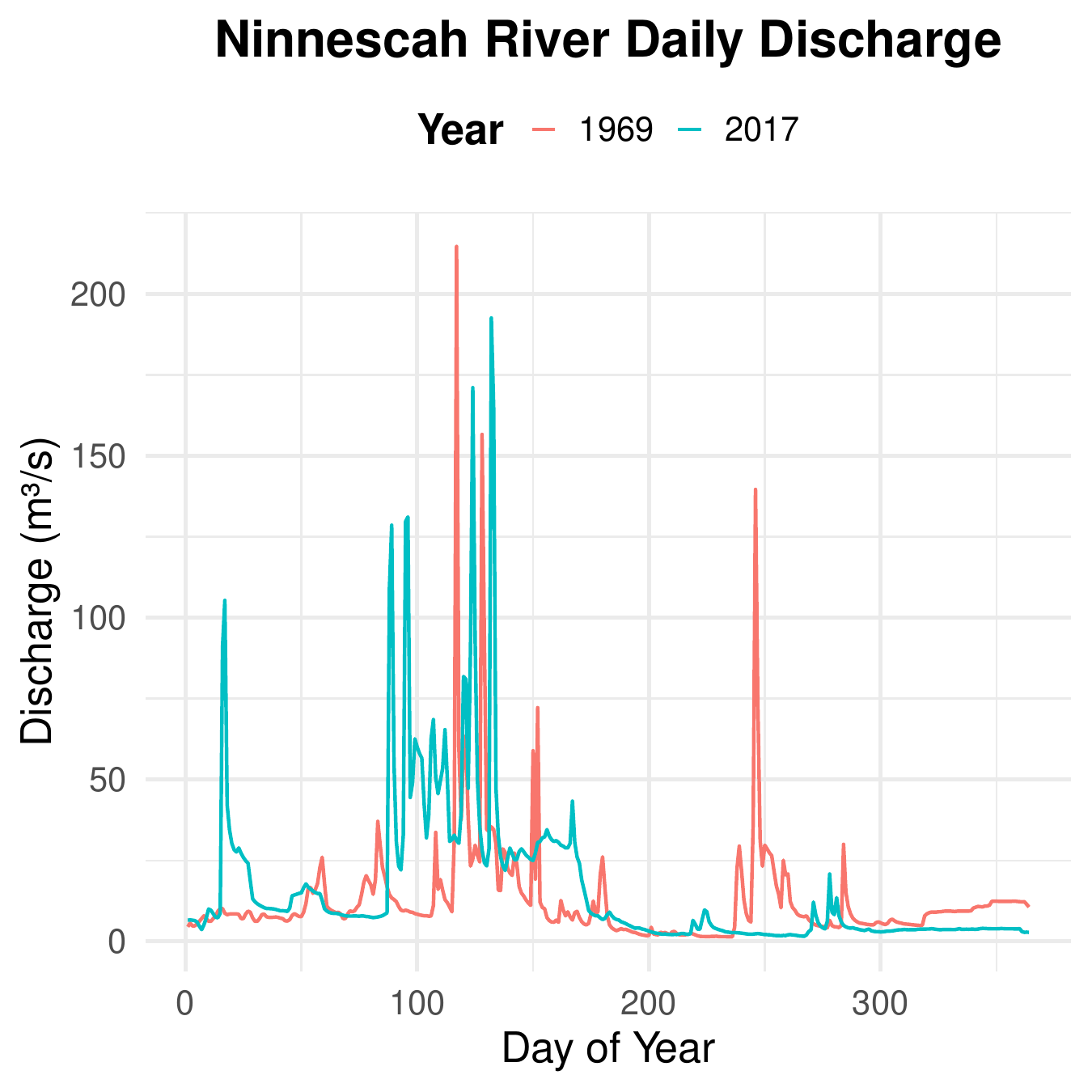}
  \end{subfigure}
  \hfill
  \begin{subfigure}{0.49\textwidth}
    \centering
    \includegraphics[width=\linewidth]{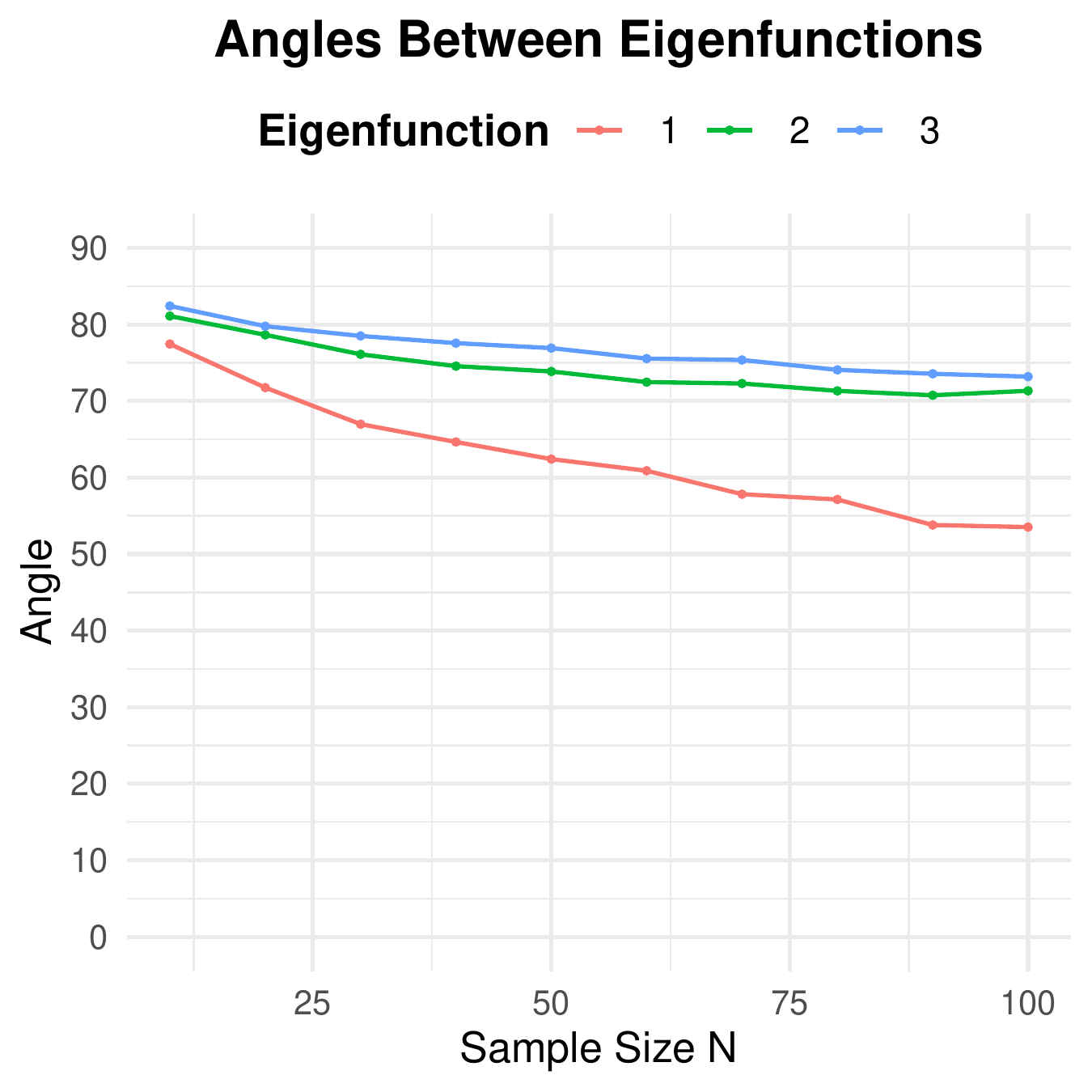}
  \end{subfigure}
  \caption{\label{fig:2} Left: Selected discharge profiles.  Right:  Average angles between estimated $k$th principal components, $k = 1, 2, 3$, computed from
two random samples of $N$ discharge curves. }
\end{figure}

\noindent \textbf{Interpretation} The above examples illustrate that FPCA can suffer extreme breakdowns for rough functional data. In such cases, any downstream statistical analysis that uses the estimated principal components runs a high risk of being affected. To the best of our knowledge, these problems of FPCA have not been systematically studied. 
A reason for this is the reliance of current FDA methodology on probability theory in Hilbert and Banach spaces initiated by \cite{kuelbs:1973,saintflour:probabilities:1976,dehling:1983} and summarized in the influential monograph of \cite{bosq:2000}. While mathematically general, this theory practically covers only fairly regular functions -- exactly those for which FPCA is guaranteed to work. 
%So 
These standard models cannot even in principle describe the breakdown of FPCA that we observe empirically.  The aim of this paper is to fill this gap and describe the effect, both empirically and mathematically, that rough data has on FPCA.

\noindent \textbf{Main contributions}
We introduce a new model that explicitly accounts for functional roughness. \textit{Within this model, we identify a critical phase transition: if the data functions are too rough, empirical principal components and eigenvalues are asymptotically unrelated to their population counterparts. For moderate roughness, a relation exists, but empirical versions remain inconsistent, and we derive explicit asymptotic formulas for the estimation error.}
 Because it is non-trivial in practice to distinguish informative from uninformative components, we develop statistical inference tools that enable such a distinction.  As a byproduct of our theory, we also derive a law of large numbers for spectral statistics of functional data, potentially useful for future goodness-of-fit tests. 
Our analysis combines ideas from FDA, random matrix theory (RMT), and generic chaining. To our knowledge, this is the first work that combines these distinct research fields. 
In an empirical analysis, we study the effects of roughness on the geometric shape of eigenfunctions, parameter projections and hypothesis testing. Summarizing, we find that for rough data, only a few eigenfunctions have a meaningful geometric form, while the rest look similar to a white noise. Projecting smooth parameter functions on these eigenfunctions leads to imprecise and geometrically bizarre approximations. The effect on hypothesis testing is more benign. Our analysis indicates, both theoretically and empirically, that basic hypothesis tests become more conservative for rough functional data -- in particular, they still seem to hold their significance level.

\noindent \textbf{Organization of the paper} The main theoretical results of this work on FPCA are collected in Section \ref{sec:main}. Section \ref{sec:simulations} contains Monte Carlo studies that empirically validate our findings and Section \ref{sec:data} provides a more detailed analysis of the data samples discussed in the Introduction. 
The Supplement consists of three sections. Section \ref{sec_proofs}
is dedicated to the proofs of our main results. In Section \ref{sec:spec}, we introduce some basic spectral statistics for functional data and prove a new law of large numbers for them. 
Section \ref{sec_background_rmt} contains a review of some essential results from RMT. Additional numerical results are presented in Section \ref{sec_supp_numeric}.

\section{Main results} \label{sec:main}

After providing some basics on random functions in Section \ref{sec:basics}, we define our model for rough functional data in Section \ref{sec:model}. The main probabilistic results are given in Section \ref{sec:phaset}. Diagnostic tests for informative principal components are presented in Section \ref{sec:diag}. 

\subsection{Definitions and notations} \label{sec:basics}

We consider the space $L^2([0,1])$, which consists of all measurable functions $g:[0,1] \to \mathbb{R}$ that satisfy $\int_0^1 g^2(x) dx <\infty$. Identifying functions equal a.e., the space $L^2([0,1])$ becomes a Hilbert space, when equipped with the inner product
\[
\langle g,h\rangle := \int_0^1 g(x) h(x) dx.
\]
The corresponding norm is given by $\|g\|^2:=\langle g,g \rangle$. 
A random function in $L^2([0,1])$ is a measurable map $X: \Omega \to L^2([0,1])$, defined on a probability space $(\Omega, \mathcal{A}, \mathbb{P})$. The mean $\mu \in L^2([0,1])$ of $X$ is well defined if $\mathbb{E}\|X\|<\infty$, and is characterized by the identity
\[
\langle \mu, g\rangle = \mathbb{E}[\langle X, g\rangle], \qquad \forall\, g \in L^2([0,1]).
\]
Similarly, if $\mathbb{E}\|X\|^2<\infty$ and $\mathbb{E}X=0$, the covariance operator $C$  of $X$ is well defined and satisfies 
\begin{align} \label{e:C}
\langle C g, h\rangle = \mathbb{E}[\langle X, g\rangle\langle X,h\rangle], \qquad \forall \, g,h \in L^2([0,1]).
\end{align}
The covariance operator $C$ is a trace-class operator and can be represented using a non-increasing sequence of non-negative eigenvalues $(\lambda_n)_n$ and an orthonormal basis of eigenfunctions $(e_n)_n$ as follows
\[
C = \sum_{n \ge 1} \lambda_n (e_n \otimes e_n),
\]
where $\otimes$ refers to the standard outer product for Hilbert spaces. Throughout this work, we denote the operator norm by $\|\cdot\|_\mathcal{L}$, and we recall that $\|C\|_\mathcal{L} = \lambda_1$.

\subsection{A model for rough functional data} \label{sec:model}

We assume that the sample $X_1, \ldots X_N$ consists of i.i.d., centered, Gaussian random variables in $L^2([0,1])$ with covariance operator $C=C_N$ defined by \eqref{e:C}, which determines the distribution of the sample. We usually do not make the dependence of $C$
 and its eigensystem on $N$ explicit for the sake of lighter notation. Unless indicated otherwise, asymptotic results are formulated for $N\to\infty.$
Now, let $(e_n)_n$ be an orthonormal basis of  $L^2([0,1])$,  $\bb:[0,\infty) \to [0,\infty)$ be a non-increasing function, $K\in \N$ be a positive integer and $s_1> \cdots > s_K > \bb(0)$ be positive numbers.
We define
\begin{align} \label{e:sigrep}
C:= \sum_{n \ge 1}\lambda_{n} (e_n \otimes e_n), \quad \textnormal{where}\quad 
\lambda_n = 
\begin{cases}
    s_n, \quad n=1,\ldots,K, \\ \bb\big(\frac{n-K}{N}\big), \,\, n>K.
\end{cases}
\end{align}
In other words, $C$ has eigenfunctions $e_n$ and eigenvalues $\lambda_{n}$. The eigenvalues $\lambda_n, n=1,\ldots,K$, are denoted by $s_n$ because they are \textit{spiked}, and the eigenvalues $\lambda_n, n>K$, are characterized by the function $\bb$, because they are part of the \textit{bulk} of eigenvalues. Accordingly, we sometimes refer to $\bb$ as the \textit{bulk function}. For the bulk to be non-trivial $\bb$ cannot be a.e. $0$, which can be prevented by assuming $b(x)>0$ for some $x>0$. In this work, we will always assume the slightly stronger condition that $b(0)>0$ and that $b$ is Lipschitz continuous at $0$.
Now, for  $C $ to be a trace-class operator, we assume that $\int_0^\infty \bb(x) dx <\infty$, which implies that
\begin{align} \label{e:TC}
    \|C\|_\mathcal{L}=s_1, \qquad \operatorname{Tr}[C] \sim N \int_0^\infty \bb(x) dx,
\end{align}
where $\rm Tr[\cdot]$ denotes the trace.
Finally, we define the empirical covariance
\begin{align} \label{e:empcov}
\widehat{C} = \frac{1}{N}\sum_{i=1}^N X_i \otimes X_i,
\end{align}
and denote its eigenfunctions and eigenvalues, respectively, by $(\hat e_n)_n, (\hat \lambda_n)_n$.
\begin{rem} \label{rem:1}
    We discuss some details of our model.
\begin{itemize}
    \item[1)] A standard complexity measure for linear (covariance) operators is the effective rank; see, e.g., \cite{koltchinskii:lounici:2017b}. It is defined for a trace class operator $A$ as 
    \[
    \mathfrak{r}[A]:= \frac{\operatorname{Tr}[A]}{\|A\|_\mathcal{L}},
    \]
    and we thus conclude that \eqref{e:TC} implies $\mathfrak{r}[C] \sim N$ in our model. It is well-known that under $\mathfrak{r}[C]=o(N)$, consistency results for the empirical covariance operator and its eigensystem hold, see among many others \cite{koltchinskii:lounici:2017b,alghattas:chen:sanzalonso:2025}. Yet the case $\mathfrak{r}[C]\sim N$, that is the focus of this work, remains much less explored. An exception to this is when the (finite-dimensional) space itself grows at the same rate as  $ N$, which is the setting of classical RMT. In the asymptotic framework where the dimension and the sample diverge proportionally, it is well-known that the largest eigenvalue undergoes a phase transition, the so-called BBP phase transition due to the seminal work of \cite{baik:ben:peche:2005}. Roughly speaking, this phenomenon describes the different types and orders of fluctuations of the leading eigenvalue of the sample covariance matrix in the subcritical and supercritical case. 
  As an important statistical consequence, in the subcritical regime, the largest sample eigenvalue and its eigenvector asymptotically carry no information about their population counterparts, whereas in the supercritical regime they do.
  %  In the toy model $\Sigma = (\alpha, 1, \ldots, 1)$, the critical threshold is $1+\sqrt{p/N}$, that is, $\alpha < 1+\sqrt{p/N}$ corresponds to subcritical and $\alpha > 1+\sqrt{p/N}$ to the supercritical regime.
    A review of some related results in random matrix theory can be found in Section \ref{sec_background_rmt} of the supplement. 
    
    \item[2)] In RMT, it is uncommon to specify the eigenvalues directly, as in our model \eqref{e:sigrep}. A more typical approach is, for a $p$-dimensional covariance, to determine an eigenvalue distribution, using the  spectral measure 
\begin{align} \label{e:spectral}
\mathfrak{s}_{N,p}:=\frac{1}{p} \sum_{i=1}^p \delta_{\lambda_i},
\end{align}
where $\delta_{x}$ is the Dirac measure at $x$. Notice that $\mathfrak{s}_{N,p}$ depends on $N$, because the eigenvalues may also depend on $N$.
RMT then usually requires that $\mathfrak{s}_{N,p}$ converge weakly to some limiting measure $\mathfrak{s}$ as $N \to \infty$ and $p=p(N) \to \infty$, where $p/N$ has a limit in $(0,\infty).$ Since spectral measures can only be properly defined for a finite dimension $p$, we have to pursue a different approach for functional data. 
\item[3)] The data in our model follow a Gaussian distribution on the space $L^2([0,1])$. Extensions to non-Gaussian data are likely possible but non-trivial. There are two main technical reasons why the Gaussian assumption is made in our derivations:
First, we exploit recent results from generic chaining, which are best developed for Gaussian data. Some progress has been made in the investigation of subgaussian vectors suggesting that extensions of our results might be possible for that regime (\cite{koltchinskii:lounici:2017b,alghattas:chen:sanzalonso:2025}). A second reason is that in our proofs we use the property that the projection of our data $\Pi_\kappa[X_i]$ on any finite-dimensional subspace spanned by the eigenfunctions $\{e_1,\ldots,e_\kappa\}$ is \textit{separable}. This means that $\Pi_\kappa[X_i] \overset{d}{=} C_{\Pi_\kappa}^{1/2} Z$, where $Z$ has i.i.d. components and $C_{\Pi_\kappa}$ is the covariance matrix of $\Pi_\kappa[X_i]$. This condition is always true for Gaussian data, but not necessarily for other distributions.
\end{itemize}
\end{rem}

\subsection{A phase transition in FPCA} \label{sec:phaset}

We now turn to the theoretical analysis of the model presented in the previous section. Our first aim is to study the behavior of the empirical eigencomponents $\hat \lambda_n, \hat e_n$ for $1 \le n \le K$, i.e. for the spiked components. To set the stage, we cite a well-known result on the consistency of FPCA for functional data that is not rough. 

\begin{prop} \label{prop:trad} (Theorem 2.10, \cite{HKbook})
Let $Y,Y_1,..,Y_N$ be i.i.d. centered random functions in $L^2([0,1])$ with fixed distribution $\mathcal{L}(Y)$. Suppose that $\mathbb{E}\|Y\|^4<\infty$ and that the covariance operator of the data, $C^Y$, has the eigensystem $(\lambda^Y_n)_n, (e^Y_n)_n$. The empirical covariance is defined analogously to \eqref{e:empcov} and has the eigensystem $(\hat \lambda^Y_n)_n, (\hat  e^Y_n)_n$. Then, it holds that
\[
|\lambda_n^Y-\hat \lambda_n^Y|=\mathcal{O}_{\PR}(N^{-1/2}).
\]
Moreover, if $\lambda_1^Y>\cdots >\lambda_{K+1}^Y$ (and formally defining $\lambda_0:=\infty$), then
\[
\|e_n^Y-\hat e_n^Y\|= \frac{\mathcal{O}_{\PR}(N^{-1/2})}{\min(\lambda^Y_{n-1}-\lambda^Y_n, \lambda^Y_{n}-\lambda^Y_{n+1})} = \mathcal{O}_{\PR}(N^{-1/2}).
\]
\end{prop}
The proposition shows that in FPCA, standard theory implies that empirical eigenvalues are consistent. Empirical eigenfunctions are consistent, too, if there exists a non-vanishing eigengap that identifies them uniquely. Importantly, none of these consistency results distinguishes between the leading eigencomponents, which are intuitively believed to be the components that matter, and the remaining eigencomponents, which are ignored in the applications of FPCA. In contrast, our model and theory provide a quantitative distinction between these two sets of eigencomponents for rough functional data. 
Our subsequent analysis has two aims: 1) We will derive a phase transition that distinguishes spiked eigencomponents that are asymptotically uninformative about their population counterparts (the so-called \textit{subcritical} case) from those that are informative (\textit{supercritical} case). 2) In the supercritical case, we will study  the estimation bias, to measure just how informative the empirical versions are. \\
\textbf{Preliminaries} 
We define the number $\xi(\infty)$ as the solution in $\xi^\star$ to the following equation
\begin{align} \label{e:def:xiinf}
    \int_0^\infty\bigg(\frac{\bb(x)\xi^\star}{1-\bb(x)\xi^\star} \bigg)^2 dx =1, \quad \xi^\star \in (0,1/\bb(0)).
\end{align}
Existence and uniqueness of $\xi(\infty)$ are proven in Lemma \ref{lem:xi}.
The notation $\xi(\infty)$ reflects the fact that it occurs as a limit as $\gamma \to\infty$ of $\xi(\gamma)$, which occurs in random matrix theory. Roughly speaking, $\xi(\gamma)$ captures the noise from the bulk eigenvalues when the data is projected into a finite-dimensional subspace of dimension $\lfloor N\gamma \rfloor $. The noise of the entire bulk, without projection, is then captured by $\xi(\infty)$. Using $\xi(\infty)$, we may now formulate qualitative conditions that characterize when the empirical eigencomponents are informative. More precisely, for $k \in \{1,\cdots,K\}$ we speak of
\begin{itemize}
    \item[] a \textit{subcritical} signal if $ s_k\,\xi(\infty)<1$,
    \item[] a \textit{supercritical} signal if $ s_k\,\xi(\infty)>1$.
\end{itemize}
These notions stem from RMT, where they are also used to distinguish informative from uninformative empirical eigencomponents. In RMT, a number of equivalent characterizations of criticality have been derived, typically involving the spectral measure of eigenvalues (see eq. \eqref{e:spectral} for a definition). For example, in a recent work, \cite{doernemann:lopes:2025} derived a characterization in terms of solutions for eigenvalue equations. This approach inspires the definition of $\xi(\infty)$ in \eqref{e:def:xiinf}. 
An elegant feature of our theory is that the critical threshold $\xi(\infty)$ is the solution to a simple integral equation and does not depend on the sample size $N$ anymore.
To define limits of the sample eigenvalues in both criticality regimes, we define the function
\begin{align}
   \psi(y) =  y\Big(1 + \int_0^\infty \frac{\bb(x)}{y-\bb(x)}dx\Big), \quad y > b(0).
   \label{eq_def_psi}
\end{align}

\noindent \textbf{Mathematical results}  We can now state our main theoretical results. Recall the definition of $K$ in \eqref{e:sigrep}. We define $1 \le M \le K$ as the largest number such that still $s_M \xi(\infty) >1$. If no such number exists, we set $M=0$. 
We gather our assumptions from Section \ref{sec:model}. 

\begin{enumerate}[label=(A-\arabic*)]
\item \label{ass:data} (On the data)  $X_1,\ldots, X_N$ are centered, normally distributed random functions in $L^2([0,1])$, with covariance operator $C$ specified in \eqref{e:sigrep}.
\item \label{ass:bulk:ev} (On the bulk function)  $b: [0,\infty) \to [0,\infty)$ is monotonically decaying with $b(0)>0$, Lipschitz continuous in $x=0$ and satisfies $\int_0^\infty b(x) dx <\infty$.
\item \label{ass:spiked:ev} (On the spiked eigenvalues) If $M<K$, it holds that $s_{M+1} \xi(\infty) < 1.$ 
\end{enumerate}

\begin{theo} \label{e:thm:main1}
    Suppose that conditions \ref{ass:data}-\ref{ass:spiked:ev} are satisfied, and let $k>M$ be a fixed integer. Then, as $N \to \infty$, it holds that
  \begin{align} \label{eq_eigenvals_sub}
      \hat \lambda_k \overset{P}{\to} \psi\big(\xi^{-1}(\infty)\big).
  \end{align}
  If additionally the bulk function $\bb$ is compactly supported, it follows that
    \begin{align}
    \label{eq_delocal_e1}
    \langle \hat e_k , e_k \rangle  \overset{P}{\to} 0, \quad N \to\infty.
   \end{align}
\end{theo}
As we can see, in the subcritical case, $\hat \lambda_k$ converges to a limit that is independent of $\lambda_k$. To see this, notice that the limit of $\hat \lambda_k$ depends on the functions $\psi$, which in turn only depends on the bulk function $\bb$, not on the spiked eigenvalues $\lambda_1,\cdots, \lambda_{K}$. We also see that under subcriticality, there is asymptotically no difference between the $k$th and the $(k+1)$st empirical eigenvalue. The second part of the theorem states that in the subcritical case, the empirical $k$th eigenfunction is asymptotically orthogonal to $e_k$. This is not because $\hat e_k$ converges to any fixed orthogonal element, but rather because $\hat e_k$ is asymptotically non-tight and behaves increasingly like a "uniform distribution" on a cone in $L^2([0,1])$. The assumption of compact support is made for technical reasons and we conjecture that it can be dropped.\\
In the next step, we study the supercritical case.
\begin{theo} \label{e:thm:main2}
  Suppose that conditions \ref{ass:data}-\ref{ass:spiked:ev} are satisfied and let $1 \leq k \leq M$. Then, as $N \to \infty$, it holds that
   \begin{align} \label{eq_eigenvals_super}
      \hat \lambda_k \overset{P}{\to} \psi\big(s_k\big),
  \end{align}
  and 
   \begin{align} \label{e:idev}
   | \langle \hat e_k , e_k \rangle | \overset{P}{\to} \frac{s_k \psi'\big(s_k\big)}{\psi   \big(s_k\big)}. 
  \end{align}

\end{theo}
In the supercritical case, the empirical $k$th eigenvalue is informative about $\lambda_k=s_k$, even though it has an upward bias, captured by $\psi(\cdot)$, and there is a non-vanishing angle between the empirical eigenfunction and its population counterpart. Notice that here we do not need to impose any restrictions regarding the support of $\bb$.

\noindent \textbf{Interpretation for FPCA} The above Theorems \ref{e:thm:main1} and \ref{e:thm:main2} demonstrate that for rough functional data, as considered in this work, the success of FPCA is not a given. Indeed, for the $k$th eigenvalue and eigenfunction to have any meaningful relation to their population counterparts, a supercritical signal is necessary. Otherwise, projecting on the $k$th empirical eigenfunction is basically equivalent to projecting onto a random direction, and the corresponding eigenvalue does not reflect the amount of variance explained. In the supercritical case, the empirical estimators are informative, though biased. These results contrast sharply with the benign behavior of FPCA in the traditional smooth setting that we reviewed in Proposition \ref{prop:trad}. The biggest difference is the important role played by the bulk eigenvalues in the rough regime, while for smooth data their effect is asymptotically negligible. \\
Our results are formulated for a number of $K$ spiked eigenvalues $\lambda_k = s_k, \,k=1,\cdots,K$ that are fixed. If one considers, in contrast, a scenario where $\lambda_{k} = s_{k,N} \uparrow \infty$, while still ensuring that $s_{k,N}-s_{k+1,N}$ is bounded away from $0$, it can be shown that the bias vanishes. This result is unsurprising, as it is known that for an effective rank $\mathfrak{r}[C]=o(N)$ consistency of eigencomponents can be ensured. 
%Our findings suggests that FPCA becomes unstable for rough functional data. If data become too rough, a breaking point is reached and empirical estimators are essentially random guesses. Even before that, roughness comes at the cost of a substantial bias.

\noindent \textbf{Proof techniques} The proofs of our main theorems combine results from RMT with state-of-the-art concentration bounds for Hilbert space valued operators. The latter results are based on recent progress in generic chaining (see, e.g., \cite{alghattas:chen:sanzalonso:2025}). The fundamental strategy is to define the projection $\Pi_{\lfloor N \gamma \rfloor}$ on the eigenspace spanned by the population eigenfunctions $\{e_1,e_2,\ldots,e_{\lfloor N \gamma \rfloor}\}$, where $\gamma>0$ is the parameter discussed after eq. \eqref{e:def:xiinf}. Denoting the projection on the orthogonal complement by $\Pi_{\lfloor N \gamma \rfloor}^c$, we can decompose the empirical covariance $\widehat{C}$ into
\[
\widehat{C} = A(\gamma)+B(\gamma)
\]
where
\[
A(\gamma) = \frac{1}{N}\sum_{i=1}^N \Pi_{\lfloor N \gamma \rfloor}[X_i] \otimes \Pi_{\lfloor N \gamma \rfloor}[X_i]
\]
and $B(\gamma)$ is then defined as the remainder. The proofs consist of three key ingredients:\\
1) We show that the eigenvalues(-vectors) of $\widehat{C}$ are close to those of $A(\gamma)$. This follows partly by results from perturbation theory and partly from eigengap derivations that rely on BBP-transitions for high-dimensional covariances. \\
2) The eigenvalues(-vectors) of $A(\gamma)$ are analyzed by tools from RMT for fixed $\gamma$ and as $N \to \infty$. Here, conditions for sub- and supercriticality exist, which "converge" for any fixed $\gamma$ in our model to an integral condition involving $b$. 
We then study separately the limit as $\gamma \to \infty$.\\
3) Handling the remainder $B(\gamma)$ is non-trivial because it can be shown not to vanish as $N \to \infty$. So, we derive a result of the following form, for any $\epsilon>0$:
\[
\lim_{\gamma \to \infty} \limsup_{N \to \infty} \mathbb{P}(\|B(\gamma)\|_\mathcal{L}>\epsilon)=0.
\]
Analyzing the inner limit (and obtaining bounds that decay in $\gamma$) requires the use of concentration bounds for sums of random tensors that explicitly depend on the effective rank. Here the recent results from generic chaining are used, and especially those by \cite{chen:sanzalonso:2025}.

\begin{rem}
Pre-smoothing has been used for decades to convert visually rough functional data to
smooth curves, to which various FDA methods would then be applied. We briefly comment on this strategy in light of the results of this section. 
\begin{itemize}[leftmargin=*]
    \item[1)] From a practical perspective, pre-smoothing the random functions $X_i$ can help to concentrate them more closely in a low-dimensional space of (smooth) functions. This can lead to more stable estimates in FPCA.
    However, it is important to notice that a reasonable amount of smoothing does not automatically eliminate the phase transitions and biases described by our theory. Indeed, in the Supplement, we have made versions of Figure \ref{fig:1} using different strengths of pre-smoothing for the data. There, we see essentially the same picture with biases that are only very slowly receding as more smoothing is applied. Very aggressive smoothing may change this picture and lead to strong FPCA consistency expected from classical theory. However, applying a lot of smoothing also entails statistical costs.  
    First, a smoothed curve will have a systematically different structure from the true curve $X_i$, making it harder to interpret. Notice that for the examples in the Introduction, as in many others, roughness is not an artifact of measurement errors that should be removed, but a true feature of the observed scientific variable.
    Second, since smoothing is a deliberate attempt to filter out "noisy frequencies", it runs the risk of filtering out relevant high-frequency events too as, e.g., short-duration extreme events in weather or flood data.
    \item[2)] From a mathematical perspective, smoothing a function in a geometric sense is not the same as smoothing it in an $L^2$-sense. From an $L^2$-perspective, the ideal pre-smoothing of the data $X_i$ for eigenanalysis would be projecting it on the space spanned by the leading eigenfunctions $\{e_1,\ldots,e_K\}$. Practical smoothing will, however, look different from that. For example, it might involve some sort of local averaging, or more mathematically speaking, convolving the function $X_i$ with some kernel function to get a visually smooth curve. A somewhat unpleasant finding of our data analysis is that these two notions of smoothing are not only theoretically but also practically different; see the end of Section \ref{sec:data}. Indeed, the relation between more geometric smoothing to supercriticality of eigencomponents seems to be highly nonlinear. Typically a small amount of smoothing leads to stronger supercriticality, but beyond that smoothing may increase or decrease supercriticality in a rather unpredictable way. 
\end{itemize}
\end{rem}

\subsection{Diagnostic tools} \label{sec:diag}

In the previous section, we have determined theoretical criteria for informative and uninformative FPCA. We now discuss a statistical tool for users to practically determine whether FPCA is informative or not. For this purpose, let $K_1 \ge 1$ be a fixed, natural number and define the eigengap ratio statistic as 
\begin{align} \label{e:def:Lh}
\widehat{\Lambda}(K_1) := \max_{1 \le k \le K_1} \widehat{R}_k\quad \textnormal{where} \quad \widehat{R}_k := \frac{\hat\lambda_k - \hat\lambda_{k+1}}{\hat\lambda_{k+1} - \hat\lambda_{k+2}}. 
\end{align}
We also define the limiting distribution 
\[
\Lambda(K_1) := \max_{1 \le k \le K_1} R_k\quad \textnormal{where} \quad R_k := \frac{\zeta_k - \zeta_{k+1}}{\zeta_{k+1} - \zeta_{k+2}},
\]
and where the vector
$(\zeta_1,\dots,\zeta_{K_1+2})$ follows a $(K_1+2)$-dimensional Tracy-Widom distribution. For details on that distribution we refer to Section \ref{sec_background_rmt} in the Supplement.
The next result discusses the convergence behavior of $\widehat{\Lambda}(K_1)$ in sub- and supercritical regimes. We impose the assumption of a finitely supported bulk function for technical convenience. Since data need to be saved as finite-dimensional discretizations, in practice $\bb$ will always have compact support.

\begin{prop} \label{prop:onat}
    Suppose that conditions \ref{ass:data}-\ref{ass:spiked:ev} are satisfied and that $\bb$ is compactly supported. 
    \begin{itemize}
        \item[(i)] If $M=0$, the weak convergence $ \widehat{\Lambda}(K_1) \overset{d}{\to} \Lambda(K_1)$ follows.
        \item[(ii)] If $1 \le M \le K_1$ it follows that $ \widehat{\Lambda}(K_1) \overset{P}{\to} \infty$.

    \end{itemize}
\end{prop}
Following the construction by \cite{onatski:2009}, we can use the eigengap ratio statistic $\widehat{\Lambda}(K_1) $ to test the null hypothesis of no supercritical components, i.e. of $H_0: M=0$. To be precise, for a significance level $\alpha \in (0,1)$ let $q_{1-\alpha}$ be the  $(1-\alpha)$-quantile of the distribution $\Lambda(K_1)$. It then follows under $H_0$ that
  \[
    \mathbb{P}(\widehat{\Lambda}(K_1)>q_{1-\alpha}) \to \alpha. 
\]
In contrast, if $1 \le M \le K_1$, the second part of Proposition \ref{prop:onat} entails that
    \[
    \mathbb{P}(\widehat{\Lambda}(K_1)>q_{1-\alpha}) \to 1. 
    \]
Notice that the threshold $K_1$ needs to be chosen by the user. Since $M$ is practically unknown, it is typically safest to err on the side of a larger $K_1$.

\section{Simulations} \label{sec:simulations}
To illustrate the theory of Section \ref{sec:main},
we study related finite sample properties of FPCA for rough functional data.\footnote{Codes for simulations and data analyses are available at our repository \url{https://github.com/timkutta/PCA-for-rough-functional-data}}
To make our presentation compact, we will focus on the case  $K=1$, i.e., of a single spike that can be either super- or subcritical. Throughout this section, the sample size is fixed at $N=100$, unless stated otherwise. Random functions are generated using the Fourier basis $(f_i)_{i \in \mathbb{N}}$ of $L^2([0,1])$. Theoretically, functional data can be infinite dimensional, but for our simulations we need to truncate at some point. Therefore, we consider a $p$-dimensional basis expansion with $p=5 \,N$. Non-reported results for $p=10\,N$ look almost identical. Any function $X_i$ is then generated by drawing iid standard Gaussian variables $(Z_{i,j})_{1 \le j \le p}$ and setting
\[
X_i(t) = \sum_{j=1}^p Z_{i,j} \sqrt{\lambda_j} f_j(t), \qquad t \in [0,1].
\]
The first eigenvalue $\lambda_1$ is spiked, and the remaining eigenvalues are characterized by a bulk function $b$ introduced in Section \ref{sec:model}. We consider three bulk functions, namely
\begin{align*}
    b_1(x) = & (1+ax)^{-3},\\
    b_2(x) = & \exp(-ax),\\
    b_3(x) = & \max(1-ax,0).
\end{align*}
All bulk functions satisfy $b(0)=1$, which corresponds to the largest non-spiked eigenvalue. The constant $a>0$ is different in each function and adjusted to yield the same value for the criticality threshold $1/\xi(\infty)=1.3$, which makes our subsequent results more comparable.

\begin{figure}[h]
  \centering
  \begin{subfigure}{0.49\textwidth}
    \centering
    \includegraphics[width=\linewidth]{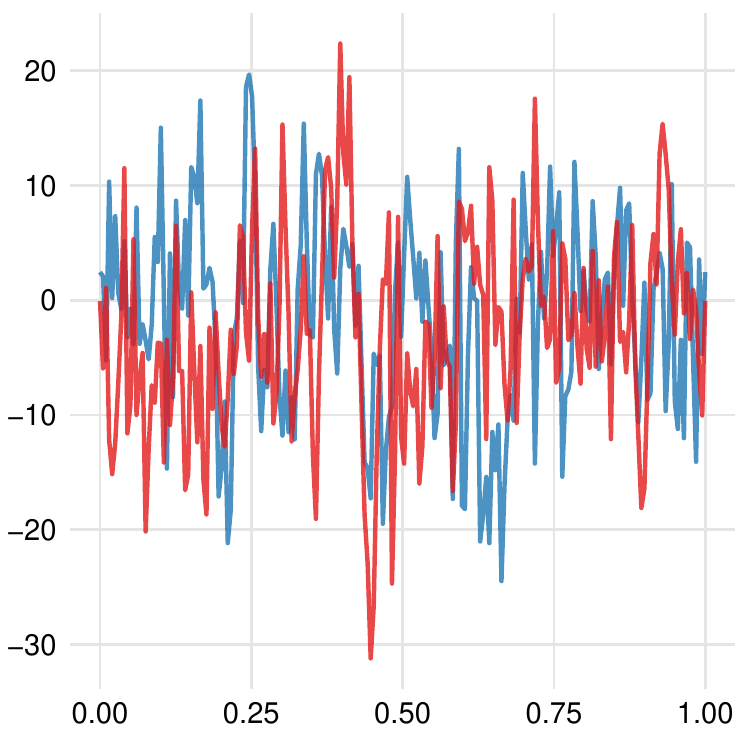}
  \end{subfigure}
  \hfill
  \begin{subfigure}{0.49\textwidth}
    \centering
    \includegraphics[width=\linewidth]{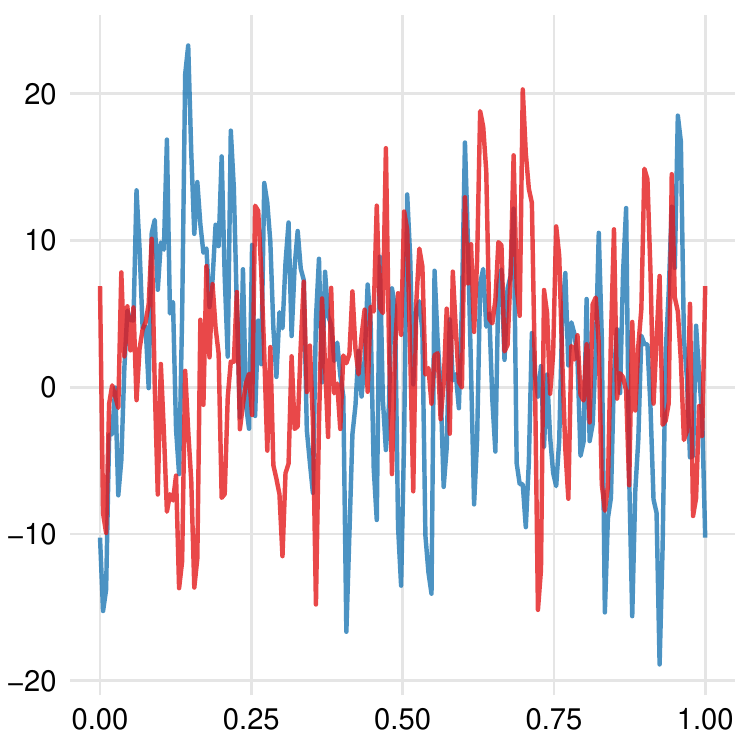}
  \end{subfigure}
  \caption{\label{fig:EMPEF}  Two realizations of $X_i$ for $b_1$ (left) and $b_2$ (right).}
\end{figure}

\noindent \textbf{Phase transitions for eigenvalues} For the subcritical case, we consider the value $\lambda_1=1.1$ and for the supercritical value $\lambda_1=2$. 
In Figure \ref{Fig:eigenW1}, we display histograms  of the largest empirical eigenvalue $\hat \lambda_1$ from $500$ simulation runs for bulk function $b_1$. Corresponding plots for bulk functions $b_2, b_3$ are provided in the Supplement. In each plot, the blue vertical line marks the median value of $\hat \lambda_1$ across the simulations and the red vertical line the theoretical limit from Theorems \ref{e:thm:main1} and \ref{e:thm:main2} respectively. 

\begin{figure}[H]
  \centering
  \begin{subfigure}{0.49\textwidth}
    \centering
    \includegraphics[width=\linewidth]{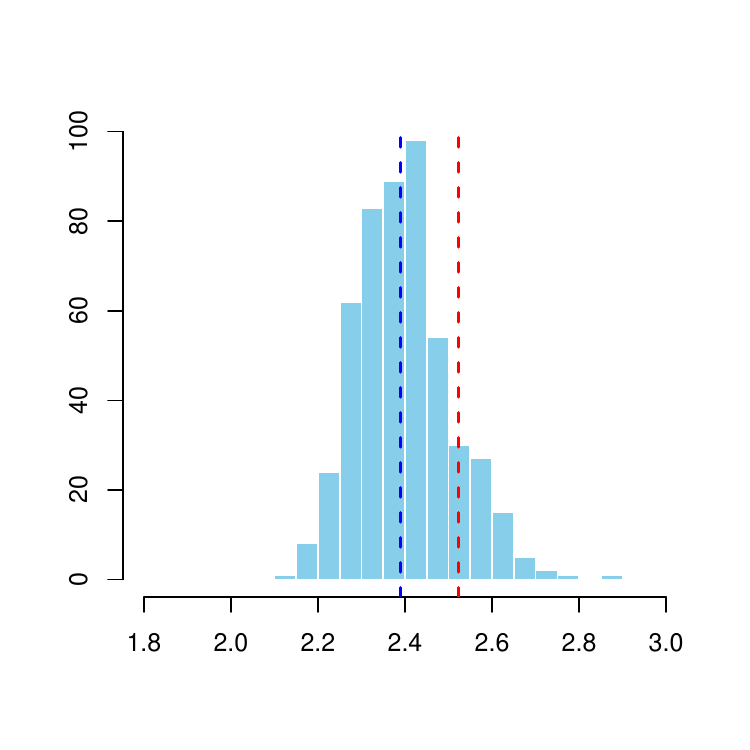}
  \end{subfigure}
  \hfill
  \begin{subfigure}{0.49\textwidth}
    \centering
    \includegraphics[width=\linewidth]{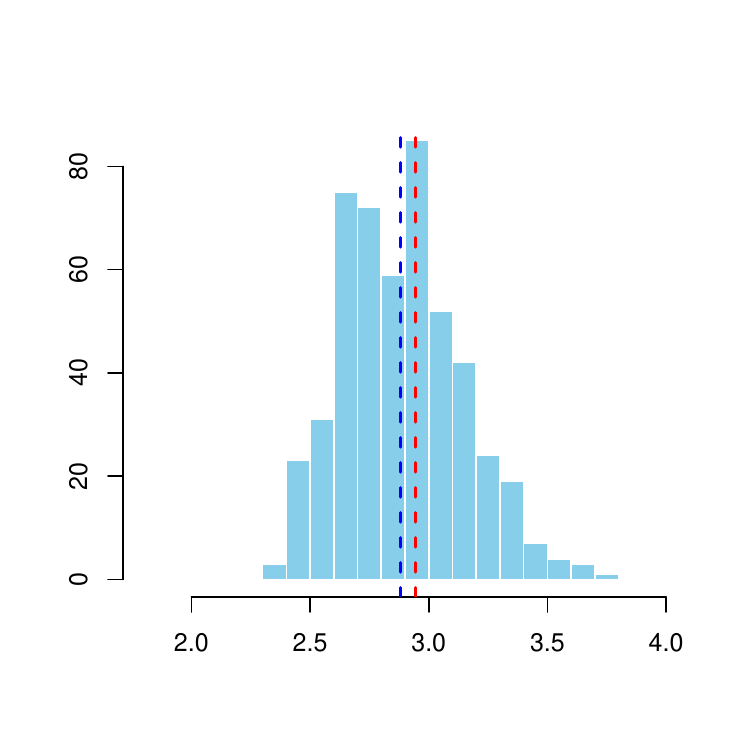}
  \end{subfigure}
  \caption{\label{Fig:eigenW1}   Distribution of largest eigenvalue $\hat \lambda_1$ in $500$ simulation runs.
  Blue vertical line indicates median of the values and red vertical line the theoretical limit. The subcritical case ($\lambda=1.1$) is left, supercritical case ($\lambda=2$)  is right. The bulk function is $b_1$.}
\end{figure}

\begin{figure}[H]
  \centering
  \begin{subfigure}{0.49\textwidth}
    \centering
    \includegraphics[width=\linewidth]{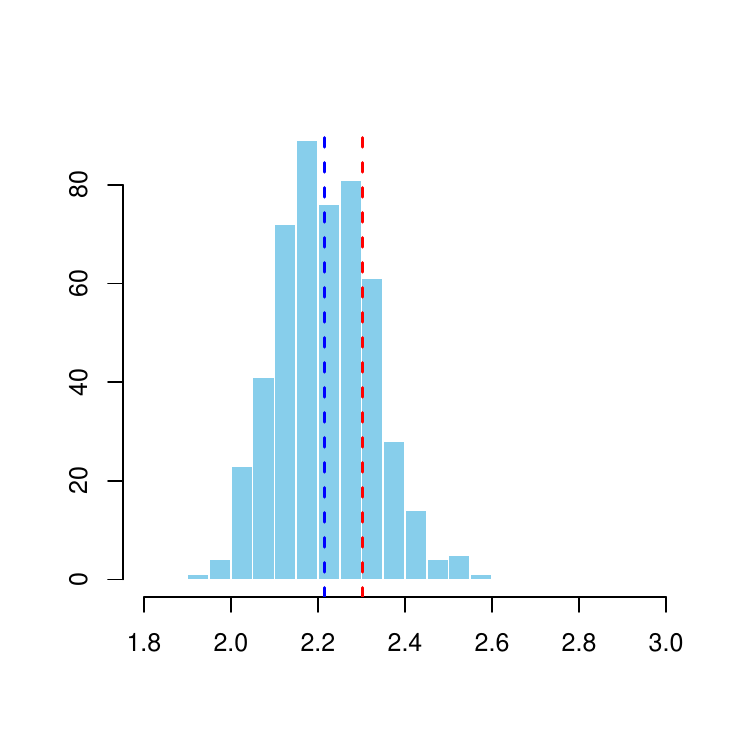}
  \end{subfigure}
  \hfill
  \begin{subfigure}{0.49\textwidth}
    \centering
    \includegraphics[width=\linewidth]{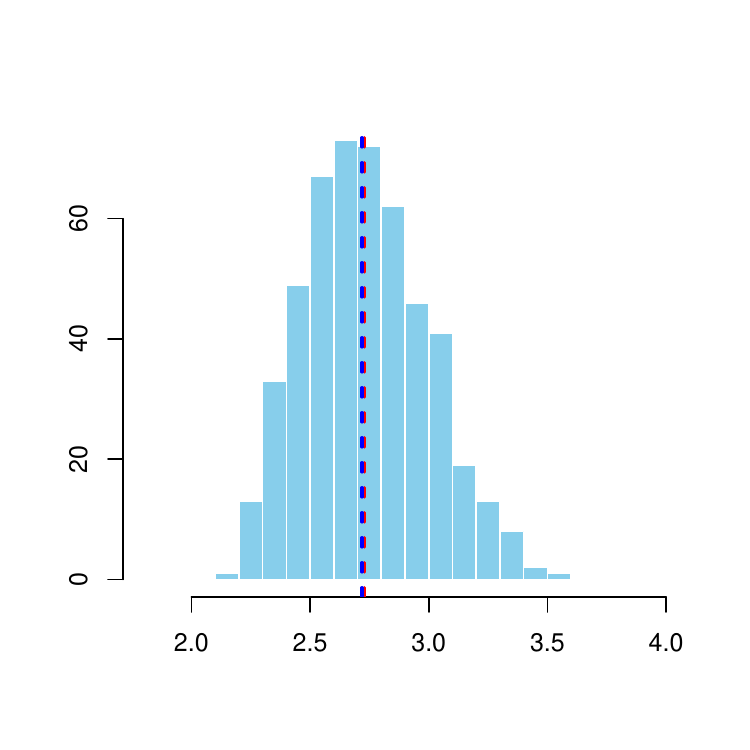}
  \end{subfigure}
  \caption{\label{Fig:eigenW2}   Distribution of largest eigenvalue $\hat \lambda_1$ in $500$ simulation runs for empirically centered data.
  Blue vertical line indicates median of the values and red vertical line the theoretical limit. The subcritical case ($\lambda=1.1$) is left, supercritical case ($\lambda=2$) is right. The bulk function is $b_1$.}
\end{figure}

We see that in both cases the median value of $\hat \lambda_1$ is much higher than the population eigenvalue $\lambda_1$ ($1.1$ and $2$ respectively), which is congruent with our theory. In both scenarios we see alignment of the median value of $\hat \lambda_1$ with its theoretical limit, yet the approximation seems substantially tighter in the supercritical regime. In the subcritical case, we see that the theoretical limit is systematically  larger than the median value of $\hat \lambda_1$, though by a small margin relative to the difference $\hat \lambda_1-\lambda_1$.
This phenomenon has only a weak connection to the specific shape of the bulk function $b$ and non-reported results confirm that the distance of the median to the theoretical limit diminishes for larger $N$ at a moderate pace. 
In the simulations for Figure \ref{Fig:eigenW1}, we used the covariance estimator \eqref{e:empcov}, for which we have developed our theory and which does not involve empirical centering of the data. In most applications, centering the data will be required and we therefore also study the eigenvalue distribution for the data where $X_i$ is replaced by $X_i-\frac{1}{N}\sum_{j=1}^N X_j$. We only report histograms for the case $b_1$ in Figure \ref{Fig:eigenW2}, with the same setting as before. We cannot see a systematic difference to the  case without centering.

\noindent \textbf{Phase transitions for eigenfunctions} For our investigation of empirical eigenfunctions, we use the same simulation setups as before and report the average angle between $\hat e_1$ and $e_1$. Notice that eigenfunctions are only determined up to sign and thus we always specify the sign that yields an angle $\le 90$°.  Again we run $500$ simulations in each scenario, using the bulk functions $b_1$ (results for $b_2, b_3$ in Supplement) and the sub- and supercritical eigenvalues of $\lambda_1=1.1$ and $\lambda_1=2$ respectively. Our results are displayed in Figure \ref{Fig:eigenV1}. In the subcritical case, we see that angles are on average large, concentrating close to the maximum possible value of $90$° that Theorem \ref{e:thm:main1} predicts. Of course values are, in finite samples, not precisely $90$°, but the qualitative result of strong unreliability of the estimate $\hat e_1$ becomes apparent. In the supercritical case, angles between empirical and true eigenfunctions do not degenerate (Theorem \ref{e:thm:main2}) and this is apparent on the right-hand side of Figure \ref{Fig:eigenV1}. Again we see that real angles enclosed by $\hat e_1$ and $e_1$ are smaller than their theoretical limit, and the distance of the theoretical prediction to the median simulated angle is about $13$° - a relatively large deviation. In non-reported simulations, we observed that increasing $\lambda_1$ leads to smaller angles between $\hat e_1$ and $e_1$ and also to a narrowing of the gap between theoretical prediction and median angle.\\
%Finally, the second part of Theorem \ref{e:thm:main1} states that in the subcritical case empirical and population eigenfunctions are asymptotically orthogonal, assuming that $b$ has compact support. In our simulations, we truncate data at dimension $p=5\,N$, effectively ensuring that the support of $b$ is inside the compact interval $[0,5]$. We do, however, conjecture that this assumption can be dropped. We investigate the effect of larger $p$, by setting $p=10\, N$ and $15\,N$ respectively in the subcritical case, and for $b_1$. Results are reported in Figure \ref{Fig:eigenV2}. As expected, there is no discernible effect of raising $p$.

\begin{figure}[H]
  \centering
  \begin{subfigure}{0.49\textwidth}
    \centering
    \includegraphics[width=\linewidth]{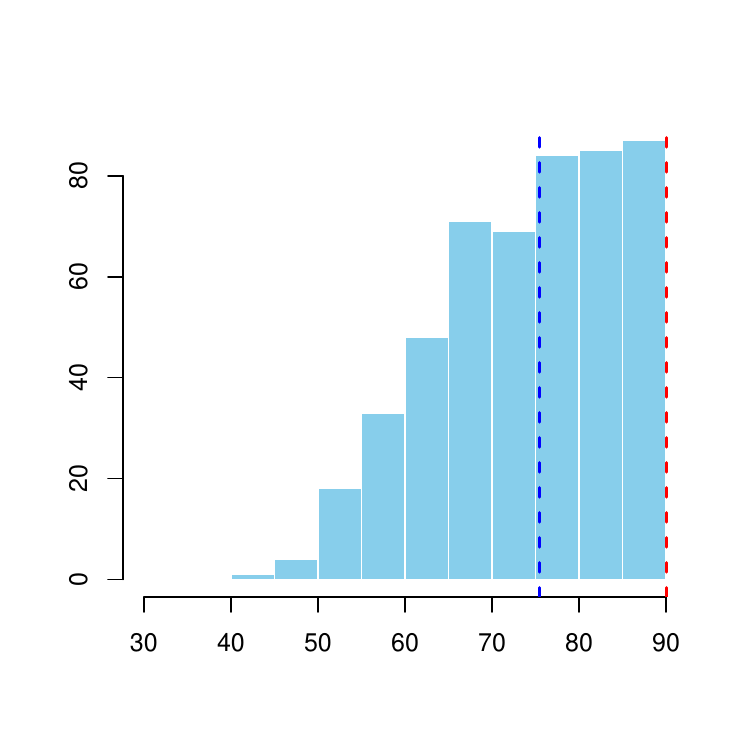}
  \end{subfigure}
  \hfill
  \begin{subfigure}{0.49\textwidth}
    \centering
    \includegraphics[width=\linewidth]{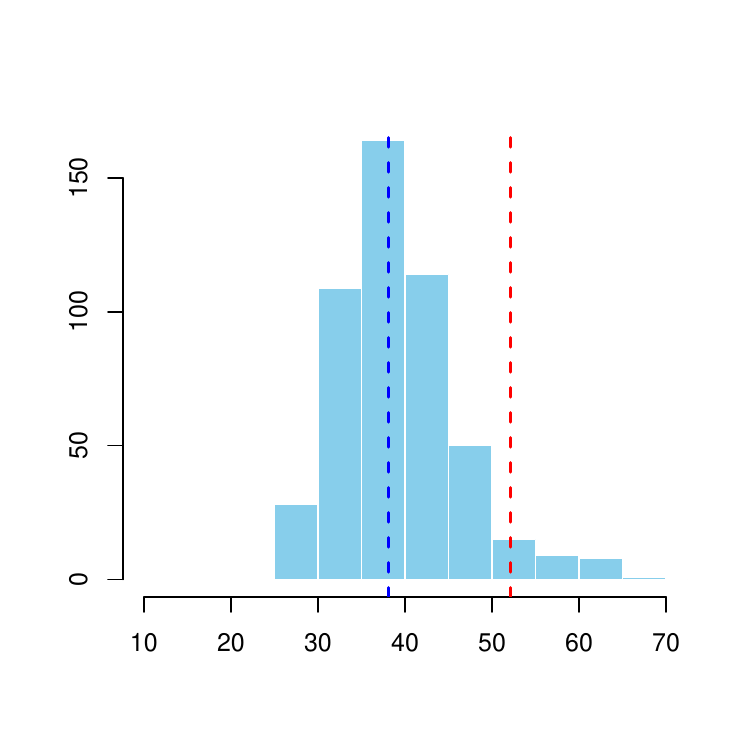}
  \end{subfigure}
  \caption{ \label{Fig:eigenV1} Distribution of the angle between $\hat e_1$ and $e_1$ in $500$ simulation runs. Blue vertical line indicates median of the values and red vertical line the theoretical limit. The subcritical case is left, the supercritical case is right. The bulk function is $b_1$.}
\end{figure}

\noindent \textbf{Testing for supercriticality}  Finally, we want to consider the eigenvalue ratio test statistic, presented in Section \ref{sec:diag}. This method is used to statistically test whether there are any supercritical components.   For the application of the test, we need to select a 
\begin{wraptable}{r}{0.45\textwidth}
\centering
\vspace{-10pt} % reduces vertical gap; adjust if needed
\caption{\label{Tab:1} Empirical rejection probabilities of the eigenvalue ratio statistic in percent.}
\begin{tabular}{c | c c c}
 & \multicolumn{3}{c}{bulk function} \\
\hline
\(K_1\) & $b_1$ & $b_2$ & $b_3$ \\
\hline
%1 & 7.0 & 6.6 & 8.0 \\
2 & 6.0 & 5.2 & 5.0 \\
3 & 4.8 & 6.0 & 6.0 \\
\end{tabular}
\end{wraptable}
hyper-parameter $K_1$ and we consider the choices $K_1=2,3$ which should lead to asymptotically consistent tests according to our discussion in Section \ref{sec:diag}.
The significance level is set to $\alpha=5\%$ and the theoretical quantile $q_{1-\alpha}$ (which depends on $K_1$) is computed using $1000$ replications from a parametric bootstrap. Empirical rejection rates are based on $500$ simulation runs. We consider the setting from this section, where the threshold that separates subcritical from supercritical eigenvalues is $1.3$.
Under the null-hypothesis we consider $\lambda_1=1.1$ and for alternatives of varying size $\lambda_1= 2 +i/4$ with $i=0,1,\ldots,12$. Under the null-hypothesis we consider bulk functions $b_i, i=1,2,3$. Under the alternative we focus exemplarily on $b_1$. Empirical rejection rates under $H_0$ are gathered in Table \ref{Tab:1} and well approximate $5\%$ in all scenarios. Power, depending on $K_1$ and $\lambda_1$ is displayed in Figure \ref{Fig:power}. Power is non-trivial for the supercritical value of $\lambda_1=2$ and increases rapidly for stronger spikes. The larger value $K_1=3$ is associated with a mild power loss compared to $K_1=2$, as expected by our theory.

\begin{figure}[h!]
\centering
\includegraphics[width=0.9\textwidth]{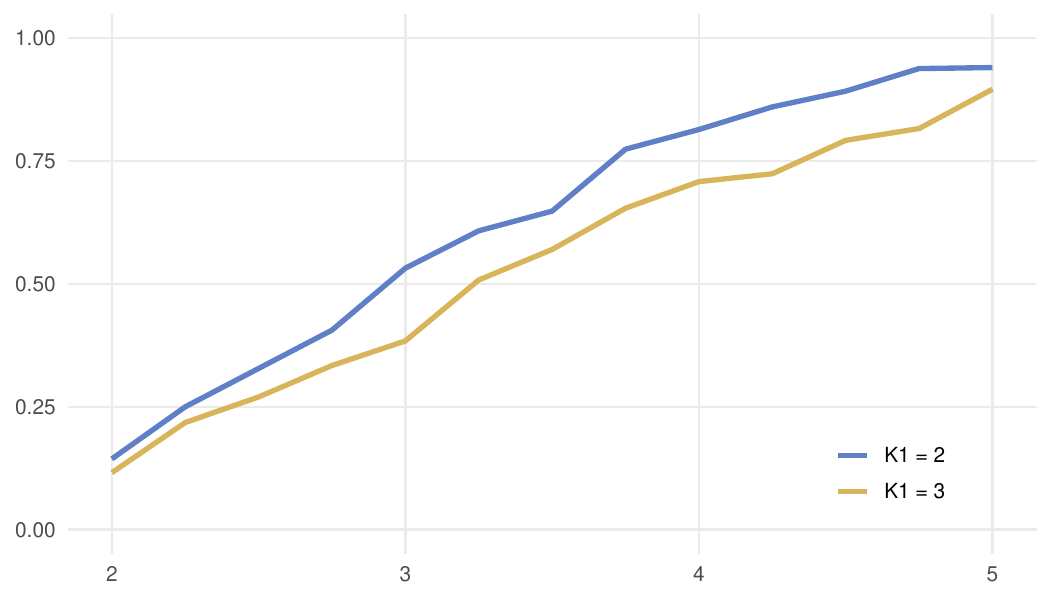}
\caption{\label{Fig:power} Empirical power of the eigengap ratio statistic, depending on the size of $\lambda_1$ ($x$-axis) and the choice of $K_1$.}
\end{figure}

\section{Data analysis} \label{sec:data} 

Recall the two datasets (Examples 1 and 2) from the Introduction. The datasets consist of annual profiles of temperature curves and river discharge data. Both datasets consist of fairly rough functional data, and in the Introduction, we presented evidence for highly unstable FPCA. In this section, we analyze these datasets using the background of our new mathematical theory.

\noindent \textbf{Geometric properties} In a first step, we consider the geometric properties of the eigenfunctions, and we confine our presentation to the temperature data in Example 1 to avoid redundancy. Theorems \ref{e:thm:main1} and \ref{e:thm:main2} suggest that while the most spiked (supercritical) empirical eigenfunctions may have some meaningful relation to their population counterparts, subsequent empirical eigenfunctions should be close to orthogonal to their population counterparts and basically be random guesses. For the temperature data, Figure \ref{fig:1} suggests that the first eigenfunction is stable enough to be truly informative, while for the remaining empirical eigenfunctions we expect uninformativeness.
% In the temperature data, we saw that (perhaps) the first eigenfunction is stable enough to be truly informative (see Figure \ref{fig:1}), while for the remaining empirical eigenfunctions we expect uninformativeness. 
A first interesting question is how informative versus uninformative eigenfunctions look like, geometrically compared to each other. For this purpose, we consider the first and the tenth empirical eigenfunctions, based on the entire 235 years of temperature data. The first and the tenth eigenfunctions are depicted in the upper panels of Figure \ref{fig:stacked_plots}. Visually, a feature of more informative eigenfunctions seems to be a visible macro structure (as e.g. the U-shape in panel (a)), while higher order eigenfunctions look more like a weakly dependent, stationary time series (as in panel (b)). This impression gets stronger, when discretizing the eigenfunctions into daily values and calculating the empirical autocorrelation function (ACF). In the case of the first eigenfunction it decays much slower than for the tenth eigenfunction, corresponding to a clearer macro structure. Eigenfunctions of very high order visually look increasingly like a white noise, and the illustrations given here are replicated for eigenfunctions of order  $k=50, 100$ in the Supplement. 

\begin{figure}[htbp]
    \centering
    % Upper row
    \begin{subfigure}[b]{0.48\textwidth}
        \centering
        \includegraphics[width=\linewidth]{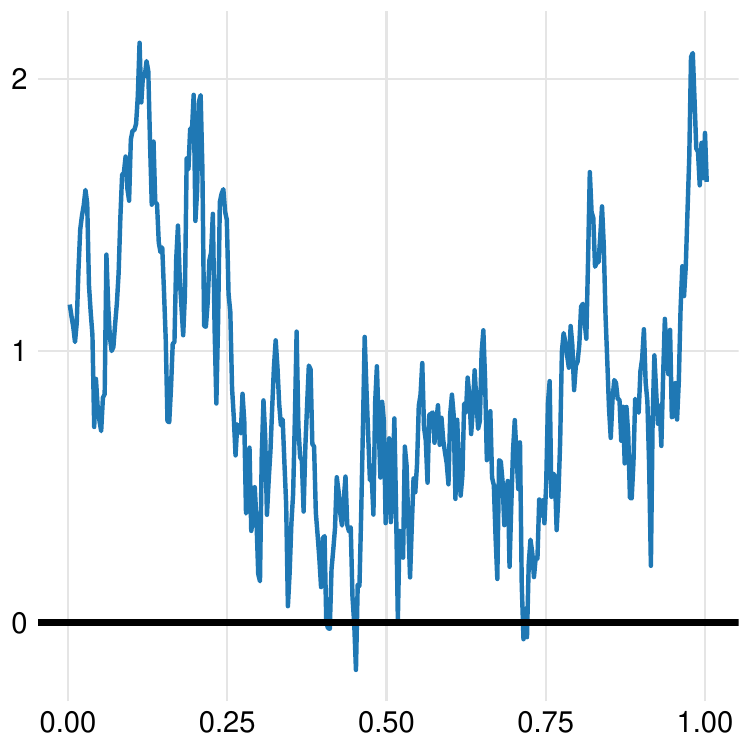}
        \caption{First empirical eigenfunction for temperature data.}
    \end{subfigure}
    \hfill
    \begin{subfigure}[b]{0.48\textwidth}
        \centering
        \includegraphics[width=\linewidth]{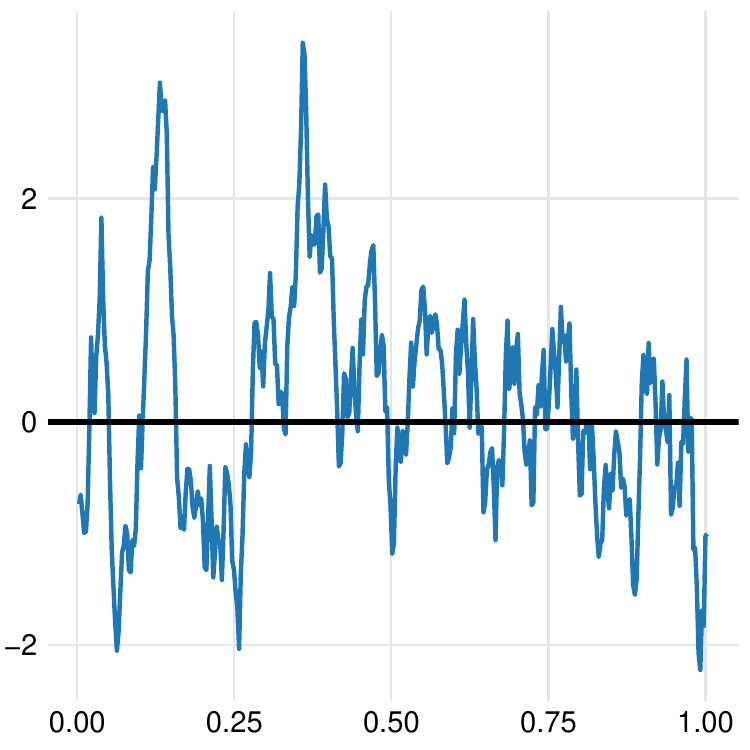}
        \caption{Tenth empirical eigenfunction for temperature data.}
    \end{subfigure}
    
    \vspace{0.5cm} % optional vertical spacing between rows
    
    % Lower row
    \begin{subfigure}[b]{0.48\textwidth}
        \centering
        \includegraphics[width=\linewidth]{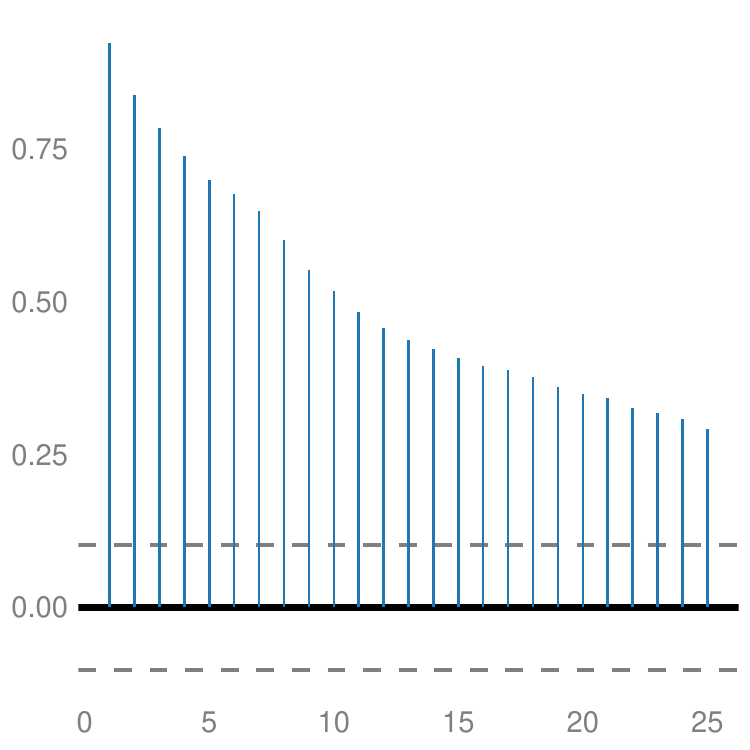}
        \caption{ACF for first eigenfunction.}
    \end{subfigure}
    \hfill
    \begin{subfigure}[b]{0.48\textwidth}
        \centering
        \includegraphics[width=\linewidth]{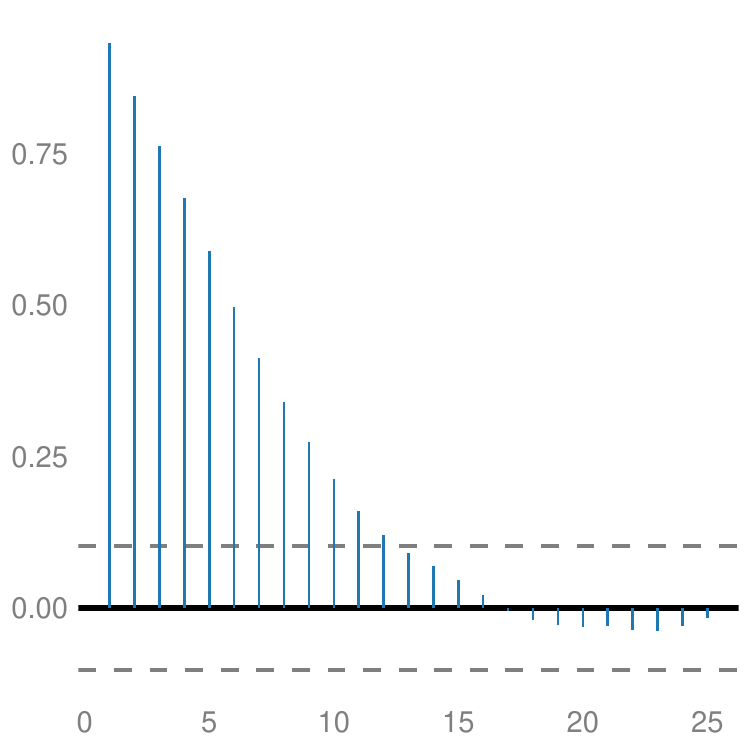}
        \caption{ACF for tenth eigenfunction.}
    \end{subfigure}
    \caption{Comparison of eigenfunctions for first (left) and tenth (right) component. ACFs are calculated (lower) for a daily discretization of these functions. Grey horizontal lines indicate the standard $95\%$ confidence interval around a correlation of $0$. \label{fig:stacked_plots}}

\end{figure}

Next, it is of interest to look at the variability and instability of empirical eigenfunctions. For that purpose, we take random subsamples of size $N=100$ from the temperature curves (with replacement) and calculate the empirical eigenfunctions. We display $50$ realizations of the  eigenfunctions for $k=1$ and $k=10$ in Figure \ref{fig:EMPEF} as thin blue lines. The bold blue line indicates the average over the eigenfunctions. We observe a high degree of variability in both cases. Yet, for $k=1$, the U-shape of the estimates is visible for individual curves and their average. For $k=10$ no such shape is visible, the distribution is more symmetrical around the line  $y=0$, and the mean is close to this (uninformative) line as well.

\begin{figure}[h]
  \centering
  \begin{subfigure}{0.49\textwidth}
    \centering
    \includegraphics[width=\linewidth]{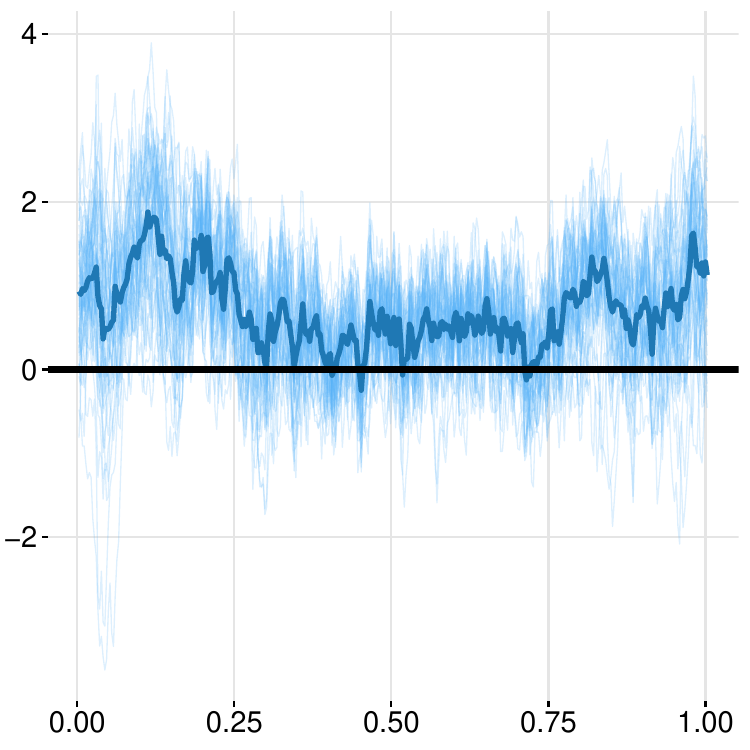}
  \end{subfigure}
  \hfill
  \begin{subfigure}{0.49\textwidth}
    \centering
    \includegraphics[width=\linewidth]{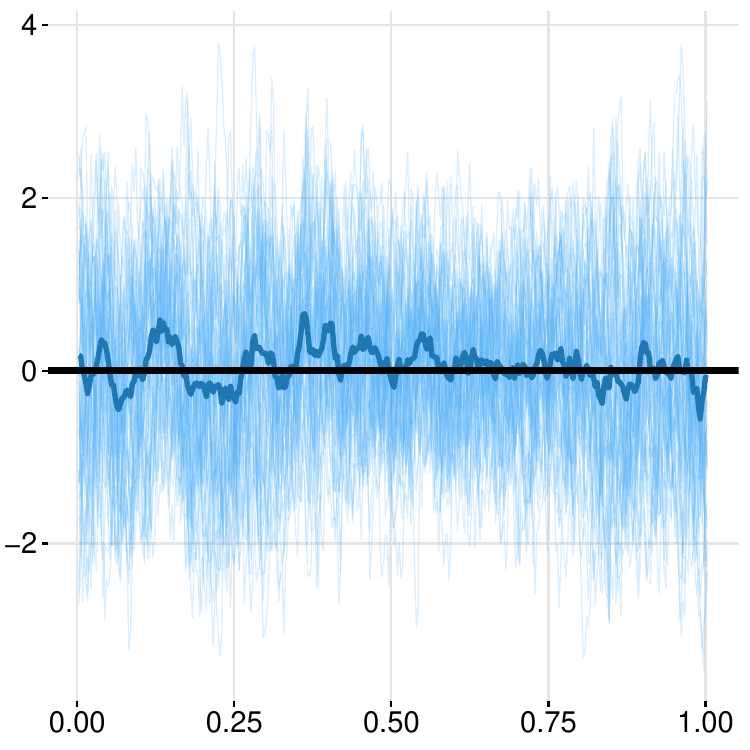}
  \end{subfigure}
  \caption{\label{fig:EMPEF} $50$ realizations of empirical eigenfunctions of order $k=1$ (left) and order $k=10$ (right). Individual estimates are depicted as light blue lines and the mean as a bold blue line.
  All estimates are based on a sample size of $N=100$.}
\end{figure}

\noindent \textbf{Parameter approximation} FPCA is often used to provide low-dimensional subspaces for parameter estimation. Given the highly irregular structure of empirical eigenfunctions, we might wonder how effective this strategy is. For both datasets, we consider the empirical mean and compare it to a projection on $k=5,10$ empirical eigenfunctions. To have a baseline, we compare to a projection on the first $k$ Fourier basis elements. Results are displayed in Figure \ref{fig:means} with the solid blue line the mean before projection, the dotted green line the mean after projecting on the Fourier basis and the red dashed line the mean after projecting on the empirical eigenfunctions. The projection on the Fourier basis is often so close to the original that it is visually hard to distinguish. In panel (b) we see that the Fourier projection is visibly smoother than the original mean. 
The projections on the empirical eigenfunctions provide extremely poor approximations, both in an $L^2$- and, more visibly, in a uniform sense. From an $L^2$-perspective, the very rough eigenfunctions have a large angle with any smooth function (such as the temperature mean) and are therefore bad at representing it by a basis expansion. 
This could be a serious problem, because often parameter functions in FDA are assumed to be smooth. However, even for the much rougher discharge mean, approximations are poor. This should remind us that there are so many ways for a function to be rough, that representing a rough function by a rough basis is often unsuccessful (here too, the smooth Fourier basis performs much better).
It will be interesting to study these effects in other settings such as functional linear regression, where parameter estimates are often made using the space of eigenfunctions; however, this study goes beyond the scope of this work. 

\begin{figure}[htbp]
    \centering
    % Upper row
    \begin{subfigure}[b]{0.49\textwidth}
        \centering
        \includegraphics[width=\linewidth]{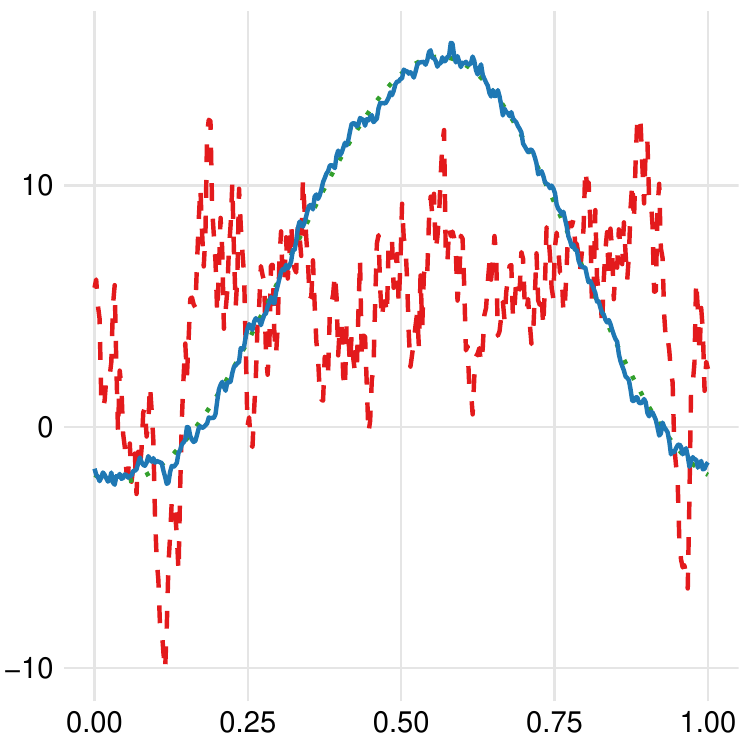}
        \caption{Temperature mean and projections for $k=5$ basis functions.}
    \end{subfigure}
    \hfill
    \begin{subfigure}[b]{0.49\textwidth}
        \centering
        \includegraphics[width=\linewidth]{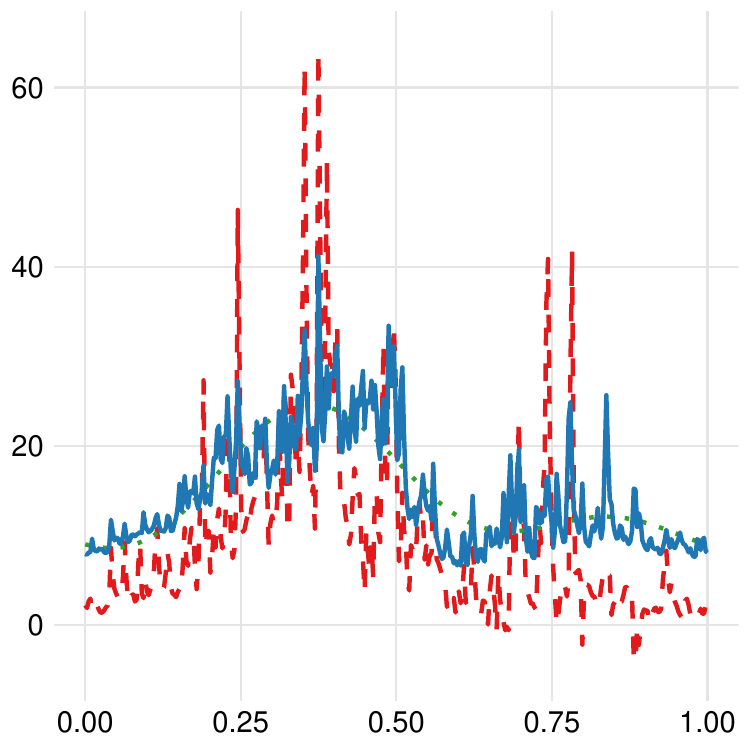}
        \caption{River discharge mean and projections for $k=5$ basis functions.}
    \end{subfigure}
    
    \vspace{0.5cm} % optional vertical spacing between rows
    
    % Lower row
    \begin{subfigure}[b]{0.48\textwidth}
        \centering
        \includegraphics[width=\linewidth]{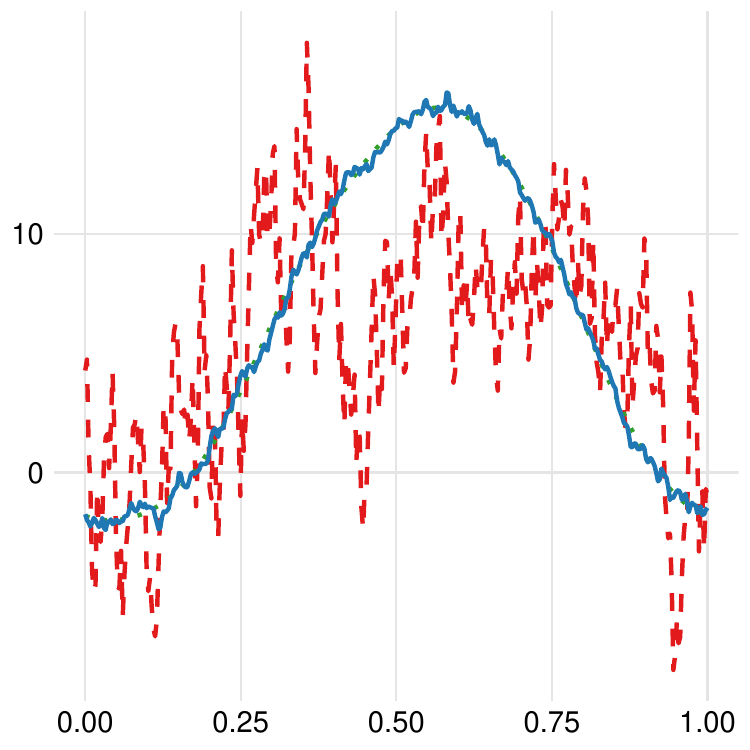}
        \caption{Temperature mean and projections for $k=10$ basis functions.}
    \end{subfigure}
    \hfill
    \begin{subfigure}[b]{0.48\textwidth}
        \centering
        \includegraphics[width=\linewidth]{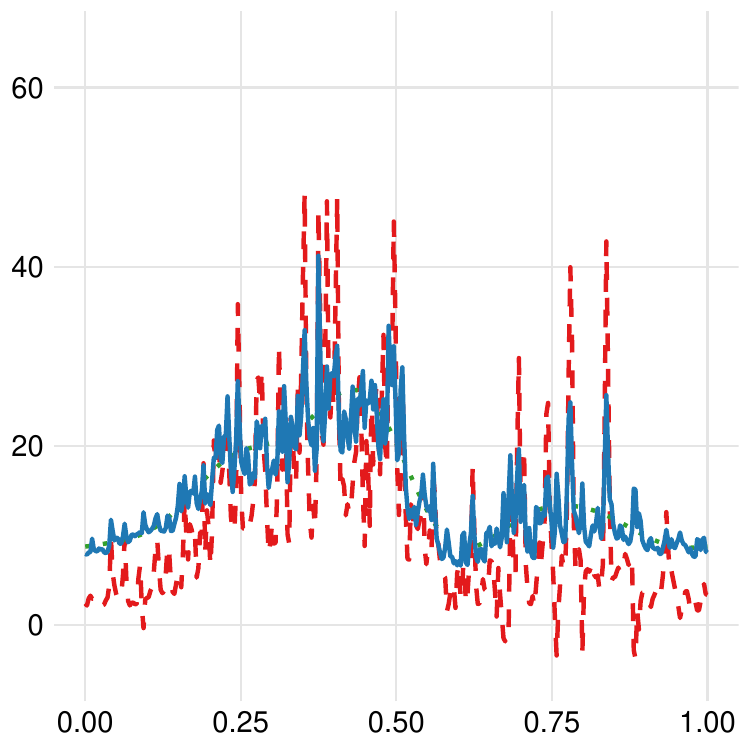}
        \caption{River discharge mean and projections for $k=10$ basis functions.}
    \end{subfigure}
    \caption{Mean function and its projections for temperature data (left) and river discharge data (right), using a basis expansions with $k=5$ (upper) and $k=10$ (lower) components. The unprojected mean is the blue solid line. The projection on the Fourier basis is the dotted green line and the projection on the empirical eigenfunctions is the red dashed line. \label{fig:means}}

\end{figure}

From the perspective of uniform convergence, it is well known that the $L^2$-projection of a function may be a very poor approximation if the basis is not in some sense "well-adapted". For smooth functions, projections on the Fourier basis can be shown to converge rapidly in a uniform sense -- however, this property cannot be expected from a basis of eigenfunctions which geometrically looks like sawteeth. In the Supplement, we have plotted projections for higher orders of $k$ and see that even for $k=50, 100$, the projections look poor. Moreover, we can see that there is only very slow progress when increasing $k$.

\noindent \textbf{Implications for statistical inference} Another important application of functional principal components in practice is hypothesis testing. We will here explore the impact of instabilities on one-sample tests for the functional mean parameter. The central question is whether FPCA instability inflates or deflates the empirical type-I-error. The lesser of two evils would be type-I-error deflation, so that significance levels  still hold, though possibly at a cost to statistical power. Fortunately, this is indeed what we are observing, and we can even provide some mathematical intuition why this is an expected outcome beyond our particular analysis.

We focus on the temperature data from Example 1 and we consider the empirical measure across all temperature curves and call it $\mathbb{P}_{temp}$. The mean function of this measure will be denoted by $\mu_{temp}$. In the following, we will draw i.i.d. samples of size $N$ from $\mathbb{P}_{temp}$, say $X_1,\cdots,X_N$ to test the (true) hypothesis
\[
H_0: \mu_0 = \mu_{temp}.
\]
The test will be based on FPCA with $k$
 components. More precisely, we will calculate in each run, for the sample $\{X_1,\cdots X_N\}$ the empirical eigenfunctions $\hat e_1,\cdots, \hat e_k  $ and eigenvalues $\hat \lambda_1,\cdots,\hat \lambda_k$ and therewith compute the statistic
 \[
 \widehat{S}_k:=\sum_{j=1}^k \frac{\Big\langle  \frac{1}{\sqrt{N}}\sum_{i=1}^N[X_i-\mu_0], \hat e_j\Big\rangle^2}{\hat \lambda_j}.
 \]
It is standard in FDA to use as a large-sample approximation for $\widehat{S}_k$ the $\chi^2$-distribution with $k$ degrees of freedom. Thus, we reject if $\widehat{S}_k$ surpasses the corresponding $95\%$ quantile. Empirical type-I-errors based on $1000$ runs are displayed in Table \ref{tab:empirical_levels}. 
\begin{table}[ht]
  \centering
  \caption{Empirical levels in percent for different sample sizes $N$ and dimensions $k$. The nominal level is $\alpha = 5\%$.}
  \label{tab:empirical_levels}
  \renewcommand{\arraystretch}{1.2}
\begin{tabular}{c ccccccc}
  \toprule
  \multirow{2}{*}{$N$} & \multicolumn{7}{c}{$k$} \\ 
  \cmidrule(lr){2-8}
   & 1 & 2 & 3 & 4 & 5 & 6 & 7 \\ 
  \midrule
   50  & 0.6 & 0.0 & 0.0 & 0.0 & 0.0 & 0.0 & 0.0  \\
  100  & 1.9 & 0.4 & 0.6 & 0.2 & 0.2 & 0.1 & 0.0 \\
  150  & 2.4 & 1.6 & 1.0 & 0.7 & 0.7 & 0.4 & 0.4 \\
  200  & 3.4 & 2.0 & 1.6 & 0.8 & 1.4 & 0.6 & 0.9 \\
  \bottomrule
\end{tabular}
\end{table}

Our results clearly show that the empirical nominal level is lower than the targeted $5\%$. The approximation gets better for larger $N$ and worse for larger $k$. Overall, this means that power for the corresponding tests will be reduced but significance levels remain valid. The reason for our outcomes is presumably that according to Theorem \ref{e:thm:main2} empirical eigenvalues overestimate the true eigenvalues. The effect is smaller for larger $N$, and applies more strongly to later eigenvalues (larger $k$). Accordingly, $ \widehat{S}_k$ is stochastically smaller than a $\chi^2$ with $k$ degrees of freedom, particularly if $k$ is larger. While the overestimation of eigenvalues unambiguously should make tests more conservative (and seems to dominate the testing outcomes), it is harder to theoretically assess the effect of unstable eigenfunctions on rejection rates, and more research is needed.

\noindent \textbf{Testing for supercritical components}  Finally, we use the diagnostic tools developed in Section \ref{sec:diag}, to test for the existence of supercritical components at a controlled level. We use the eigengap ratio statistic described there and fix $K_1=3$, which is the maximum number of components under consideration. The motivation here is that from the plots in Figures \ref{fig:1} and \ref{fig:2} it seems that the second and third eigenfunctions have highly unstable estimators and are thus probably (without any pre-smoothing of the data at least) subcritical. Recall that the sample sizes are for the temperature data  $N=235$ and for the hydrological data $N=86$. The level is held constant at $\alpha=0.05$ and we notice that our results are similar for $\alpha=0.1$. The limiting quantile is obtained using a parametric bootstrap with  $1000$ simulation runs. 
\begin{figure}[h]
  \centering
  \begin{subfigure}{0.49\textwidth}
    \centering
    \includegraphics[width=\linewidth]{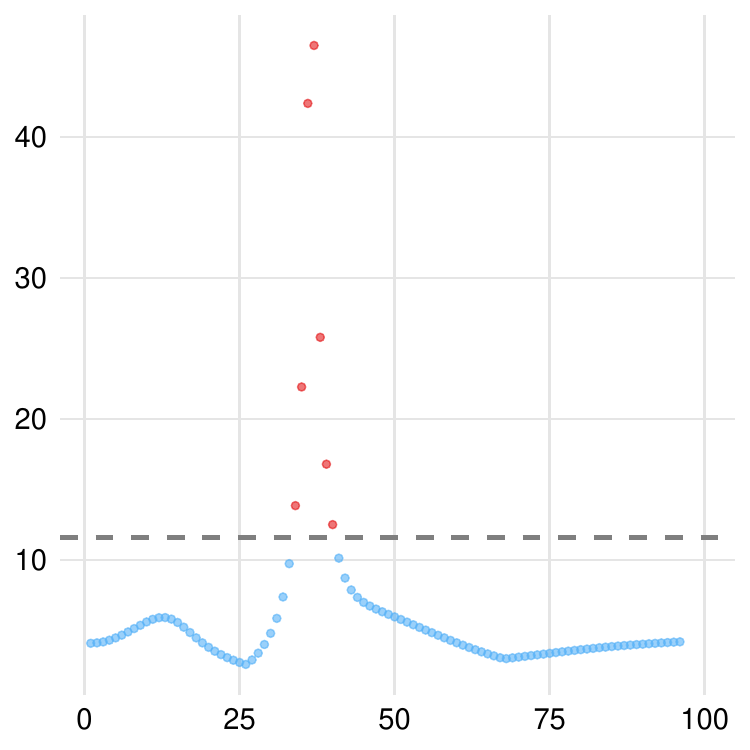}
  \end{subfigure}
  \hfill
  \begin{subfigure}{0.49\textwidth}
    \centering
    \includegraphics[width=\linewidth]{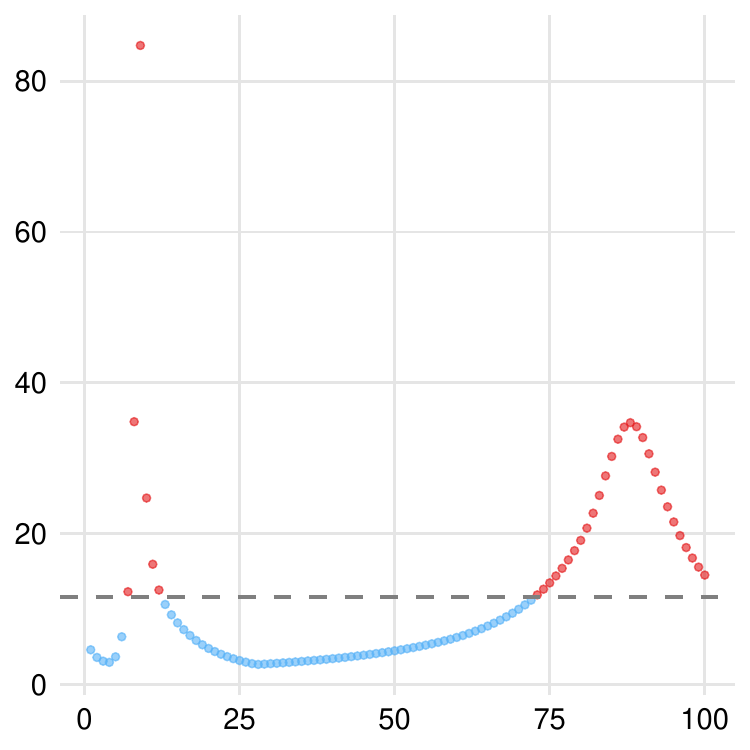}
  \end{subfigure}
  \caption{\label{fig:ondou}  Values of the eigengap ratio statistic $\widehat{\Lambda}^{(t)}(K_1)$ for $K_1=3$, depending on the degree of smoothing $t$, where $t=1$ corresponds to no smoothing. 
  Large values of $\widehat{\Lambda}^{(t)}(K_1)$ provide stronger evidence for the existence of supercritical components. The horizontal line marks the $95\%$ quantile of the limiting distribution and (red) values above that line are significant. Results for temperature data are left and for river discharge data right.}
\end{figure}
Applied to the raw data, we obtain for both datasets that no supercritical components are identified.  This finding corresponds to our observations of instability discussed in the Introduction (see Figures \ref{fig:1} and \ref{fig:2}). 
It is now of natural interest to ask whether pre-smoothing can change this. An expected outcome might be that more pre-smoothing results in data functions that are more closely concentrated around a low-dimensional subspace and thus in higher supercriticality. However, the real picture is more complicated.
For pre-smoothing, we use a sliding average over a total of $t$ days, where the raw data corresponds to $t=1$. We then study the eigengap statistic $\widehat{\Lambda}^{(t)}(K_1)$ (defined in \eqref{e:def:Lh}) depending on $t$. Reasonable amounts of smoothing might be somewhere between averaging a few days up to a few weeks. We have calculated averages of up to $100$ days, which would be excessive in practice, but may help us to get a fuller picture of the behavior of the eigengap statistic. The resulting values of the eigengap statistic relative to the smoothing parameter $t$ are displayed in Figure \ref{fig:ondou}, with the plot for temperature data on the left and for river discharge data on the right. The $95\%$ quantile is displayed as a horizontal line, and to make significantly supercritical results more distinguishable, we have colored them in red. 
The two most important observations are: First, that for $t=1$ (without smoothing) we cannot find any supercritical components in either dataset. Second, the size of the eigengap ratio statistic is highly nonlinear in $t$. This is best visible in the right panel of Figure \ref{fig:ondou}, where a small amount of smoothing $t=2,3,4$ actually decreases the size. Afterwards, the size rises sharply,  giving for $7 \le t \le 11$ significantly supercritical values before a long and extended drop. It seems like there is an ideal amount for pre-smoothing from an FPCA-perspective here ($t=8$) but this is a non-obvious finding. Third, the curves for both datasets look very different -- the ideal amount of smoothing for one dataset is not automatically the ideal amount for another dataset. Overall, we therefore gain the impression that, while smoothing may help to arrive at a more stable FPCA, it is not a universal remedy, and results of a statistical procedure may strongly depend on the level of smoothing.

\section*{Acknowledgments}
The work of Nina Dörnemann was partially supported by Aarhus University Research Foundation (AUFF), project numbers 47221 and 47388.  Tim Kutta’s work was partially
funded by AUFF grants 47331 and 47222. Piotr Kokoszka was partially supported by the United States National Science Foundation grant DMS-2412408. 
\putbib  % uses the default .bib file
\end{bibunit}

%\bibliographystyle{chicago}
%\small\renewcommand{\baselinestretch}{0.95}
%\bibliography{Literature}

\newpage
\appendix
\setcounter{page}{1}
\begin{bibunit}
\begin{center}
\LARGE \textbf{Supplement}
\end{center}

\section{Proofs and technical details} \label{sec_proofs}

The proofs of our results rely mainly on a projection technique that we will briefly explain. Recall that $e_1, e_2,\ldots$ denote the eigenfunctions of the true covariance. For $\kappa \in \mathbb{N}$, we define 
$\Pi_{\kappa} := \sum_{k=1}^{\kappa} e_k \otimes e_k$, which is the projection on the space spanned by $\{e_1,\ldots,e_{\kappa}\}$.
$\Pi_{\kappa}^c:=\Id-\Pi_\kappa$ is the projection on the orthogonal complement. We will frequently use the following decomposition of the empirical covariance operator
%Nina changed K to $\kappa$, as $K$ is already used
\begin{align} \label{e:sig:dec}
    \widehat{\cc} = A_1(\kappa)+A_2(\kappa)+A_3(\kappa)+A_4(\kappa),
\end{align}
where
\begin{align*}
    A_1(\kappa) = & \frac{1}{N}\sum_{i=1}^N \Pi_{\kappa}[X_i] \otimes \Pi_{\kappa}[X_i],\\
    A_2(\kappa) = &  \frac{1}{N}\sum_{i=1}^N \Pi_{\kappa}^c[X_i] \otimes \Pi_{\kappa}[X_i],\\
    A_3(\kappa) = &  \frac{1}{N}\sum_{i=1}^N \Pi_{\kappa}[X_i] \otimes \Pi_{\kappa}^c[X_i],\\
     A_4(\kappa) = &  \frac{1}{N}\sum_{i=1}^N \Pi_{\kappa}^c[X_i] \otimes \Pi_{\kappa}^c[X_i].
\end{align*}
Throughout the proofs, we will set for some $\gamma>0$ the dimension ${\kappa}=\lfloor \gamma N \rfloor$ and then consider a double asymptotic $\lim_{\gamma \to \infty} \lim_{N \to \infty}$. We will start by analyzing the behavior of $A_1({\kappa})$, which is a finite-dimensional covariance matrix of the projected data $\Pi_{\kappa}[X_i].$ 

\subsection{Analysis of $A_1$} 

To analyze $A_1(\lfloor \gamma N \rfloor)$, we can use results from RMT for covariance matrices of data with sample size $N$ and of dimension $\lfloor \gamma N \rfloor$. For background results on random matrices, we refer to Section \ref{sec_background_rmt} below. For ease of reference, we define the $k$th largest eigenvalue $\hat \lambda_k(\gamma)$ of $A_1(\lfloor \gamma N \rfloor)$. 
To analyze it, we notice that $\hat \lambda_k(\gamma) $ as a real-valued random variable has the same distribution as 
\[
\hat \lambda_k(\gamma) \overset{d}{=}\tilde \lambda_k(\gamma),
\]
where $\tilde \lambda_k(\gamma)$ is the $k$th largest eigenvalue of 
\[
\tilde A_1(\lfloor \gamma N \rfloor) =  \frac{1}{N}\sum_{i=1}^N Z_iZ_i^T.
\]
Here $Z_1,\ldots,Z_N \in \mathbb{R}^{\lfloor \gamma N \rfloor}$ are i.i.d. Gaussian random vectors with covariance matrix 
$$\tilde \cc:= \operatorname{diag}( \lambda_1, \lambda_2 , \ldots ,\lambda_{\lfloor N\gamma \rfloor}),$$
where the eigenvalues $\lambda_i$ are defined in \eqref{e:sigrep}.
To analyze $\tilde A_1(\lfloor \gamma N \rfloor) $, we need to introduce a key object to describe the structure of covariance. 
For $\gamma>0$, we define $\xi(\gamma)$ as the solution in $(0,1/b(0))$ to
    \begin{align} \label{e:xi:def}
    \int_0^\gamma\bigg(\frac{b(x)\xi(\gamma)}{1-b(x)\xi(\gamma)} \bigg)^2 =1.   
    \end{align}
Next, we state some fundamental results about $\xi(\gamma)$ and its limit  $\xi(\infty) \in (0,1/\bb(0))$, which was introduced in \eqref{e:def:xiinf}. Also recall the definition of the function $\psi$ in \eqref{eq_def_psi}.

\begin{lem} \label{lem:xi}
    Suppose that Assumptions \ref{ass:bulk:ev} and \ref{ass:spiked:ev} are satisfied. Then, the following results hold.
    \begin{enumerate}[label=(\alph*)]
    \item 
    \label{lem_xi_part_b}
    For any $\gamma>0$, the solution $\xi(\gamma)$ exists and is unique. As a function $\xi: (0,\infty) \to (0,\infty) $ is monotonically decaying and continuous. 
    \item 
    \label{lem_xi_part_c}
    It holds that 
    $\lim_{\gamma \to \infty}\xi(\gamma) =\xi(\infty)$.
    \item \label{lem:psi} It holds that
$
        \lim_{\gamma \to \infty}\psi \big(1/\xi(\gamma) \big)= \psi \big(1/\xi(\infty) \big).
        $
    \end{enumerate}
\end{lem}

The proof of this lemma mainly consists in the analysis of certain integral convergences. We therefore defer it to Section \ref{sec:ana} below. 

Next, we establish the weak limit of the empirical spectral measure $\mathfrak{s}_N(\gamma)$. More precisely, $\mathfrak{s}_N$ is the empirical spectral measure of the first $\lfloor N\gamma \rfloor$  eigenvalues defined as
\begin{align} \label{eq_def_s_n}
\mathfrak{s}_N(\gamma):= 
\frac{1}{\lfloor \gamma N \rfloor } \sum_{i=1}^{\lfloor N \gamma \rfloor}  \delta_{\lambda_i}
=
\frac{1}{\lfloor \gamma N \rfloor }\lb
\sum_{i=1}^K \delta_{s_i} + 
\sum_{i=K+1}^{\lfloor \gamma N \rfloor} \delta_{b([i-K]/N)} \rb.
\end{align}
\begin{lem}\label{lem:s_N}
    Suppose that Assumptions \ref{ass:bulk:ev} and \ref{ass:spiked:ev} are satisfied. Then, it holds for any fixed $\gamma>0$
    \begin{align*}
        \mathfrak{s}_N(\gamma) \stackrel{d}{\to} \mathfrak{s}(\gamma), \quad N\to\infty,
    \end{align*}
   where $\mathfrak{s}(\gamma)$ is the push-forward measure by $b$ of the uniform distribution on $[0,\gamma].$ In particular, $\mathfrak{s}(\gamma)$ is a non-degenerate probability measure with compact support. 
\end{lem}
Once more, we have deferred  the proof to Section \ref{sec:ana}. We will now study the weak convergence of the leading eigenvalues of $A_1(\lfloor N \gamma\rfloor)$, which as we recall are denoted by $\hat \lambda_k(\gamma)$. Again, recall the definition of $\psi$ in \eqref{eq_def_psi}.

\begin{lem} \label{lem:A1}
  Suppose that Assumptions \ref{ass:data}-\ref{ass:spiked:ev} are satisfied.
\begin{enumerate}[label=(\alph*)]
\item \label{lem_part_a} 
Let $k>M$ be a fixed integer.
Then, it holds that 
 \[
    \lim_{\gamma \to \infty}\lim_{N \to \infty} \mathbb{P}(|\hat \lambda_{k}(\gamma) - \psi\big( 1/\xi(\infty)\big)|>\epsilon)=0.
    \]
    % \item \label{lem_part_b} Suppose that  $ s_{k}\xi(\infty)<1$ holds for some $1 \leq k \leq K.$
    % Then, it holds that
    % \[
    % \lim_{\gamma \to \infty}\lim_{N \to \infty} \mathbb{P}(|\hat \lambda_k(\gamma)  - \psi\big( 1/\xi(\infty)\big)|>\epsilon)=0.
    % \]
\item \label{lem_part_c}
    Let $1 \leq k \leq M.$
    Then, it holds that
        \[
    \lim_{\gamma \to \infty}\lim_{N \to \infty} \mathbb{P}\Big(\Big|\hat \lambda_k(\gamma)  - \psi\big(s_k\big)\Big|>\epsilon\Big)=0.
    \]
    \end{enumerate}
\end{lem}

\begin{proof}

\textbf{Proof of \ref{lem_part_a}.} 
We first verify that 
\begin{align} \label{eq_lambda_k_sub}
    \sup_{N\in\N} \lambda_k \xi(\infty) <1,
\end{align}
where we note that $\lambda_k$ may depend on $N$ if $k>K$; see \eqref{e:sigrep}.
 For this purpose, we distinguish the cases $M+1 \leq k \leq K$ and $k>K.$ In the case $M+1 \leq k \leq K$, we have $\lambda_k = s_k$ and \eqref{eq_lambda_k_sub} follows from assumption \ref{ass:spiked:ev}. If $k>K$ is fixed, then $\lambda_k = b((k-K)/N) \leq b(0) < 1/\xi(\infty)$, which shows \eqref{eq_lambda_k_sub}.  

 Due to $\xi(\infty) = \lim_{\gamma \to \infty}\xi(\gamma)$ (Lemma \ref{lem:xi} \ref{lem_xi_part_c}) and \eqref{eq_lambda_k_sub}, there exists a fixed $\tau\in(0,1)$ such that $\lambda_k(A_1(\lfloor N \gamma\rfloor) \xi(\gamma) < 1 - \tau$ for all large $N, \gamma.$
This implies by Lemma \ref{lem_rmt_eigenvalues} and Lemma \ref{lem:s_N} for any fixed, large enough $\gamma$
 \begin{align*}
   & \lambda_k( A_1(\lfloor \gamma N \rfloor)) \stackrel{d}{=}  \lambda_k( \tilde{A}_1(\lfloor \gamma N \rfloor)) \\
   \conp &
    \frac{1}{\xi(\gamma)}\bigg(1+\gamma \int \frac{x}{1/\xi(\gamma)-x}d\mathfrak{s}(\gamma)(x) \bigg) 
    = \psi\big( 1/\xi(\gamma)\big), \quad N \to\infty.
 \end{align*}
     The right side is deterministic, and since $\psi\big( 1/\xi(\gamma)\big) \to \psi\big( 1/\xi(\infty)\big)$ as $\gamma \to \infty$ (Lemma \ref{lem:xi} \ref{lem:psi}), the claim follows.
\\
    \textbf{Proof of \ref{lem_part_c}.} Since $\xi$ is monotonically decreasing (Lemma \ref{lem:xi} \ref{lem_xi_part_b}), it follows from assumption \ref{ass:spiked:ev} that for any $\gamma>0$ we have   $
    s_k \xi(\gamma)>1.$ This means that  $\lambda_k (\tilde{A}_1(\lfloor N \gamma\rfloor))$ is in the supercritical regime for $\tilde{A}_1(\lfloor N \gamma\rfloor)$. 
    According to Lemma \ref{lem_rmt_eigenvalues} and Lemma \ref{lem:s_N}, it holds for $N \to \infty$ (and $\gamma$ fixed) that 
    \[
    \lambda_k( A_1(\lfloor \gamma N \rfloor)) \stackrel{d}{=}
    \lambda_k(\tilde A_1(\lfloor \gamma N \rfloor))\overset{P}{\to} s_k\bigg(1+\gamma \int \frac{x}{s_k-x}d\mathfrak{s}(\gamma)(x) \bigg) =RS.
    \]
    Plugging in the definition of $\mathfrak{s}(\gamma)$  on right side, $RS$ satisfies as $\gamma \to \infty$
    \[
     RS = s_k\bigg(1+ \int_0^\gamma \frac{b(x)}{s_k-b(x)}dx \bigg)\to  s_k\bigg(1+ \int_0^\infty \frac{b(x)}{s_k-b(x)}dx \bigg) =\psi\big( s_k \big),
    \]
where the last equality follows from the definition of $\psi$ in \eqref{eq_def_psi}.

\end{proof}

\subsection{Analysis of $A_2,A_3,A_4$}

Up to this point, we have focused on the leading term $A_1(\lfloor N\gamma \rfloor)$ in the decomposition \eqref{e:sig:dec}. We now turn to the remainders $A_i(\lfloor N\gamma \rfloor)$ for $i=2,3,4$. These remainders are for any fixed $\gamma>0$ and as $N \to \infty$ non-negligible. However, as we demonstrate in the next lemma, for larger $\gamma$, we can show that these terms become stochastically smaller and in the double limit of first $N \to \infty$ and then $\gamma \to \infty$ they do vanish.

\begin{lem}\label{lem:A234}
Suppose that Assumptions \ref{ass:data}-\ref{ass:spiked:ev} are satisfied.
    For any of the terms $A_{i}, \, i=2,3,4$ and any $\epsilon>0$, it holds that
    \begin{align} \label{e:A234}
    \lim_{\gamma \to \infty}\limsup_{N \to \infty} \mathbb{P}(\|A_i(\lfloor N \gamma\rfloor) \|_\mathcal{L}>\epsilon)=0.  
    \end{align}
\end{lem}

\begin{proof}
The proofs for $A_2, A_3$ are identical and we focus on  $A_2$.
For this purpose, define
\[
\cc_\Pi := \mathbb{E}\Pi_{\lfloor N \gamma\rfloor}[X_i]\otimes \Pi_{\lfloor N \gamma\rfloor}[X_i], \qquad \cc_\Pi^c := \mathbb{E}\Pi^c_{\lfloor N \gamma\rfloor}[X_i]\otimes \Pi^c_{\lfloor N \gamma\rfloor}[X_i].
\]
Moreover, we need the the notion of effective rank that we have introduced previously but recall for convenience: For an operator $A$, it is defined as
    \[
    \mathfrak{r}[A] := \frac{\operatorname{Tr}[A]}{\|A\|_\mathcal{L}}.
    \]
    Finally, we define the constant  
    \begin{align*}
        \mathcal{E}_N(\cc_\Pi, \cc_\Pi^c):= & \bigg( \frac{\mathfrak{r}[\cc_\Pi]+\mathfrak{r}[\cc_\Pi^c]}{N}\bigg)^{1/2} + \frac{1}{N}\bigg(\{\mathfrak{r}[\cc_\Pi]+\log(N)\}\{\mathfrak{r}[\cc_\Pi^c]+\log(N)\}\bigg)^{1/2}.
    \end{align*}
    
    Theorem 2.1 in \cite{chen:sanzalonso:2025}  states that for some absolute constant $c>0$ 
    \begin{align} \label{e:th21}
         \mathbb{E}\|A_2(\lfloor N \gamma\rfloor) -\mathbb{E}A_2(\lfloor N \gamma\rfloor)\|_\mathcal{L} 
        \le  c \|\cc_\Pi\|_\mathcal{L}^{1/2 }\|\cc_\Pi^c\|_\mathcal{L}^{1/2 } \mathcal{E}_N(\cc_\Pi, \cc_\Pi^c).
    \end{align}
    To use this formula, we have to calculate all entities in our particular case. For the operator norms, this gives (for sufficiently large $N$ such that $\lfloor N\gamma \rfloor > K$)
    \begin{align*}
        \|\cc_\Pi\|_\mathcal{L} = s_1, \qquad    \|\cc_\Pi^c\|_\mathcal{L} =  b \lb \frac{\lfloor N\gamma \rfloor - K +1}{N}\rb  \stackrel{N \to \infty}{\longrightarrow} b(\gamma).
    \end{align*}
    Next, we come to the effective ranks, where we need to calculate the traces of $\cc_\Pi, \cc_\Pi^c$. Using integral approximations, we get
    \begin{align*}
        \frac{\operatorname{Tr}[\cc_\Pi]}{N} 
        = \frac{1}{N}\sum_{i=1}^{K} s_i
        + \frac{1}{N}\sum_{i=K+1}^{\lfloor N \gamma \rfloor}  b((i-K)/N) 
        \stackrel{N \to \infty}{\longrightarrow} \int_0^\gamma b(x) dx, \quad  
        \frac{\operatorname{Tr}[\cc_\Pi^c]}{N} \stackrel{N \to \infty}{\longrightarrow}  \int_{\gamma}^\infty b(x) dx.
    \end{align*}
    Combining these limits with the inequality \eqref{e:th21} gives (note that the logarithmic terms are vanishing as $N \to \infty$)
    \begin{align} \label{z1}
      \limsup_{N\to\infty}   \mathbb{E}\|A_2(\lfloor N \gamma\rfloor) -\mathbb{E}A_2(\lfloor N \gamma\rfloor)\|_\mathcal{L} 
      \leq c G(\gamma),
    \end{align}
    where
    \begin{align*}
    G(\gamma)& := \lb s_1 b(\gamma) \rb^{1/2} \bigg[
    \Big( \frac{1}{s_1} \int_0^\gamma b(x) dx + \frac{1}{b(\gamma)} \int_\gamma^\infty b(x) dx \Big)^{1/2} 
    \\ & \quad \quad \quad \quad \quad \quad \quad \quad 
    +  \bigg(\frac{1}{s_1} \int_0^\gamma b(x) dx \bigg)^{1/2}
    \bigg(\frac{1}{b(\gamma)} \int_\gamma^\infty b(x) dx \bigg)^{1/2}\bigg] \\ 
    & = \lb b(\gamma) \int_0^\gamma b(x) dx + s_1 \int_\gamma^\infty b(x) dx \rb ^{1/2} + \lb \int_0^\gamma b(x) dx \int_\gamma^\infty b(x) dx \rb^{1/2}.
     \end{align*}
    Since $b(\gamma)\to 0$ and $\int_\gamma^\infty b(x) dx \to 0$ as $\gamma \to \infty$, it also follows that $G(\gamma) \to 0$ as $\gamma \to \infty$.
    Now we note that 
    \[
    \mathbb{E}A_2(\lfloor N \gamma\rfloor) = \mathbb{E}\Pi_{\lfloor N \gamma\rfloor}^c[X_i] \otimes \Pi_{\lfloor N \gamma\rfloor}[X_i] = 0, 
    \]
     since $\Pi_{\lfloor N \gamma\rfloor}[X_i],\Pi_{\lfloor N \gamma\rfloor}^c[X_i]$ are independent.
    Using Markov's inequality and \eqref{z1}, we get for any $\epsilon>0$
    \[
    \limsup_{N\to\infty} \mathbb{P}(\|A_2(\lfloor N \gamma\rfloor) \|_\mathcal{L}>\epsilon) \le \frac{c G(\gamma)}{\epsilon} \stackrel{\gamma \to \infty}{\longrightarrow} 0.
    \]
   This finishes the proof for $A_2$ and thus also $A_3$.
    \\
    Finally, we turn to $ A_4$, which is itself an empirical covariance operator (of the projected data $\Pi_{\lfloor N \gamma \rfloor}^c[X_i]$). The proof can again be conducted using the above result from \cite{chen:sanzalonso:2025}, but it is even simpler to apply Theorem 4 by \cite{koltchinskii:lounici:2017a}, which states that
    \begin{align*}
        \mathbb{E}\|A_4(\lfloor N \gamma\rfloor) -\mathbb{E}A_4(\lfloor N \gamma\rfloor)\|_\mathcal{L} 
        \le  c' \|\cc_\Pi^c\|_\mathcal{L} \max\bigg\{\lb \frac{\mathfrak{r}[\cc_\Pi^c]}{N}\rb^{1/2},\frac{\mathfrak{r}[\cc_\Pi^c]}{N}\bigg\}.
    \end{align*}
    Here, $c'$ is independent of the distribution of the $X_i$.
    Using the above established fact that 
    \[
    \frac{\mathfrak{r}[\cc_\Pi^c]}{N} \stackrel{N \to \infty}{\longrightarrow} \frac{1}{b(\gamma)}\int_{\gamma}^\infty b(x) dx, 
    \]
    we get 
    \begin{align}
    & \limsup_{N\to\infty} \mathbb{E}\|A_4(\lfloor N \gamma\rfloor) -\mathbb{E}A_4(\lfloor N \gamma\rfloor)\|_\mathcal{L} \\
    \le & c' b(\gamma) \max\bigg\{\Big( 
    \frac{1}{b(\gamma)}\int_{\gamma}^\infty b(x) dx\Big)^{1/2},
    \frac{1}{b(\gamma)}
    \int_{\gamma}^\infty b(x) dx \bigg\} \nonumber \\
    & = c'  \max\bigg\{\Big( 
   b(\gamma)\int_{\gamma}^\infty b(x) dx\Big)^{1/2},
    \int_{\gamma}^\infty b(x) dx \bigg\}
    \stackrel{\gamma \to \infty}{\longrightarrow} 0. \label{z2}
    \end{align}
   Finally, we need to analyze $\mathbb{E}A_4(\lfloor N \gamma\rfloor)$ and show that it too converges to $0$ as $\gamma \to \infty$. For this purpose, simply notice that 
    \begin{align} \label{z3}
    \|\mathbb{E}A_4(\lfloor N \gamma\rfloor)\|_\mathcal{L} = 
    b \lb \frac{\lfloor N\gamma \rfloor - K +1}{N}\rb  \stackrel{N \to \infty}{\longrightarrow}  b(\gamma),
    \end{align}
    which converges to $0$ as $\gamma \to \infty.$
     Using Markov's inequality with \eqref{z2} and \eqref{z3} establishes \eqref{e:A234} for $i=4$.
\end{proof}

\subsection{Proofs of Theorems \ref{e:thm:main1} and \ref{e:thm:main2}}

In this section, we provide the proofs of Theorems \ref{e:thm:main1} and \ref{e:thm:main2}.
The assertions \eqref{eq_eigenvals_sub} and \eqref{eq_eigenvals_super} on the eigenvalues are proven in Section \ref{sec_proof_eigenvals}, while the proofs of \eqref{eq_delocal_e1} and \eqref{e:idev} are deferred to Section \ref{sec_proofs_eigenvecs}.

  \subsubsection{Analysis of eigenvalues} \label{sec_proof_eigenvals}
  
  \begin{proof}[Proof of \eqref{eq_eigenvals_sub} and \eqref{eq_eigenvals_super}]
   Let $\gamma>0$ be fixed but arbitrary. Recall the decomposition \eqref{e:sig:dec}, which gives with $\kappa = \lfloor \gamma N \rfloor$
    \begin{align*} %\label{e:sig:dec}
    \widehat{\cc} = A_1(\lfloor \gamma N \rfloor)+A_2(\lfloor \gamma N \rfloor)+A_3(\lfloor \gamma N \rfloor)+A_4(\lfloor \gamma N \rfloor).
\end{align*}
Now, recall our notational convention of calling the eigenvalues of the empirical covariance $\widehat{\cc}$ by $\hat \lambda_i$ and the eigenvalues of the covariance of the projected data $ A_1(\lfloor \gamma N \rfloor)$ by $\hat \lambda_i(\gamma)$. 
Using Lemma 2.2 in \cite{HKbook} implies that
\begin{align*}
& |\hat \lambda_k-\hat \lambda_k(\gamma)|  \\ & 
\le \|\widehat{\cc}-A_1(\lfloor \gamma N \rfloor)\|_\mathcal{L} \le \|A_2(\lfloor \gamma N \rfloor)\|_\mathcal{L}
+\|A_3(\lfloor \gamma N \rfloor)\|_\mathcal{L}
+\|A_4(\lfloor \gamma N \rfloor)\|_\mathcal{L}.
\end{align*}
We treat both, the supercritical and the subcritical case together and we call the postulated limit of the largest eigenvalue generically $\psi$.
Accordingly, let $\epsilon>0$ be fixed but arbitrary, then we get 
\begin{align*}
   & \limsup_{N \to \infty}\mathbb{P}(|\hat \lambda_k-\psi|>\epsilon) \le     \limsup_{N \to \infty}\{P_1(N,\gamma) + P_2(N,\gamma)\},
   \end{align*}
   where
   \begin{align*}
 P_1(N,\gamma) &= \mathbb{P}(|\hat \lambda_k(\gamma)-\psi|>\epsilon/2),\\
  P_2(N,\gamma) & =   \mathbb{P}(\|A_2(\lfloor \gamma N \rfloor)\|_\mathcal{L}
+\|A_3(\lfloor \gamma N \rfloor)\|_\mathcal{L}
+\|A_4(\lfloor \gamma N \rfloor)\|_\mathcal{L}>\epsilon/2)\}.
\end{align*}
Now we can also take the limit $\gamma \to \infty$ on the right side. From Lemma \ref{lem:A234}, we know that
\[
\lim_{\gamma \to \infty}\limsup_{N \to \infty}P_2(N,\gamma)=0.
\]
From Lemma \ref{lem:A1} (use part \ref{lem_part_c} for $1\leq k \leq M$ and part \ref{lem_part_a} for $k>M$), we get that
\[
\lim_{\gamma \to \infty}\limsup_{N \to \infty}P_1(N,\gamma)=0.
\]
This completes the proof for the assertion regarding the eigenvalues in \eqref{eq_eigenvals_sub} and \eqref{eq_eigenvals_super}.
\end{proof}

\subsubsection{Analysis of eigenfunctions} \label{sec_proofs_eigenvecs}

For the analysis of the eigenfunctions, we proceed in two steps. First, we study the supercritical case. Then, we turn to the subcritical regime. Recall the notation that $e_i$ denotes the $i$th eigenfunction of the population covariance $\cc$ and $\hat e_i$ of the empirical covariance $\widehat{\cc}$. We also define the notations of $e_i(\gamma)$ as the $i$th eigenfunction of $ \mathbb{E}A_1(\lfloor \gamma N \rfloor)$ and $\hat e_i$ of $ A_1(\lfloor \gamma N \rfloor)$.
We first formulate and prove a lemma on the proximity of $e_i$ and $e_i(\gamma)$, as well as $\hat e_i$ and $\hat e_i(\gamma)$.

\begin{lem} \label{lem_eigenvector_approx}
    Suppose that assumptions \ref{ass:data}-\ref{ass:spiked:ev} are satisfied and let $1 \leq k \leq M.$
\begin{enumerate}[label=(\alph*)]
    \item \label{lem_eigenvector_a}
    It holds that
\[
\lim_{\gamma \to \infty}\limsup_{N \to \infty}\|e_k-e_k(\gamma)\|=0.
\] 
\item \label{lem_eigenvector_b} 

For any $\epsilon>0$, it holds that
\begin{align*}
    \lim_{\gamma \to \infty}\limsup_{N\to\infty} 
    \mathbb{P} \lb \|\hat e_k-\hat e_k(\gamma) \| > \epsilon \rb =0.
\end{align*}
\end{enumerate}

\end{lem}

\begin{proof} We start with a proof of \ref{lem_eigenvector_a}.
    According to Lemma 2.3 in \cite{HKbook},  it holds that
    \begin{align} \label{eq_ratio}
    \|e_k-e_k(\gamma)\| \le \frac{2 \sqrt{2}\|\cc-\mathbb{E}A_1(\lfloor \gamma N \rfloor)\|_\mathcal{L}}{g_k},
    \end{align}
    where \begin{align*}
        g_k = \begin{cases}
        s_1 - s_2 & \textnormal{if } k=1, \\
        \min \lb s_k - s_{k+1}, s_{k-1} - s_k \rb & \textnormal{if } 1<k<K, \\ 
        \min \lb s_{K-1} - s_K, s_K - b(1/N) \rb & \textnormal{if } k=K.
        \end{cases}
    \end{align*}
In the application of the lemma, we have used the definition of the eigenvalues $\lambda_n$ of $\cc$ in \eqref{e:sigrep}. By the assumptions on $\lambda_n$, we have that $g_k$ are bounded away from $0$. Thus, it remains to upper bound the numerator in \eqref{eq_ratio}.
By the decomposition \eqref{e:sig:dec}, it is clear that
\[
\cc-\mathbb{E}A_1(\lfloor \gamma N \rfloor) = \mathbb{E}A_2(\lfloor \gamma N \rfloor)+\mathbb{E}A_3(\lfloor \gamma N \rfloor)+\mathbb{E}A_4(\lfloor \gamma N \rfloor).
\]
We have seen in the proof of Lemma \ref{lem:A234} that $\mathbb{E}A_2(\lfloor \gamma N \rfloor) = \mathbb{E}A_3(\lfloor \gamma N \rfloor)=0$ and that
\[
\lim_{\gamma \to \infty} \lim_{N \to \infty} \mathbb{E}A_4(\lfloor \gamma N \rfloor)=0.
\]
This concludes the proof of \ref{lem_eigenvector_a} and we turn to proving \ref{lem_eigenvector_b}. Using again Lemma 2.3 in \cite{HKbook},  we have 
    \begin{align} \label{e1}
    \|\hat e_k-\hat e_k(\gamma)\| \le \frac{2 \sqrt{2}\|\widehat{\cc}- A_1(\lfloor \gamma N \rfloor) \|_\mathcal{L}}{\min \lb  \hat \lambda_{k}(\gamma) -  \hat \lambda_{k+1}(\gamma) , \hat \lambda_{k-1}(\gamma) -  \hat \lambda_{k}(\gamma)  \rb  }.
    \end{align}
    Here as before $ \hat \lambda_i(\gamma)$ denotes the $i$th largest eigenvalue of $A_1(\lfloor \gamma N \rfloor) $.
    Using once more the decomposition of $\widehat{\cc}$ in \eqref{e:sig:dec} and Lemma \ref{lem:A234}, we obtain for any $\epsilon >0$ 
    \begin{align} \label{e2}
     \lim_{\gamma \to \infty} \lim_{N\to\infty} \mathbb{P} \lb  \|\widehat{\cc}-\tilde A_1(\lfloor \gamma N \rfloor) \|_\mathcal{L} > \epsilon \rb = 0. 
     \end{align}
    Moreover, as a consequence of Lemma \ref{lem:A1}, we have 
    \begin{align} \label{e3}
        \frac{1}{\min \lb  \hat \lambda_{k}(\gamma) -  \hat \lambda_{k+1}(\gamma) , \hat \lambda_{k-1}(\gamma) -  \hat \lambda_{k}(\gamma)  \rb  } = \mathcal{O}_{\PR}(1). 
    \end{align}
    Combining \eqref{e1}, \eqref{e2} and \eqref{e3} concludes the proof of \ref{lem_eigenvector_b}.
\end{proof}
 
\begin{proof}[Proof of \eqref{eq_delocal_e1}]
    As $b$ is compactly supported, there exists some $\gamma_0 >0$ such that $\hat e_k (\gamma_0) = \hat e_k$ and $e_k(\gamma_0) = e_k$. It follows from Lemma \ref{lem_rmt_eigenvector} \ref{lem_rmt_eigenvector_a} that 
    \begin{align*}
        | \langle \hat e_k, e_k\rangle |^2 
        =  | \langle \hat e_k(\gamma_0), e_k(\gamma_0) \rangle |^2 
        \conp 0, \quad N\to\infty.
    \end{align*}
\end{proof}

\begin{proof}[Proof of \eqref{e:idev}] 
Using the Cauchy-Schwarz inequality and the normalization of the eigenfunctions, we have
\begin{align}
     | \langle \hat e_k , e_k \rangle |
     & = | \langle \hat e_k(\gamma) , e_k(\gamma)  \rangle 
     +  \langle e_k - e_k ( \gamma), \hat e_k \rangle
     + \langle e_k(\gamma) , \hat e_k - \hat e_k(\gamma) \rangle 
     | \nonumber \\ 
     & \leq | \langle \hat e_k(\gamma) , e_k(\gamma)  \rangle |
     + \| e_k - e_k( \gamma ) \| + \| \hat e_k - \hat e_k(\gamma)\| \label{z5}.
\end{align}
Define 
\begin{align*}
    \psi ( \gamma, y) := y \lb 1 + \int_0^\gamma \frac{b(x)}{y - b(x)} dx\rb. 
\end{align*}
Note that, using the weak limit $ \mathfrak{s}(\gamma) $ of $\mathfrak{s}_N(\gamma) $ (see Lemma \ref{lem:s_N}), the integral can be rewritten as 
\begin{align*}
    \psi ( \gamma, y) = y \lb 1 + \gamma\int_0^\gamma \frac{x}{y - x} d\mathfrak{s}(\gamma)(x)\rb,
\end{align*}
which corresponds to the definition of $\Psi_\gamma(y)$ in \eqref{eq_def_Psi}. 
For any $\epsilon>0$ and $\gamma >0$ sufficiently large, we have by Lemma \ref{lem_rmt_eigenvector} \ref{lem_rmt_eigenvector_b} 
\begin{align*}
    \lim_{N\to\infty} \PR \lb \left| | \langle \hat e_k(\gamma) , e_k(\gamma)  \rangle |^2 - \frac{\lambda_k \psi'(\gamma, \lambda_k)}{\psi   (\gamma, \lambda_k)} \right| > \epsilon \rb =0.
\end{align*}
Note that $\psi(\gamma, \lambda_k) \stackrel{\gamma \to\infty}{\to} \psi(\lambda_k)$ and $\psi'(\gamma, \lambda_k) \stackrel{\gamma \to\infty}{\to} \psi'(\lambda_k)$. Thus, we have 
\begin{align*}
    \lim_{\gamma \to\infty} \lim_{N\to\infty} \PR \lb \left| | \langle \hat e_k(\gamma) , e_k(\gamma)  \rangle |^2 - \frac{\lambda_k \psi'(\lambda_k)}{\psi   (\lambda_k)} \right| > \epsilon \rb =0.
\end{align*}
Combining this result with Lemma \ref{lem_eigenvector_approx}, the assertion follows from \eqref{z5}.
\end{proof} 

\subsection{Proofs of Lemma \ref{lem:xi}, Lemma \ref{lem:s_N} and Proposition \ref{prop:onat}} \label{sec:ana}

\begin{proof}[Proof of Lemma \ref{lem:xi}]
$ $\\
    \textbf{Proof of \ref{lem_xi_part_b}.} 
    Proving the existence and uniqueness of $\xi(\gamma)$ is basic and follows as on p. 672 in \cite{elkaroui:2007}.
    Showing that $\xi$ is monotonically decaying is again standard and hence omitted. Finally, we show that $\xi$ is continuous. For this purpose, notice first that for two values $\gamma, \gamma'$ it holds that
 \begin{align*}
     \bigg|\int_0^{\gamma}\bigg(\frac{b(x)\xi(\gamma')}{1-b(x)\xi(\gamma')} \bigg)^2dx-1\bigg| = \int_{\gamma'}^{\gamma}\bigg(\frac{b(x)\xi(\gamma')}{1-b(x)\xi(\gamma')} \bigg)^2dx \le |\gamma-\gamma'| \bigg(\frac{1}{1-c_0} \bigg)^2.
 \end{align*}
 In the second step, we have used that $\xi(\gamma') < 1/b(0)$ and accordingly $b(x)\xi(\gamma') < b(x)/b(0) \le c_0<1 $. 
 Notice that $c_0$ is a number that we can choose to only depend on some fixed $\gamma_0<\gamma,\gamma'$. The reason is that $\xi(\cdot)$ is monotonically decaying. 
 Hence, with a different constant $c_0'$ depending only $c_0$ (and thus some fixed $\gamma_0$) we  get
 \begin{align} \label{e:lip}
\bigg|\int_0^{\gamma}\bigg(\frac{b(x)\xi(\gamma')}{1-b(x)\xi(\gamma')}  \bigg)^2dx-1\bigg|\le c_0'|\gamma-\gamma'|.
  \end{align}
 Now, consider a sequence $(\gamma_n')_n$ converging to $\gamma$ and let $\gamma_0 = \inf_n \gamma_n'>0$. Suppose that $\xi(\gamma_n')$ does not converge to $\xi(\gamma)$ and w.l.o.g. (after extracting a subsequence) suppose that $\xi(\gamma_n') \to \xi(\gamma)+\delta$ for some small $\delta \in \mathbb{R}$. If this were so, it would follow by the dominated convergence theorem that
 \begin{align*}
     \int_0^{\gamma}\bigg(\frac{b(x)\xi(\gamma_n)}{1-b(x)\xi(\gamma_n)} \bigg)^2dx \to \int_0^{\gamma}\bigg(\frac{b(x)[\xi(\gamma)+\delta]}{1-b(x)[\xi(\gamma)+\delta]} \bigg)^2dx.
 \end{align*}
The right side is $\neq 1$ since $\xi(\gamma)$ is the unique solution to \eqref{e:xi:def}. On the other hand from \eqref{e:lip} it follows that the 
integral should be $1+\mathcal{O}(|\gamma-\gamma_n'|)$
i.e. the limit should be $1$. This leads to a contradiction and finishes the proof of \ref{lem_xi_part_b}. \\
   \textbf{Proof of \ref{lem_xi_part_c}.} The proof follows most readily by contradiction. We consider the limit $\gamma \to \infty$. Since the function $\xi(\gamma)$ is bounded and decaying, it also converges to a limit which we may call formally $ \xi^\star(\infty)$. We can make a similar argument as in the previous (essentially using continuity of $\xi$ at $\infty$) to see that actually $\xi^\star(\infty) = \xi(\infty)$. 
   We omit the details to avoid redundancy.\\
\textbf{Proof of \ref{lem:psi}.}  Define 
    \begin{align*}
        \tilde b(x,\gamma) :=\frac{b(x)\xi(\gamma) \mathbb{I}\{x \le \gamma\}}{1-b(x)\xi(\gamma)}, \quad \gamma >0, x\geq 0.
    \end{align*}
    It follows for any $\gamma_0>0$ from part \ref{lem_xi_part_b} ($\xi$ is monotonically decaying) that the family of functions $\{\tilde b(\cdot,\gamma): \gamma \in [\gamma_0, \infty) \}$ is uniformly dominated by the function 
    \[
    x \mapsto  \frac{b(x)\xi(\gamma_0)}{1-b(x)\xi(\gamma_0)}.
    \]
    Note that this function is integrable. Thus, we get by the dominated convergence theorem that 
    \[
   \lim_{\gamma \to \infty}\psi \big(1/\xi(\gamma) \big)= \psi \big(1/\xi(\infty) \big).
    \]
\end{proof}

\begin{proof}[Proof of Lemma \ref{lem:s_N}]
 We consider a continuous, bounded test function $g$ and show that
    \begin{align} \label{e:gconv}
    \int g(x) d\mathfrak{s}_N(\gamma)(x) \to     \int g(x) d\mathfrak{s}(\gamma)(x), \quad N\to\infty.
        \end{align}
    The integral on the left hand side of \eqref{e:gconv} can be written as
\begin{align*}
    \int g(x) d\mathfrak{s}_N(\gamma)(x) & =  
    \frac{1}{\lfloor N \gamma \rfloor}\sum_{i=1}^{K} g(s_k) + 
    \frac{1}{\lfloor N \gamma \rfloor}\sum_{i=K+1}^{\lfloor N \gamma \rfloor} g(b((i-K)/N)) \\
    & = \frac{1}{\lfloor N \gamma \rfloor}\sum_{i=K+1}^{\lfloor N \gamma \rfloor} g(b((i-K)/N)) + o(1)
    .
\end{align*}
The function $g \circ b$ is almost everywhere continuous (because $b$ is decaying and thus almost everywhere continuous) and hence we obtain
\begin{align*}
  \lim_{N\to\infty}  \frac{1}{\lfloor N \gamma \rfloor}\sum_{i=K+1}^{\lfloor N \gamma \rfloor} g(b((i-K)/N)) = \frac{1}{\gamma}
    \int_0^\gamma g(b(x)) dx 
    = \int g(x) d \mathfrak{s}(\gamma)(x)
    .
\end{align*}
 This implies \eqref{e:gconv} and thus, the weak convergence of the measure $\mathfrak{s}_N$. Note also that, by assumption \ref{ass:bulk:ev}, $\mathfrak{s}$ is a non-degenerate probability measure with compact support.   
\end{proof}

\begin{proof}[Proof of Proposition \ref{prop:onat}]
The proof follows readily from Lemma \ref{lem_tw}. 
Indeed, since $b$ is compactly supported, the covariance of the data is $A(\lfloor N \gamma\rfloor)$, where $\gamma$ is some finite number larger than the right endpoint of the support of $b$. This is the covariance of a finite-dimensional object and we may check the assumptions of Lemma \ref{lem_tw} to get the desired convergence result in the subcritical case. 
Assumption \ref{ass_rmt1} is obviously satisfied with $y=\gamma$. Assumption \ref{ass_rmt2} follows due to Gaussianity of the data, and Assumption \ref{ass_rmt3} is a direct consequence of Lemma \ref{lem:s_N}. This shows part (i) of the Lemma. Part (ii) follows from a combination of Theorems \ref{e:thm:main1} and \ref{e:thm:main2}. More precisely, it holds that
\[
0<\frac{\hat \lambda_M-\hat \lambda_{M+1}}{\hat \lambda_{M+2}-\hat \lambda_{M+3}} = \frac{\psi\big(s_M\big)-\psi\big(\xi^{-1}(\infty)\big)+o_P(1)}{\psi\big(\xi^{-1}(\infty)\big)-\psi\big(\xi^{-1}(\infty)\big)+o_P(1)}.
\]
The numerator convergences to a positive number, while the denominator is positive and converges to $0$, implying that the entire ratio goes in probability to $\infty$.
\end{proof}

\section{Some results on spectral statistics} \label{sec:spec} In the main body of this paper, we have considered the behavior of the maximum empirical eigencomponents, which are objects connected to FPCA. Here, we briefly discuss the first-order asymptotics of spectral statistics. Spectral statistics play an important role for hypothesis testing of high-dimensional data and they could be used to study rough functional data as well, once suitable tools are available. \\
\textbf{Spectral statistics and their limit} In RMT, where data are $p$-dimensional, spectral statistics are defined as follows: For a test function $h:[0,\infty) \to \mathbb{R}$, we define
\[
\widetilde{\mathfrak{S}}_{p}(h):= \frac{1}{p}\sum_{i=1}^p h\big( \hat \lambda_i\big) = \int h(x) d\tilde{\mathfrak{s}}_{N}(x), \quad \tilde{\mathfrak{s}}_{N}: =  \frac{1}{p} \sum_{i=1}^p \delta_{\hat \lambda_i}.
\]
A naive definition of the spectral measures in the above sense is of course challenging for $p=\infty$. However, since in random matrix theory $p \sim N$, it makes sense to replace (without changing the statistical content) the normalizing factor to $1/N$.  We only lose the fact that the respective measures add up to $1$. In this sense, an extension to the Hilbertian regime is possible, and we define
\[
\widehat{\mathfrak{S}} (h):= \frac{1}{N}\sum_{i=1}^\infty h\big( \hat \lambda_i\big) = \int h(x) d\hat{\mathfrak{s}}_N (x), \quad \textnormal{where}\quad \hat{\mathfrak{s}}_N  =  \frac{1}{N} \sum_{i=1}^\infty \delta_{\hat \lambda_i}.
\]
Note that $\hat\lambda_i = 0$ for $i > N$, because the empirical covariance has maximally rank $N$.
For the integral to converge it is necessary to suppose that $h(0)=0$, since there exist infinitely many empirical eigenvalues equal to $0$.
Furthermore, we require that $h$ is Lipschitz continuous (this helps to control the behavior in a neighborhood of $0$). In this setting the above integrals are a.s. finite. \\ For the next step, recall our previous notation, where $\mathfrak{s}(\gamma)$ denotes the push forward measure of the uniform distribution by the bulk function $b$ (see Lemma \ref{lem:s_N}). Moreover, consider the definition of the  deformed \MP law, introduced right before Lemma \ref{lem:mp} (near the end of Section
\ref{sec_background_rmt}). The deformed \MP law occurs as a standard limiting object for the bulk in RMT. It is denoted by $F_{y,H}$, where $y$ is a parameter of dimensionality  and $H$ a probability meausre.
The next result is a law of large numbers for the spectral statistics.
\begin{theo} \label{thm:main:3}
Suppose that assumptions \ref{ass:data} and \ref{ass:bulk:ev} are satisfied, $h(0)=0$, and that $h$ is Lipschitz continuous. Then, there exists a number $\mathfrak{S} (h) \in \mathbb{R}$ such that 
  \[
  \widehat{\mathfrak{S}} (h) \overset{P}{\to} \mathfrak{S} (h) , \quad N \to\infty.
  \]
  Moreover, the limit $\mathfrak{S} (h)$ can be appoximated for $\gamma>0$  as follows
  \[
  \mathfrak{S} (h) = \int h(x) d F_{\gamma, \mathfrak{s}(\gamma)}(x)   + o(1),
  \]
  where the $o(1)$ term vanishes as $\gamma\to \infty$. 
\end{theo}
\begin{proof}
The proof is rather short and draws on some previously established results. 
Using Lipschitz continuity of $h$ with some constant, say $L_h$, we get that
\begin{align*}
    \frac{1}{N} \sum_{i=1}^N h(\hat \lambda_i) = \frac{1}{N} \sum_{i=1}^N h(\hat \lambda_i(\gamma)) + Rem(\gamma)
\end{align*}
where
\[
|Rem(\gamma)| \le L_h\|\widehat{\cc}-A_1(\lfloor \gamma N\rfloor )\|_\mathcal{L} \leq L_h\big(\|A_2(\lfloor \gamma N\rfloor )\|_\mathcal{L}+\|A_3(\lfloor \gamma N\rfloor )\|_\mathcal{L}+\|A_4(\lfloor \gamma N\rfloor )\|_\mathcal{L}  \big).
\]
It follows from Lemma \ref{lem:A234} that the right side converges to $0$ when taking the double limit $\lim_{\gamma \to \infty} \lim_{N \to \infty}$. Thus, it is sufficient to study the asymptotic behavior of $\frac{1}{N} \sum_{i=1}^N h(\hat \lambda_i(\gamma)) $.
Using Lemma \ref{lem:mp}, it follows that
\[
\frac{1}{N} \sum_{i=1}^N h(\hat \lambda_i(\gamma)) \to \int h(x) d F_{\gamma, \mathfrak{s}(\gamma)}(x)  =:\ell(\gamma), \quad N\to\infty.
\]
Finally, we need to show convergence of $\ell(\gamma)$. Consider therefore two values $\gamma_1$ and $\gamma_2$ which are larger than some third value $\gamma_0$.
It holds that 
\begin{align*}
    |\ell(\gamma_1)-\ell(\gamma_2)|=&\lim_{N \to \infty}\bigg|\frac{1}{N} \sum_{i=1}^N h(\hat \lambda_i(\gamma_1))- h(\hat \lambda_i(\gamma_2))\bigg| \le L_h \limsup_N \| \widehat{\cc}(\gamma_1)- \widehat{\cc}(\gamma_2)\|_\mathcal{L}\\
    \le & L_h \Big( \limsup_N \| \widehat{\cc}(\gamma_1) - \cc\|_\mathcal{L} +\|\cc-\widehat{\cc}(\gamma_2)\|_\mathcal{L} \Big)
\end{align*}
Arguing as in the previous section, and
letting $\gamma_0 \to \infty$, the right side goes to $0$ (regardless of how $\gamma_1, \gamma_2$ are chosen). This implies that $\ell(\cdot)$ satisfies a Cauchy property and hence converges to some limit in the real numbers that we may denote by $\mathfrak{S} (h)$.
\end{proof}

\section{Background results from random matrix theory} \label{sec_background_rmt}

In section, we consider the $p\times n$ dimensional data matrix $Y$ with covariance matrix $\Gamma_n$ and review the necessary background results from random matrix theory for the spectral analysis of the sample covariance matrix 
\begin{align*}
    \hat \Gamma_n = \frac{1}{n} Y Y^\top .
\end{align*}

 We place emphasis on clarity and a unified framework rather than on achieving the weakest possible assumptions.
The dimension $p=p(n)$ is assumed to be a function of the sample size $n$, and all subsequent limits are taken with respect to $n\to\infty.$

\begin{enumerate}[label=(C-\arabic*)]
  \item \label{ass_rmt1} The ratio $p/n$ satisfies $\lim_{n\to\infty} p/n = y$ for some $y\in (0,\infty).$
   \item \label{ass_rmt2}  The data matrix $ Y\in\R^{p \times n}$ has the form $ Y=  \Gamma_n^{1/2} Z$, where the matrix $ Z$ is the upper-left $n\times p$ block of a doubly-infinite array of i.i.d.~random variables $\{z_{ij}:i\geq 1, j\geq 1\}$  such that $\E ( z_{11})=0, \E ( z_{11}^2) = 1$ and $\sup_{q\geq 1} \E[|z_{11}|^q]<\infty$.
   %The data matrix $ Y\in\R^{p \times n}$ has the form $ Y=  \Gamma_n^{1/2} Z$, where the matrix $ Z$ is the upper-left $n\times p$ block of a doubly-infinite array of i.i.d.~random variables $\{z_{ij}:i\geq 1, j\geq 1\}$  such that $\E ( z_{11})=0, \E ( z_{11}^2) = 1$ and $\sup_{q\geq 1} \E[|x_{11}|^q]<\infty$.
  %$\sup_{q\geq 1} q^{-1/\beta}(\E|x_{11}|^q)^{1/q}<\infty$ for some $\beta>0$.
%
  \item \label{ass_rmt3} As $n\to\infty$, the probability measure $H_n:=(1/p)\sum_{i=1}^p \delta_{\lambda_i(\Gamma_n)}$ has a non-degenerate weak limit $H$ with compact support. 
 \end{enumerate}

Next, we describe the  critical threshold, say $1/\xi$, between the subcritical and supercritical regime for the first bulk component \citep{knowles:yin:2014, leeschnelli2016, ding:yang:2021}. Let $h_r$ denote the rightmost endpoint of the support of $H.$
Then, $\xi$ denotes the unique value in $(0,1/h_r)$ that solves the equation
\begin{align*}
    \int \lb \frac{t \xi}{1 - t\xi} \rb^2 d H(t) = \frac{1}{y}.
\end{align*}
We divide the spectrum of $\Gamma_n$ into subcritical and supercritical eigenvalues, as follows. 
\begin{enumerate}[resume, label=(C-\arabic*)]
\item \label{ass_rmt4_sub_super} 
There exists a fixed $\tau \in (0,1)$ and integer $M\geq 0$ such that $\lambda_M(\Gamma_n) \xi > 1$ and $\lambda_{M+1}(\Gamma_n) \xi < 1-\tau$ for all large $n.$ Moreover, $\lambda_1(\Gamma_n), \ldots \lambda_M(\Gamma_n)$ are fixed with respect to $n$ and satisfy $\lambda_1(\Gamma_n) > \ldots > \lambda_M(\Gamma_n).$  
\end{enumerate}
\begin{rem}
The aforementioned works characterize the separation between the subcritical and supercritical regime through the threshold $1/\xi_n$, defined later in \eqref{eq:def:xin}, which is a finite-sample version of $1/\xi$. Under assumptions \ref{ass_rmt1}-\ref{ass_rmt4_sub_super}, it holds $1/\xi_n \to 1/\xi$ as $n\to\infty$ (see Lemma 5 in \cite{doernemann:lopes:2025}).  Therefore, Assumption 2.6 and Assumption 3.3 of \cite{ding:yang:2021} are satisfied in our setting. 
\end{rem}

We now turn to the asymptotic behavior of the leading sample eigenvalues, which are biased upwards compared to their population counterparts. To this end, we need to define the function
\begin{align} \label{eq_def_Psi}
    \Psi_{y} (\beta) :=  \beta + y \int \frac{t \beta}{\beta-t} dH(t), \quad \beta \notin \operatorname{supp}(H).
\end{align}
with derivative
\[
\Psi_{y}'(\beta) = 1 -y \int \frac{t^2}{(\beta - t)^2} \, dH(t).
\]
This function is important to describe the limit of the leading sample eigenvalues in the sub- and supercritical regime. 
Moreover, some authors use alternatively the derivative $\Psi_{y}'$ to write the supercritical condition for $\lambda_1(\Gamma_n)$ as $\Psi_{y}'(\lambda_1(\Gamma_n))>0$ \citep{zhang2022asymptotic, bai2012sample, li2020asymptotic}.  

\begin{lem}[Inconsistency of eigenvalues] \label{lem_rmt_eigenvalues}
Suppose that assumptions \ref{ass_rmt1}-\ref{ass_rmt4_sub_super} are satisfied. 
    For $1 \leq m_u \leq M$, it holds 
      \begin{align*}
            \lambda_{m_u}(\hat\Gamma_n) - \Psi_{y} (\lambda_{m_u}(\Gamma_n)) \conp 0, \quad  %\lambda_2(\hat\Gamma_n) \conp \Psi_{y} (1/\xi), 
            \quad n\to\infty,
        \end{align*}
        while for any fixed $m_l > M$, it holds 
         \begin{align*}
            \lambda_{m_l}(\hat\Gamma_n) \conp \Psi_{y} (1/\xi), %\quad  \lambda_2(\hat\Gamma_n) \conp \Psi_{y} (1/\xi),
            \quad n\to\infty. 
        \end{align*}
\end{lem}
Lemma \ref{lem_rmt_eigenvalues} follows from Theorem 3.6 and Corollary 3.19 in \cite{ding:yang:2021}.

To specify the fluctuations in the subcritical case, we need to define the location parameter $r_n$, the scale parameter $\sigma_n$ and a finite-sample analogon $\xi_n$ of $\xi$ as follows. $\xi_n$ denotes the unique value in $(0,1/\lambda_1(\Gamma_n))$ that solves the equation
\begin{align} \label{eq:def:xin}
    \int \lb \frac{t \xi_n}{1 - t\xi} \rb^2 d H_n(t) = \frac{n}{p},
\end{align}
and we define
\begin{align*}
        r_n & := \frac{1}{\xi_{n}} \lb 1 + y_n \int \frac{\lambda \xi_{n}}{1 - \lambda \xi_{n}} dH_n(\lambda) \rb,
        \\ 
        \sigma_n^3 & :=  \frac{1}{\xi_{n}^3} \bigg( 1 + y_n \int \Big( \frac{\lambda \xi_{n}}{1 - \lambda\xi_{n}}\Big)^3 dH_n(\lambda) \bigg).
    \end{align*}
Moreover, we introduce the multivariate Tracy-Widom distribution \citep{tracy:widom:1994} as follows. 
Let $\mathbf{G}\in\mathbb{R}^{N\times N}$ be a random matrix with independent standard normal entries, and define
\[
\mathbf{W}=\frac{1}{\sqrt{2}}\bigl(\mathbf{G}+\mathbf{G}^\top\bigr),
\]
which is referred to as a Wigner matrix, or equivalently a GOE$(N)$ matrix.
For any fixed integer $k\geq 1$, it is well known that as $N\to\infty$, the rescaled eigenvalues
\[
N^{2/3}\bigl(\lambda_j(\mathbf{W})-2\bigr)_{1\le j\le k}
\]
converge in distribution to a limiting random vector known as the $k$-dimensional Tracy--Widom distribution.
As a consequence, this formulation shows that the $k$-dimensional Tracy--Widom distribution can be approximated numerically by simulating GOE matrices of sufficiently large dimension. 
    
\begin{lem}
    \label{lem_tw}
    Suppose that assumptions \ref{ass_rmt1}-\ref{ass_rmt3} and assumption \ref{ass_rmt4_sub_super} with $M=0$ are satisfied. 
 Then, for every $k\in\N$, it holds
    \begin{align*} 
       \frac{n^{2/3}}{\sigma_n} \lb \lambda_i (\hat\Gamma_n) - r_n \rb _{1 \leq i \leq k}
       \cond (\zeta_1,\dots,\zeta_k)  , \quad n\to\infty,
    \end{align*}
    where $(\zeta_1,\dots,\zeta_k)$ follows a $k$-dimensional Tracy-Widom distribution.

\end{lem}
Lemma \ref{lem_tw} follows from Proposition 3 in \cite{doernemann:lopes:2025} and Corollary 3.19 in \cite{Knowles2017}. Note that by Lemma \ref{lem_rmt_eigenvalues}, we have $r_n \to \Psi_y (1/\xi)$ for $n\to\infty.$

Next, we quantify the information carried by the sample eigenvectors. To this end, let $\hat v_k$ and $v_k$ denote a unit eigenvector corresponding to the eigenvalue $\lambda_k(\hat \Gamma_n)$ and $\lambda_k(\Gamma_n),$ respectively, for $1\leq k \leq p$. 

\begin{lem} \label{lem_rmt_eigenvector}
Suppose that assumptions \ref{ass_rmt1}-\ref{ass_rmt4_sub_super} are satisfied. 
\begin{enumerate}[label=(\alph*)]
\item \label{lem_rmt_eigenvector_a} 
Let $m_l > M$ be fixed. 
It holds for any fixed (small) $\varepsilon>0$ and unit vector $u \in \R^p$
\begin{align*}
    | \langle u, \hat v_{m_l} \rangle |^2 = o_{P}\lb n^{-1+\varepsilon}  \rb .
\end{align*}
    \item \label{lem_rmt_eigenvector_b} 
    For $1 \leq m_u \leq M$, it holds
      \begin{align*}
        | \langle v_{m_u}, \hat v_{m_u} \rangle |^2  = \frac{\lambda_{m_u}(\Gamma_n) \Psi_{y}'(\lambda_{m_u}(\Gamma_n))}{ \Psi_{y}(\lambda_{m_u}(\Gamma_n))} + o_{P}(1), \quad n\to\infty. 
    \end{align*}
\end{enumerate}
  
\end{lem}
Part \ref{lem_rmt_eigenvector_a} follows from Theorem 3.14 in \cite{ding:yang:2021}, while Part \ref{lem_rmt_eigenvector_b} follows from Example 3.12 in the same work. The result in \ref{lem_rmt_eigenvector_a} is often referred to complete delocalization of $\hat v_{m_l}$, meaning that it spreads its mass roughly equally across all coordinates. In particular, it does not carry information about the principal component $v_{m_u}$.

Next, we review the asymptotic behavior of the empirical spectral measure
\begin{align*}
    F_{\hat\Gamma_n} = \frac{1}{p} \sum_{i=1}^p \delta_{\lambda_i(\hat\Gamma)}.
\end{align*}

The generalized or deformed \MP $F_{y,H}$ law depends on two parameters, the scale $y>0$ and a compactly supported probability measure $H$. $F_{y,H}$ is uniquely characterized through its Stieltjes transform $s_{\gamma,H}$ which is the unique solution to the equation
\begin{align*}
    s_{y,H}(z) = \int \frac{1}{t ( 1 - y - y z s_{y,H}(z)) - z} dH(t), \quad z \in \mathbb{C}^+. 
\end{align*}
In general, there is no closed-form solution, unless in the case $H=\delta_{\sigma^2}$ for some $\sigma^2>0$. 
\begin{rem}
A few comments on the support of the \MP distribution are in place \citep{Knowles2017, yao:zheng:bai:2015, bai:silverstein:2010}, which consists of a finite union of compact intervals. The inconsistency of the sample eigenvalues is reflected in the fact that the support of the \MP distribution is a proper superset of the support of $H.$
Moreover, the limit of the leading sample eigenvalue in the subcritical case equals the rightmost endpoint of the support of the \MP distribution, while, in the supercritical regime, $\lambda_1(\hat\Gamma)$ falls asymptotically outside this support (see Lemma \ref{lem_rmt_eigenvalues}).  
\end{rem}

\begin{lem}[Theorem 2.14 in \cite{yao:zheng:bai:2015}] \label{lem:mp}
Suppose that assumptions \ref{ass_rmt1}-\ref{ass_rmt3} are satisfied. Then, 
\begin{align*}
    F_{\hat\Gamma_n} \cond F_{y,H} \quad \textnormal{ almost surely }, \quad n\to \infty.
\end{align*}
\end{lem}

\section{Additional numerical results}\label{sec_supp_numeric}

We provide some more empirical results for our simulation study and from our data analysis. 

\subsection{Additional simulation results}

We consider the simulation setting described in Section \ref{sec:simulations}. There, we showed the distribution of the largest eigenvalue $\hat \lambda_1$ compared to the theoretical limit in the sub- and supercritical regime, for the bulk function $b_1$. The results were displayed in Figure \ref{Fig:eigenW1} and here, we provide analogous results for $b_2$ and $b_3$ in Figure \ref{Fig:eigenW2}. We see that the approximations are similar to those for $b_1$. Next, we consider results for the bulk functions $b_2,b_3$ and the angle between the leading empirical eigenfunction $\hat e_1$ and the population eigenfunction $e_1$ (see Figure \ref{Fig:eigenV2}). This is an analogue of the setting displayed in Figure \ref{Fig:eigenV1}. Once more, we see that the approximations are very similar to those for $b_1$.

\subsection{Additional details on smoothing}

Temperature curves, as displayed in Figure \ref{fig:1} are often smoothed before use in statistical estimation or inference. As we have mentioned before, smoothing is associated with potential costs and benefits. A cost could be information loss, either about extreme events (e.g., severe heatwaves) or about the speed of change, which is smoothed out.
Another problem might be loss of local information for regression, in particular for function-on-function type. Maybe this latter effect can be countered by smoothing to the same extent regressor and response, but more investigation is needed. Potential advantages of smoothing, in the context of PCA, is certainly that data will concentrate more closely around a low-dimensional subspace, at least if sufficiently much smoothing is applied. This does not, however, mean that the phase transitions described by our theory will immediately vanish for a reasonable amount of smoothing. We illustrate this by briefly revisiting Figure \ref{fig:1} for smoothed data. We consider smoothing by taking a rolling average over one week in Figure \ref{fig:1_7} and over one month  (30 days) in Figure \ref{fig:1_30}. Depending on the context and aim of the analysis, smoothing over a week may be sensible. For most contexts, smoothing over a month may be quite extreme -- we might wonder how much extra information such a curve really entails compared to a vector of monthly averages. Regardless of these considerations, Figures \ref{fig:1_7} and \ref{fig:1_30} are perfectly in line with our theory and show that the effect of biased PCA only recedes very gradually. In the unsmoothed data (Figure \ref{fig:1}), the first eigenfunctions enclosed an angle at $N=110$ of about $40$°. This angle decreases to $30$° for $7$ day smoothing and then to $20$° for 30 day smoothing. So the function is biased but informative and the bias decreases if data are smoothed. In contrast, for the third eigenfunction, we see very little progress, even for $N=110$, going from approximately $75$° to $65$° after smoothing, and we cannot see a great difference between smoothing over a week or over a month. As  $65$° is a fairly large angle, we might say that this function is barely informative at all (it is presumably subcritical).

\begin{figure}[H]
  \centering
  \begin{tabular}{cc}
    % Row 2
    \begin{subfigure}{0.48\textwidth}
      \centering
      \includegraphics[width=\linewidth]{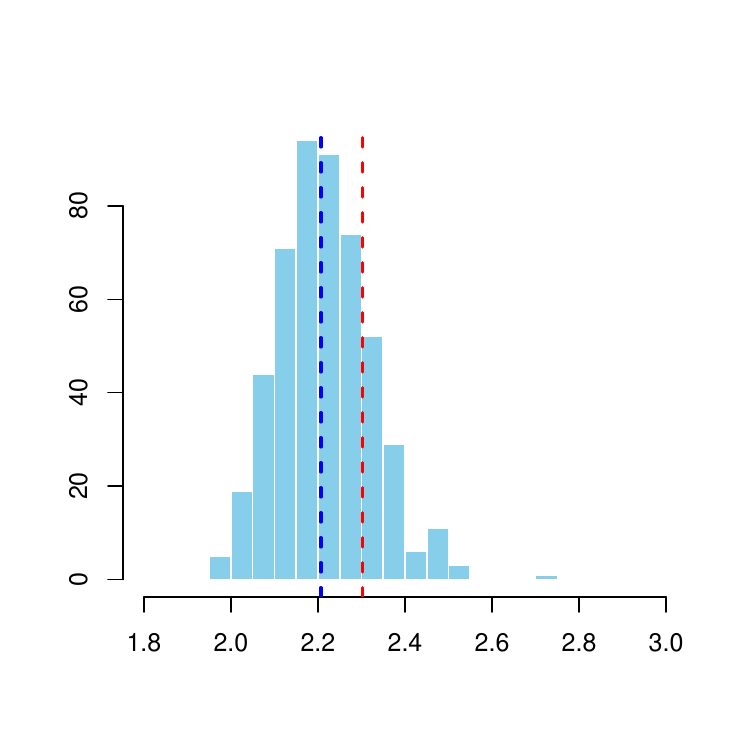}
   
    \end{subfigure} &
    \begin{subfigure}{0.48\textwidth}
      \centering
      \includegraphics[width=\linewidth]{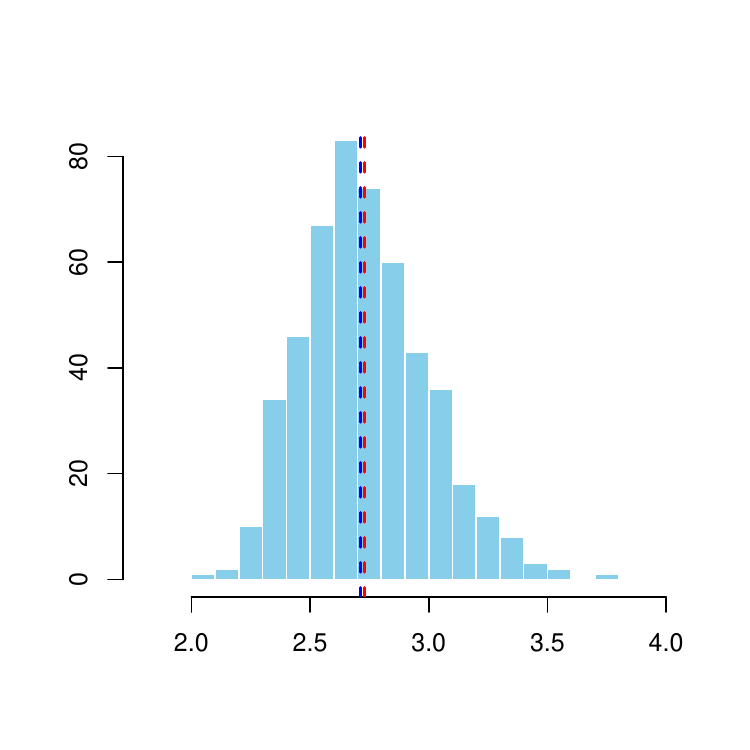}
   
    \end{subfigure}
    \\[0em]
    % Row 3
    \begin{subfigure}{0.48\textwidth}
      \centering
      \includegraphics[width=\linewidth]{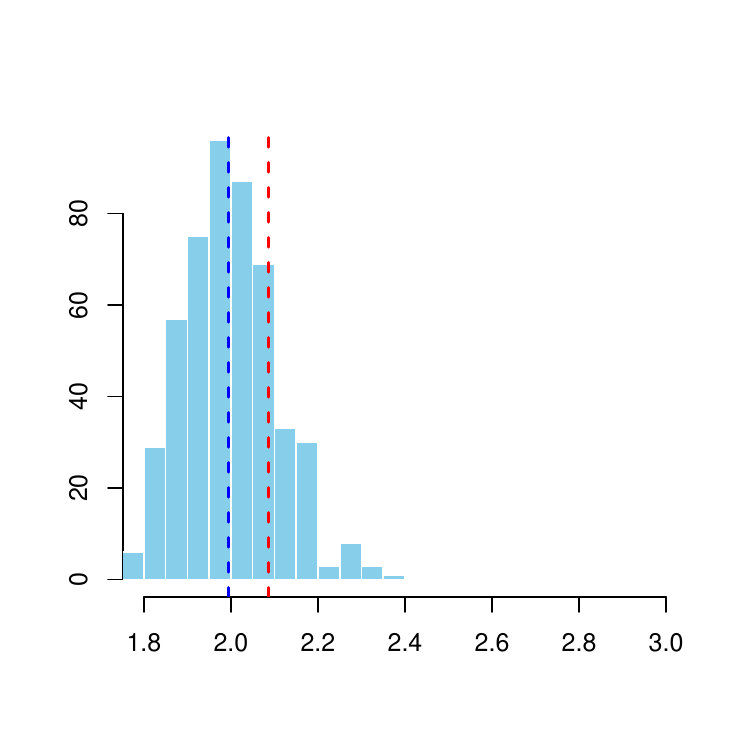}
   
    \end{subfigure} &
    \begin{subfigure}{0.48\textwidth}
      \centering
      \includegraphics[width=\linewidth]{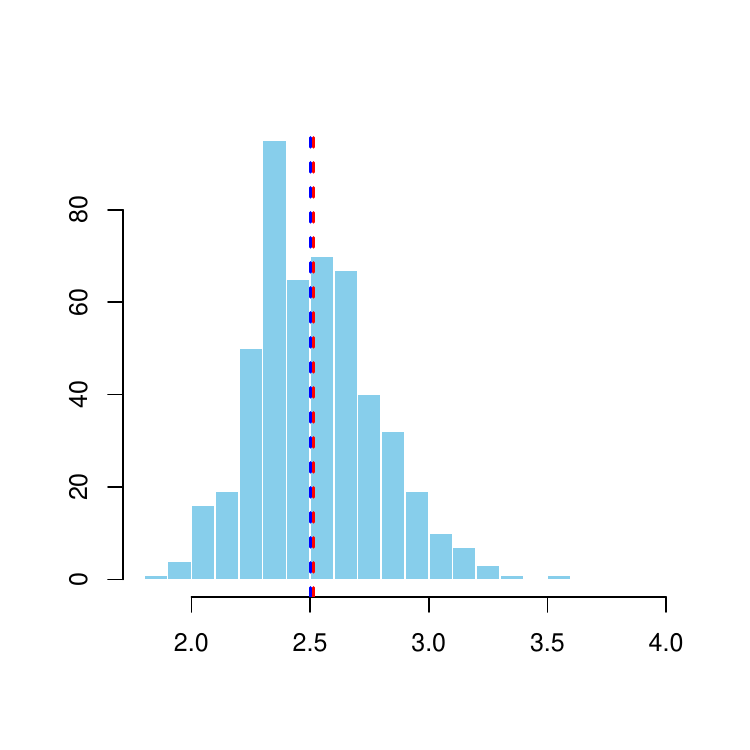}
    \end{subfigure}
  \end{tabular}
  \caption{\label{Fig:eigenW2} Distribution of largest eigenvalue $\hat \lambda_1$ in $500$ simulation runs. Blue vertical line indicates median of the values and red vertical line the theoretical limit. Subcritical cases are left, supercritical cases right. We have considered as bulk function $b_2$ in the first row and $b_3$ in the second row.}
\end{figure}

\begin{figure}[H]
  \centering
  \begin{tabular}{cc}
    % Row 2
    \begin{subfigure}{0.48\textwidth}
      \centering
      \includegraphics[width=\linewidth]{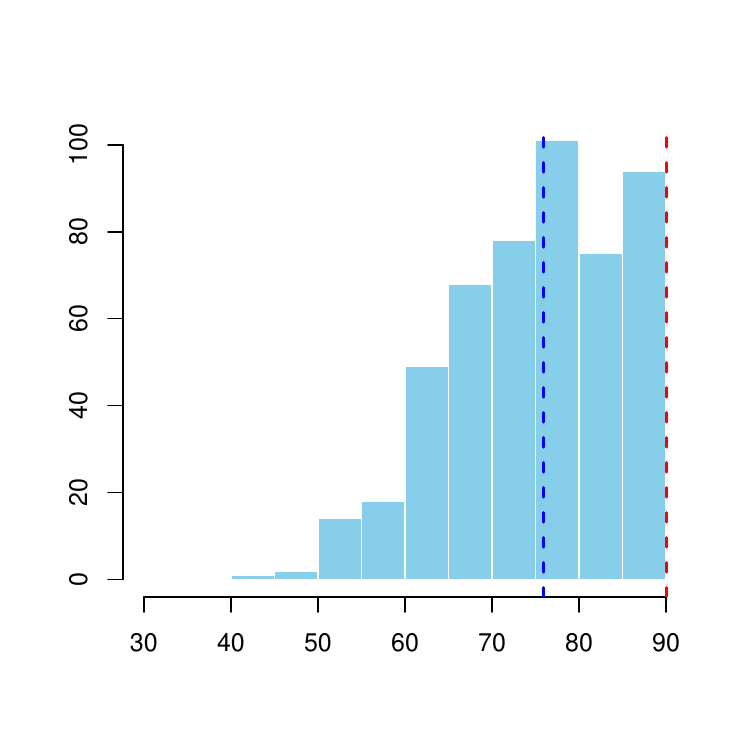}
   
    \end{subfigure} &
    \begin{subfigure}{0.48\textwidth}
      \centering
      \includegraphics[width=\linewidth]{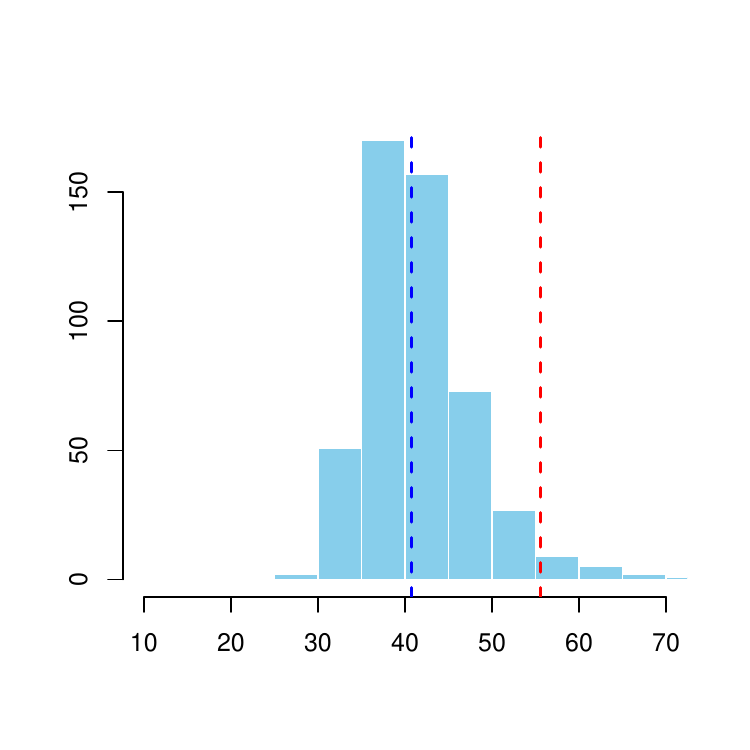}
   
    \end{subfigure}
    \\[0em]
    % Row 3
    \begin{subfigure}{0.48\textwidth}
      \centering
      \includegraphics[width=\linewidth]{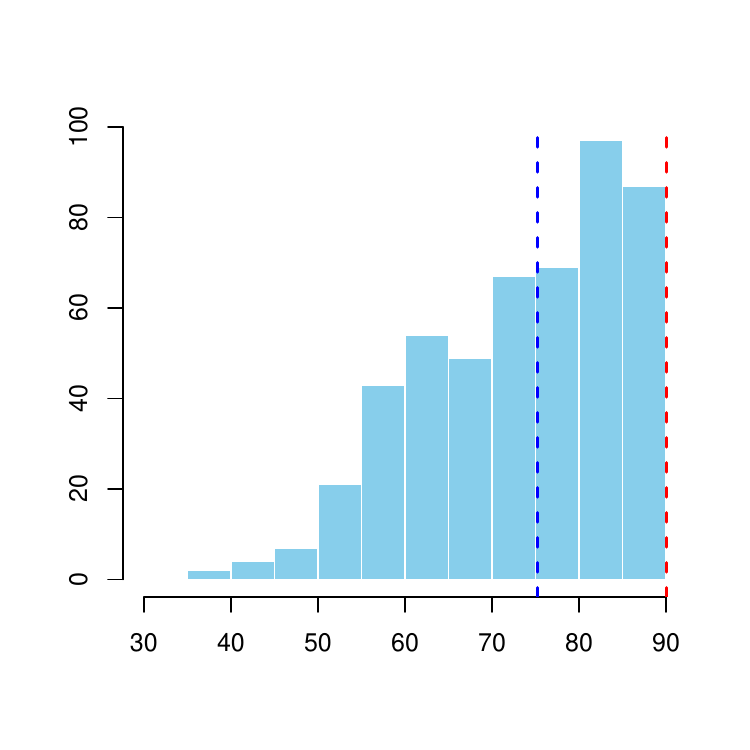}
   
    \end{subfigure} &
    \begin{subfigure}{0.48\textwidth}
      \centering
      \includegraphics[width=\linewidth]{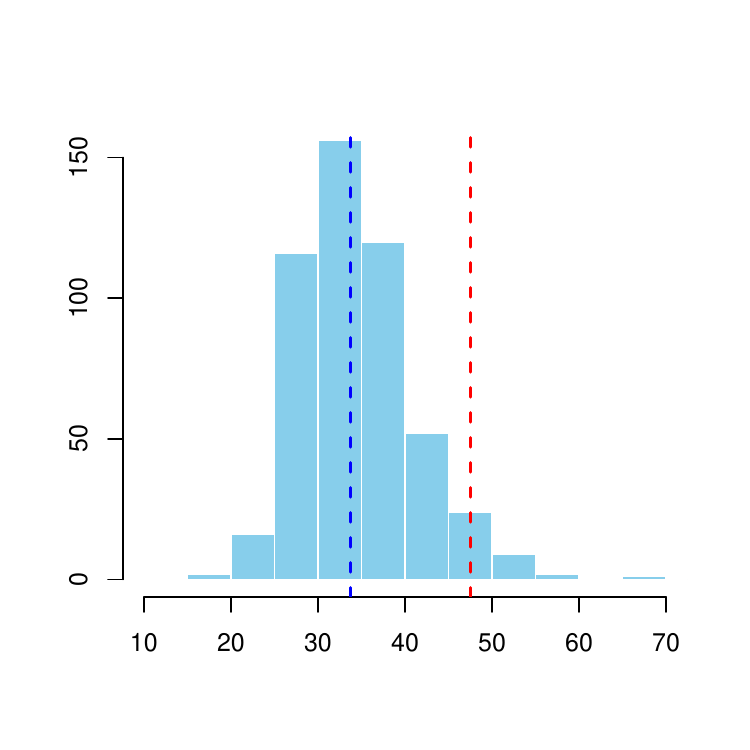}
    \end{subfigure}
  \end{tabular}
  \caption{\label{Fig:eigenV2} Distribution of the angle between $\hat e_1$ and $e_1$ in $500$ simulation runs. Blue vertical line indicates median of the values and red vertical line the theoretical limit. Subcritical cases are left, supercritical cases right. We have considered as bulk function $b_2$ in the first row and $b_3$ in the second row.}
\end{figure}

\newpage

\begin{figure}[H]
  \centering
  \begin{subfigure}{0.49\textwidth}
    \centering
    \includegraphics[width=\linewidth]{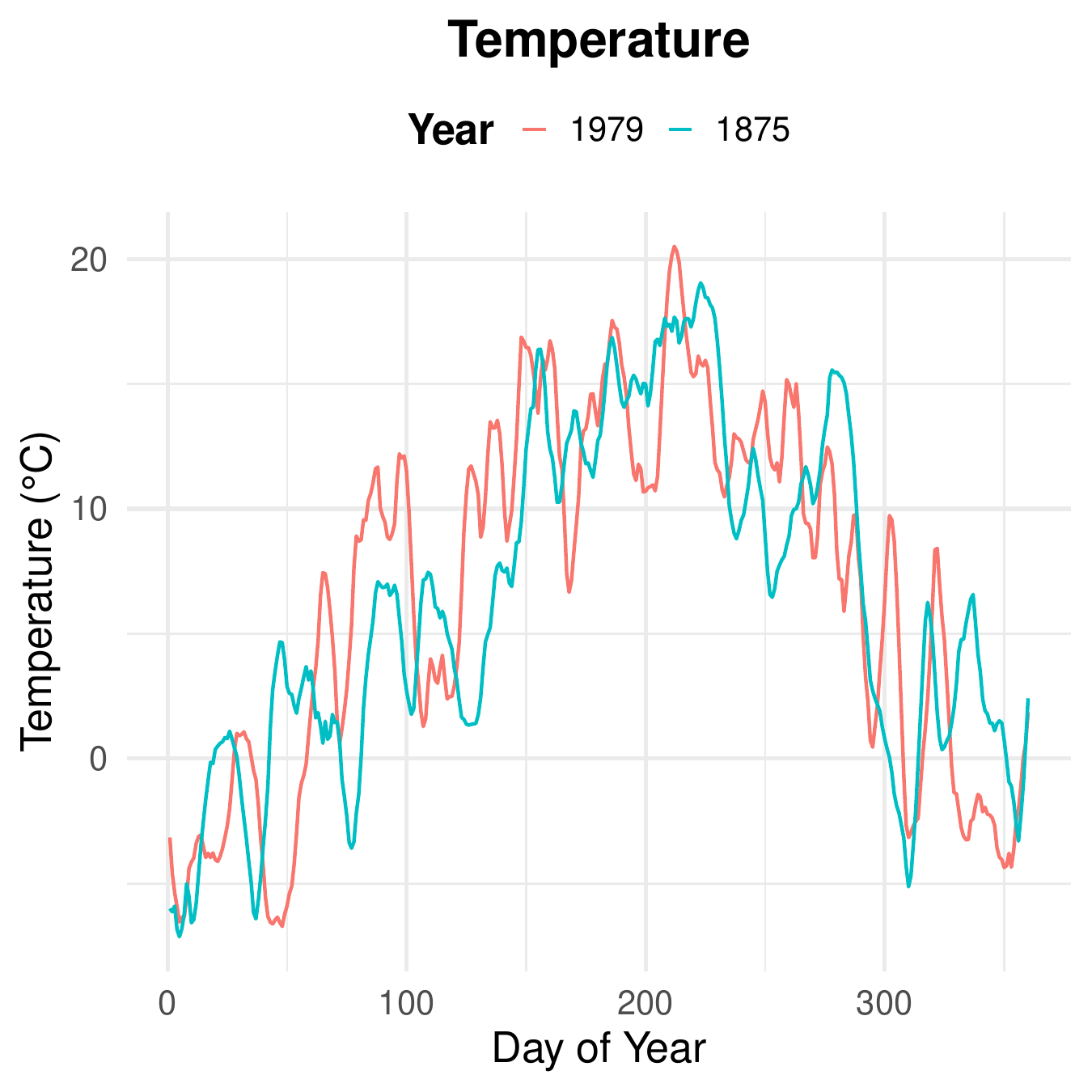} %9x9 pdf
  \end{subfigure}
  \hfill
  \begin{subfigure}{0.49\textwidth}
    \centering
    \includegraphics[width=\linewidth]{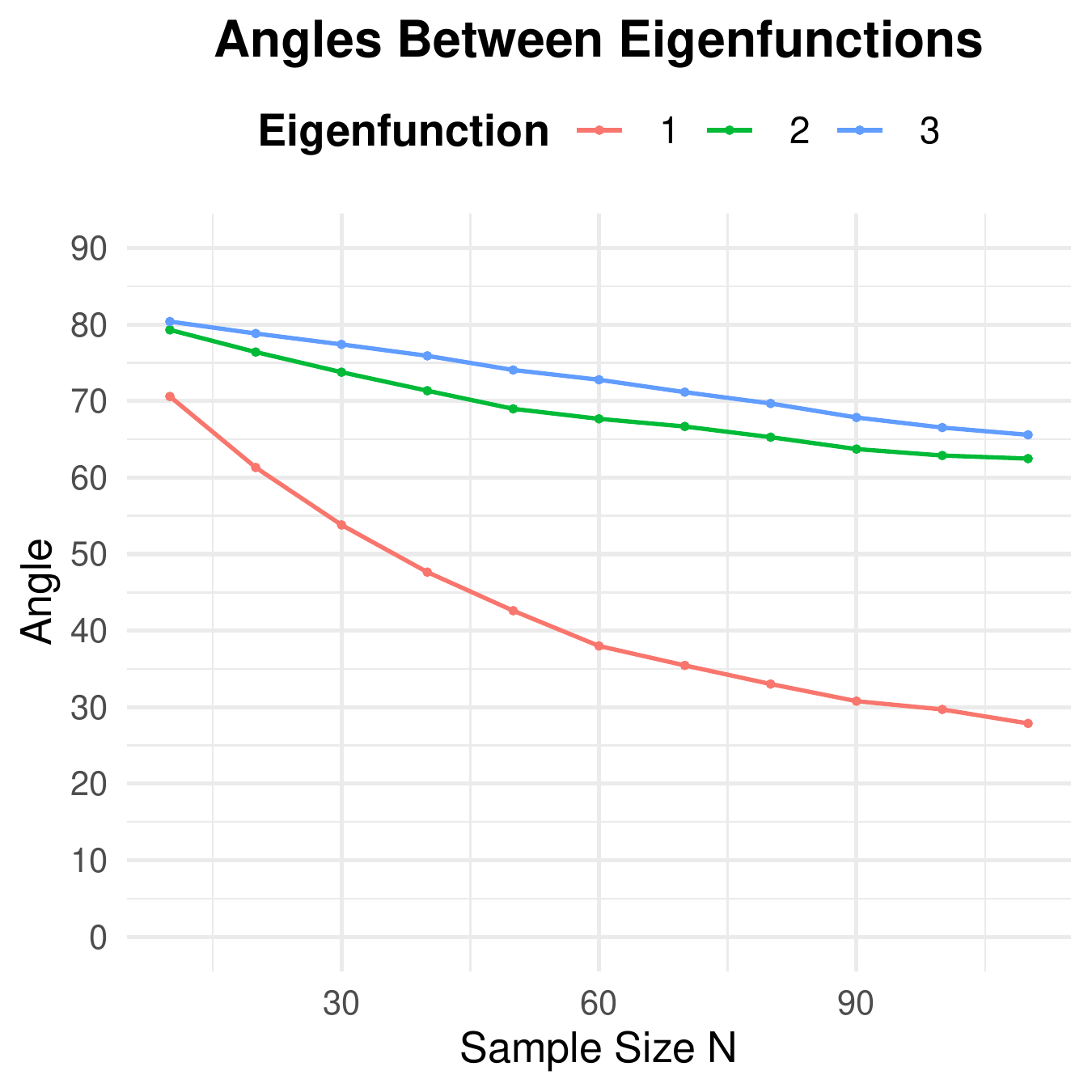}
  \end{subfigure}
  \caption{\label{fig:1_7} Left: Selected temperature profiles from Hohenpeissenberg, smoothed over 7 days. Right:  Average angles between estimated $k$th principal components, $k = 1, 2, 3$, computed from
two random samples of $N$ temperature curves smoothed over 7 days. }
\end{figure}

\begin{figure}[H]
  \centering
  \begin{subfigure}{0.49\textwidth}
    \centering
    \includegraphics[width=\linewidth]{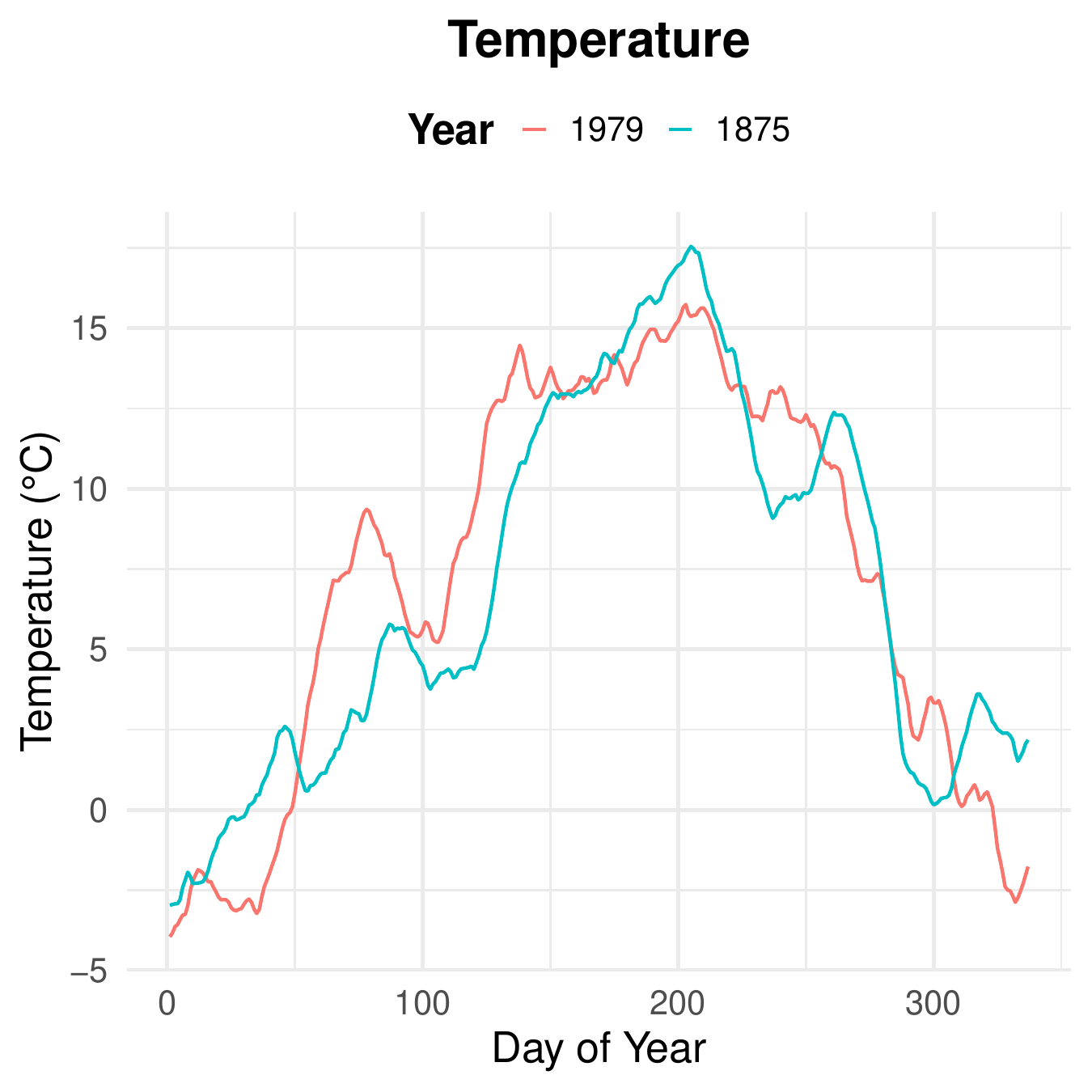} %9x9 pdf
  \end{subfigure}
  \hfill
  \begin{subfigure}{0.49\textwidth}
    \centering
    \includegraphics[width=\linewidth]{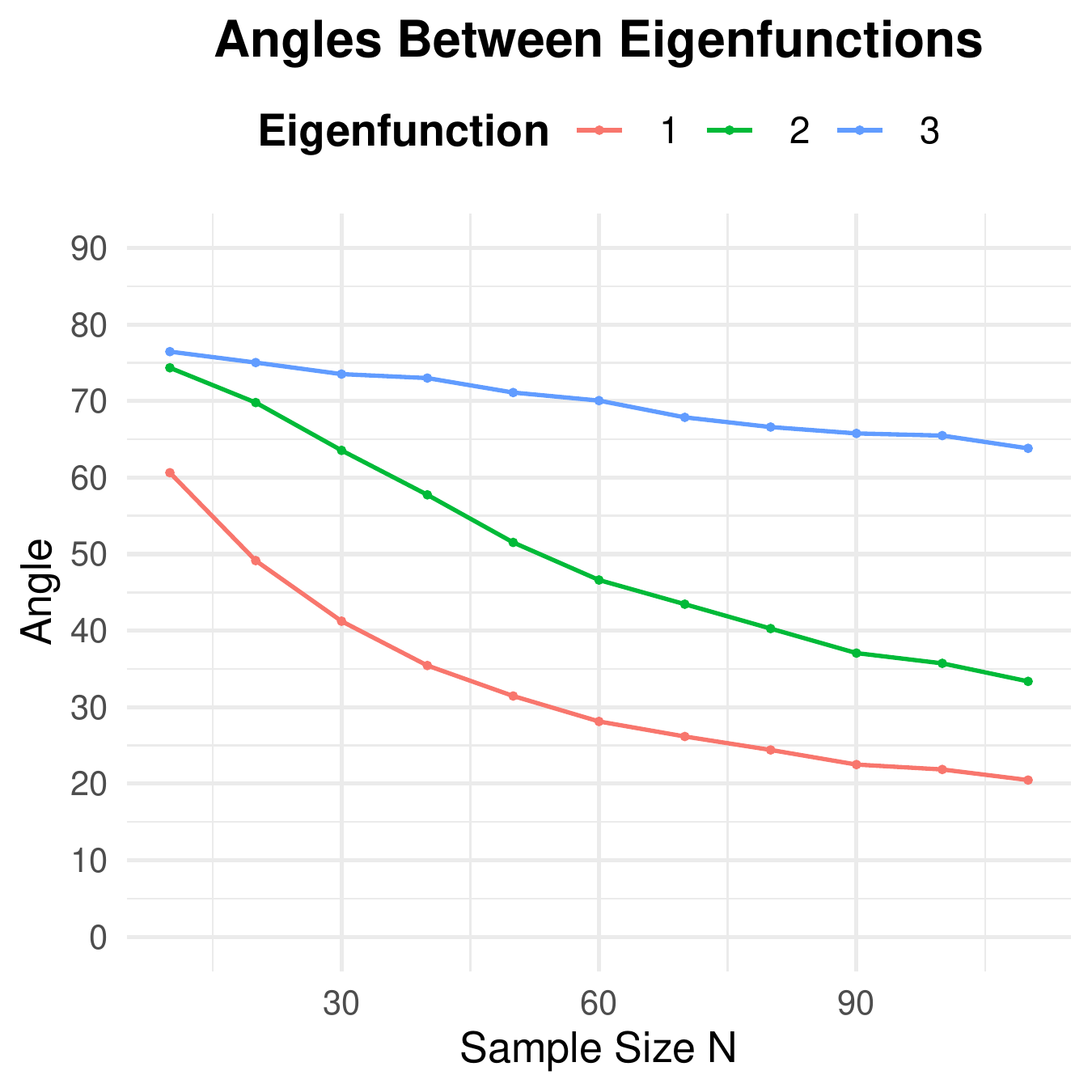}
  \end{subfigure}
  \caption{\label{fig:1_30} Left: Selected temperature profiles from Hohenpeissenberg, smoothed over 30 days. Right:  Average angles between estimated $k$th principal components, $k = 1, 2, 3$, computed from
two random samples of $N$ temperature curves smoothed over 30 days. }
\end{figure}

The second eigenfunction is still barely informative even after smoothing over a week (also about $65$°), but the angle drops below $35$° after smoothing over a month. So, for this eigenfunction, a moderate amount of smoothing was not really helpful, but substantial smoothing lead to it being informative. Since smoothing over a month may often be quite extreme (notice the smoothness of curves in Figure \ref{fig:1_30}, left, which is really far removed from actual temperature profiles), it is hard to say whether the eventual proximity of second eigenfunctions is very meaningful.

\subsection{Additional results for the data analysis}

We present some additional illustrations supplementing the data analysis in the main part of this text. Details are given in Section \ref{sec:data}.
First, we want to provide further evidence of the irregularity of higher order eigenfunctions. For this purpose, we plot the eigenfunctions of order $k=50,100$ of the temperature dataset, as well as the corresponding ACF of the discretized data in Figure \ref{ef:late}. This is an analogue of Figure \ref{fig:stacked_plots}. We observe that the eigenfunctions are extremely irregular and look visually like a stationary and weakly dependent time series. The ACFs strengthen this impression because they are fairly fast decaying. 
Next, we provide an analogue of Figure \ref{fig:EMPEF} in Figure \ref{fig:manyreps}, again for the eigenfunctions of order $k=50$ (left) and $k=100$ (right). We depict estimated eigenfunctions from $50$ random samples of size $N=100$ as light blue lines and the average as a thick blue line. In both cases, a large amount of variability is visible and the average is close to the uninformative line $y=0$. 
Finally, we come to the approximation of the mean parameter by a finite-dimensional projection; this is an analogue of Figure \ref{fig:means} in Section \ref{sec:data}. Here, the mean parameter for the temperature and river discharge data (blue solid line) are projected on the Fourier basis (green dotted line) and the eigenbasis (red dashed line). Approximations by $k=10$ Fourier basis elements were already close to the unprojected parameter, as seen in Figure \ref{fig:means}, while approximations using $k=10$ eigenfunctions were very poor. We want to find out how strong the improvement is when using more eigenfunctions and we use $k=50$ and $k=100$ -- both very large numbers for projections.\footnote{There is a small subtle point here: The river discharge data has sample size $N=86$. Eigenfunctions of order $k>N$ are still produced by our software (\texttt{R}) but are essentially arbitrary directions. Also notice that the projected mean does not have to equal the true mean for $k>N$, of course, because the it may lie orthogonal to (some of) the observed variation.} Our results are displayed in Figure \ref{fig:means:k}. They show that, while there is some improvement for larger $k$, this improvement is fairly slow, especially for the temperature data. Even a value of $k=50$ would usually be deemed large for the purpose of projection, and the parameter projection is a very rough, but not particularly precise approximation.

\begin{figure}[H]
    \centering
    % Upper row
    \begin{subfigure}[b]{0.48\textwidth}
        \centering
        \includegraphics[width=\linewidth]{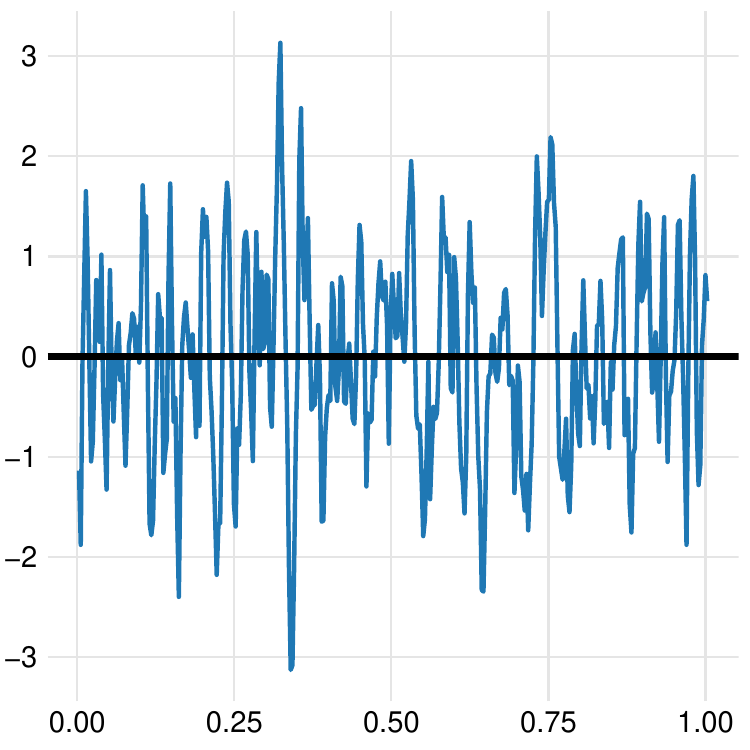}
        \caption{$50$th empirical eigenfunction for temperature data.}
    \end{subfigure}
    \hfill
    \begin{subfigure}[b]{0.48\textwidth}
        \centering
        \includegraphics[width=\linewidth]{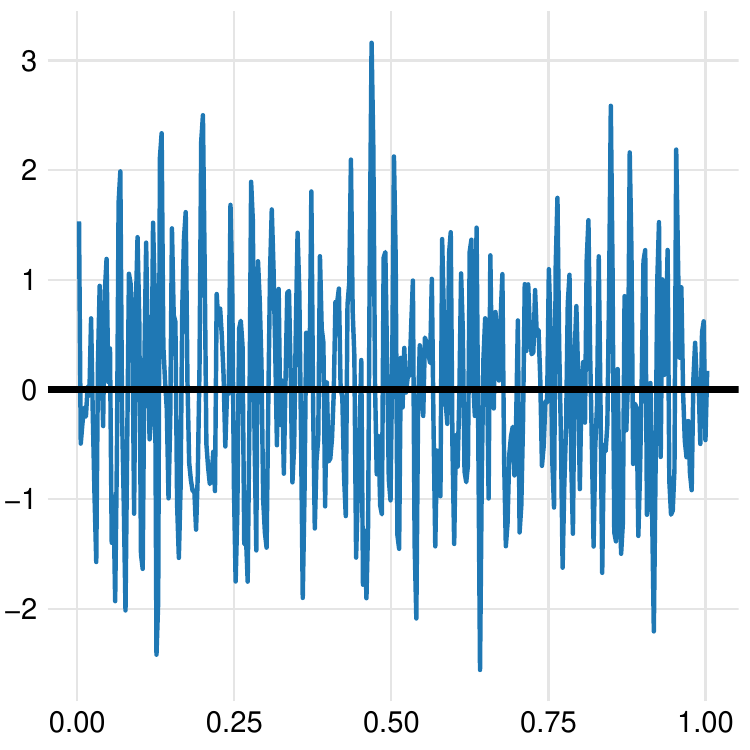}
        \caption{$100$th empirical eigenfunction for temperature data.}
    \end{subfigure}
    
    \vspace{0.5cm} % optional vertical spacing between rows
    
    % Lower row
    \begin{subfigure}[b]{0.48\textwidth}
        \centering
        \includegraphics[width=\linewidth]{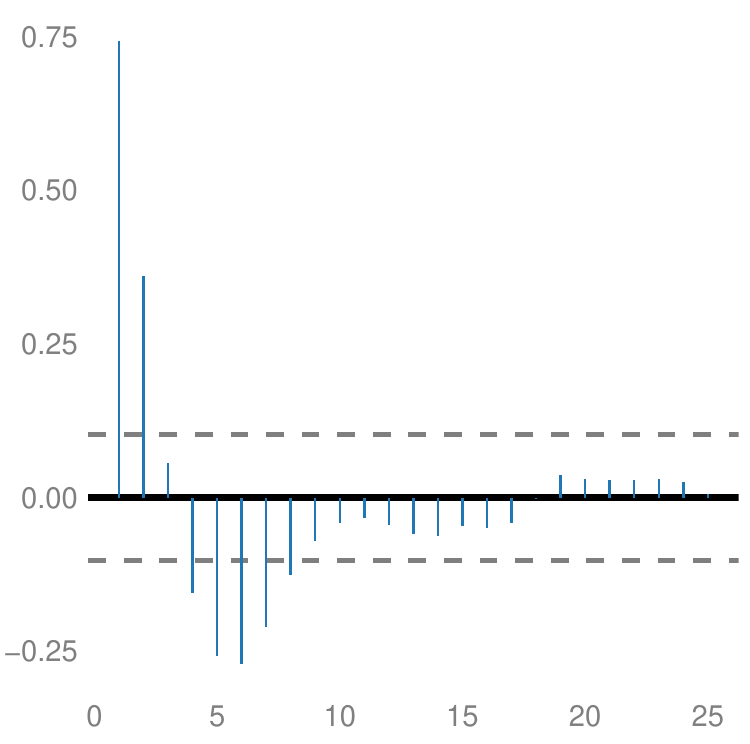}
        \caption{ACF for $50$th eigenfunction.}
    \end{subfigure}
    \hfill
    \begin{subfigure}[b]{0.48\textwidth}
        \centering
        \includegraphics[width=\linewidth]{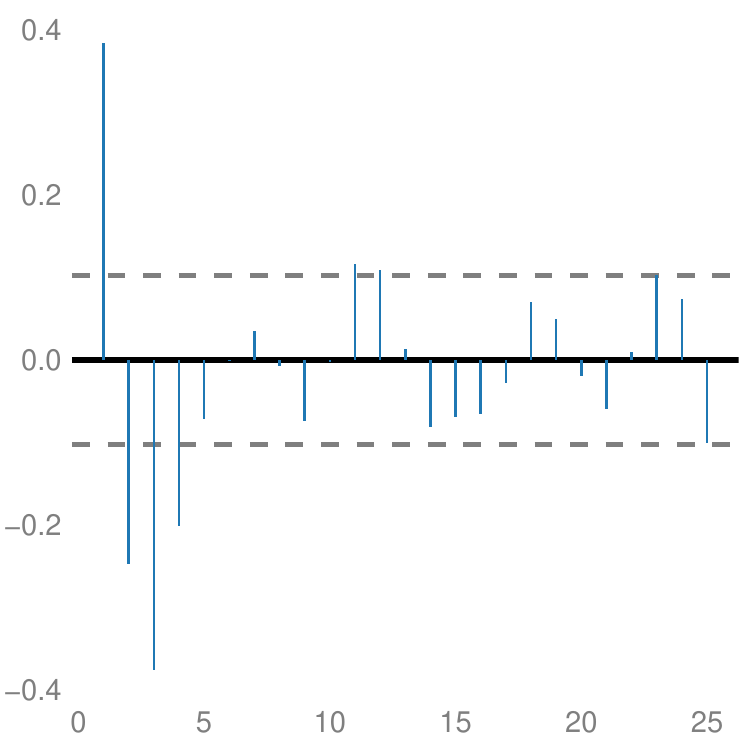}
        \caption{ACF for $100$th eigenfunction.}
    \end{subfigure}
    \caption{Comparison of eigenfunctions of order $k=50$ (left) and $k=100$ (right). ACFs are calculated (lower) for a daily discretization of these functions. Grey horizontal lines indicate the standard $95\%$ confidence interval around a correlation of $0$. \label{ef:late}}

\end{figure}

\begin{figure}[H]
    \centering
    % Upper row
    \begin{subfigure}[b]{0.49\textwidth}
        \centering
        \includegraphics[width=\linewidth]{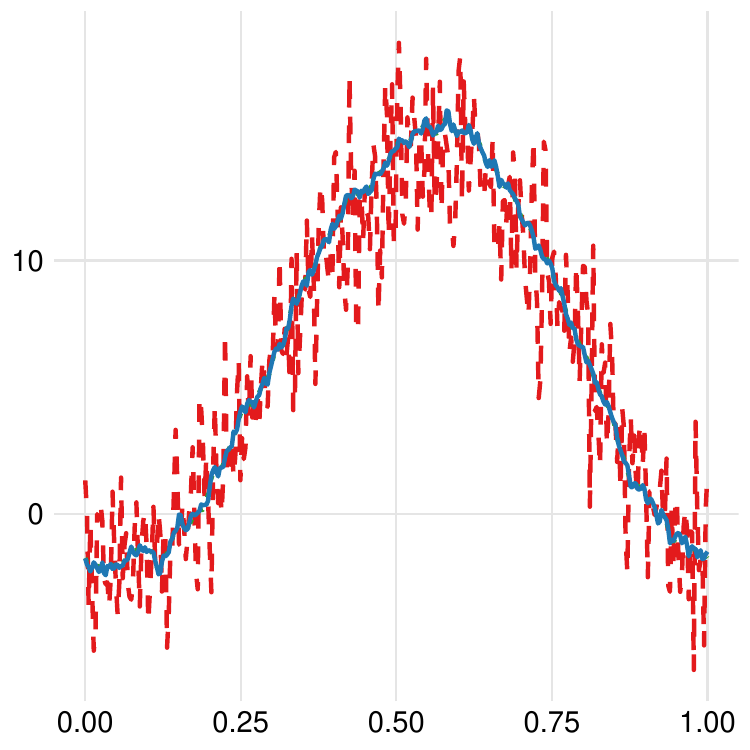}
        \caption{Temperature mean and projections for $k=50$ basis functions.}
    \end{subfigure}
    \hfill
    \begin{subfigure}[b]{0.49\textwidth}
        \centering
        \includegraphics[width=\linewidth]{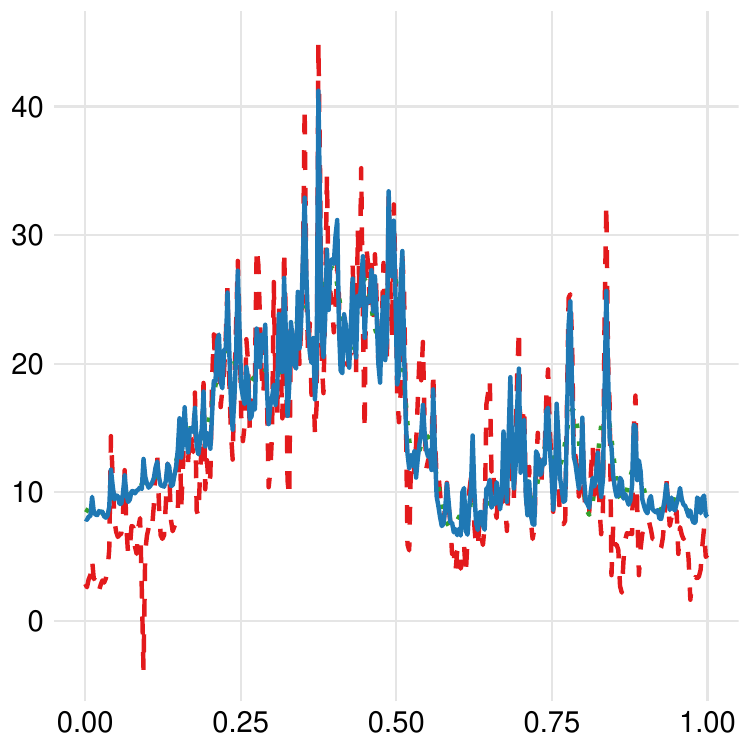}
        \caption{River discharge mean and projections for $k=50$ basis functions.}
    \end{subfigure}
    
    \vspace{0.5cm} % optional vertical spacing between rows
    
    % Lower row
    \begin{subfigure}[b]{0.48\textwidth}
        \centering
        \includegraphics[width=\linewidth]{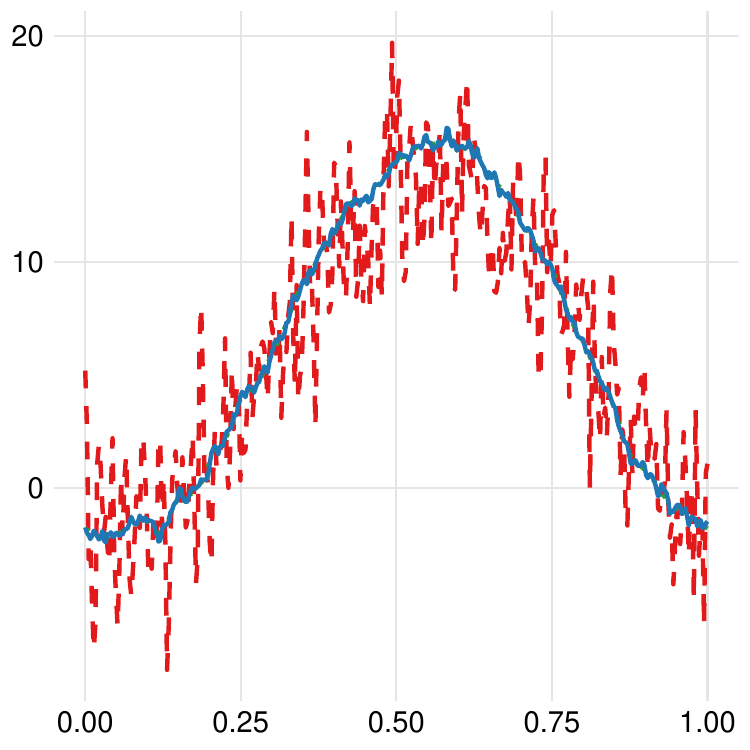}
        \caption{Temperature mean and projections for $k=100$ basis functions.}
    \end{subfigure}
    \hfill
    \begin{subfigure}[b]{0.48\textwidth}
        \centering
        \includegraphics[width=\linewidth]{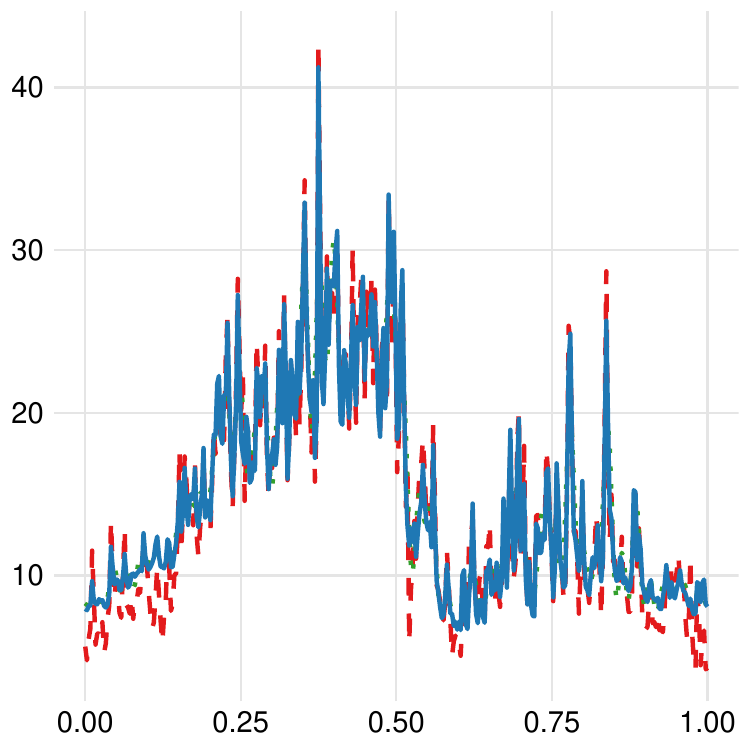}
        \caption{River discharge mean and projections for $k=100$ basis functions.}
    \end{subfigure}
    \caption{Mean function and its projections for temperature data (left) and river discharge data (right), using a basis expansions with $k=50$ (upper) and $k=100$ (lower) components. The unprojected mean is the blue solid line. The projection on the Fourier basis is the dotted green line and the projection on the empirical eigenfunctions is the red dashed line. \label{fig:means:k}}

\end{figure}

\begin{figure}[H]
  \centering
  \begin{subfigure}{0.49\textwidth}
    \centering
    \includegraphics[width=\linewidth]{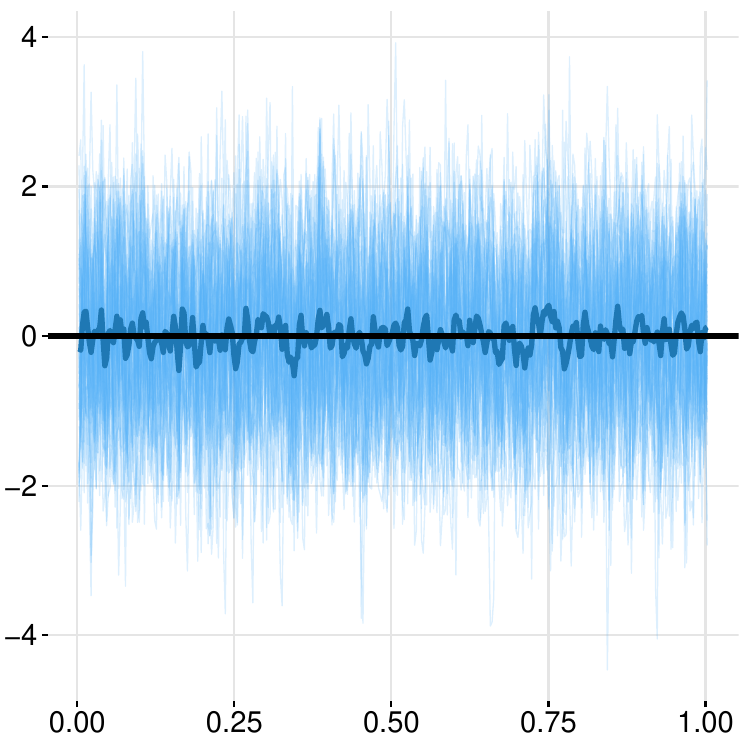}
  \end{subfigure}
  \hfill
  \begin{subfigure}{0.49\textwidth}
    \centering
    \includegraphics[width=\linewidth]{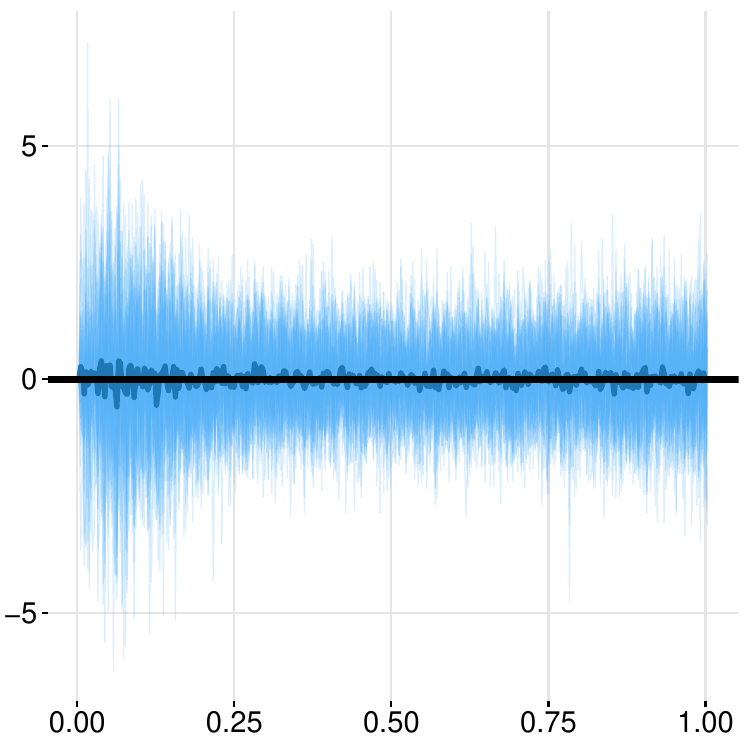}
  \end{subfigure}
  \caption{\label{fig:manyreps}  $50$ realizations of empirical eigenfunctions of order $k=50$ (left) and order $k=100$ (right). Individuals estimates are depicted as light blue lines and the mean as a bold blue line.
  All estimates are based on a sample size of $N=100$.}
\end{figure}

\putbib  % uses the default .bib file
\end{bibunit}
\end{document}